\documentclass[reqno,11pt,a4paper]{amsart}

\usepackage[margin=1in]{geometry}
\usepackage{xspace}
\usepackage[utf8]{inputenc}
\usepackage{graphicx}
\usepackage{amssymb}
\usepackage{mathrsfs}
\usepackage[active]{srcltx}
\usepackage{url}
\usepackage{hyperref}
\usepackage{amsmath,amstext,amsxtra,amsgen,amsbsy,amsopn,amscd,amsthm,amsfonts,stmaryrd}
\usepackage{latexsym}
\usepackage{dsfont}
\usepackage{enumitem}

\usepackage{dsfont}

\usepackage{graphicx}
\usepackage{xcolor}

\theoremstyle{plain}

\newtheorem{prop}{Proposition}
\newtheorem{thm}{Theorem}
\newtheorem{lemma}[prop]{Lemma}
\newtheorem{cor}[prop]{Corollary}
\newtheorem{claim}[prop]{Claim}

\theoremstyle{definition}

\newtheorem{defi}[prop]{Definition}
\newtheorem{assump}[prop]{Assumption}

\theoremstyle{remark}

\newtheorem{rmk}[prop]{Remark}

\numberwithin{equation}{section}
\numberwithin{prop}{section}

\newcommand{\CR}{\mathcal{C}}
\newcommand{\DR}{\mathcal{D}}
\newcommand{\ER}{\mathcal{E}}

\newcommand{\KR}{\mathcal{K}}
\newcommand{\OR}{\mathcal{O}}

\newcommand{\XR}{\mathcal{X}}

\newcommand{\N}{\mathbb{N}}

\newcommand{\R}{\mathbb{R}}
\newcommand{\C}{\mathbb{C}}

\newcommand{\prodscal}[2]{\left\langle#1,#2\right\rangle}
\newcommand{\abs}[1]{\left\lvert #1 \right\rvert}
\newcommand{\norm}[1]{\left\lVert #1 \right\rVert}

\DeclareMathOperator{\range}{\text{range}}
\DeclareMathOperator{\tr}{\operatorname{Tr}}
\DeclareMathOperator{\id}{\operatorname{Id}}
\DeclareMathOperator{\spec}{\operatorname{Spec}}

\DeclareMathOperator{\loc}{loc}

\newcommand{\indicatrice}[1]{{\mathds 1}\left(#1\right)}

\newcommand{\normLp}[3]{\lVert #1 \rVert_{L^{#2}#3}}

\newcommand{\schatten}{\mathfrak{S}}
\newcommand{\normSch}[3]{\left\lVert #1 \right\rVert_{\mathfrak{S}^{#2}#3}}

\newcommand{\test}[1]{\mathcal{C}_c^{\infty} (#1)}
\newcommand{\schwartz}{\mathscr{S}}

\newcommand{\fourierh}{\mathcal{F}_\hbar}

\DeclareMathOperator{\supp}{supp}

\newcommand{\bra}[1]{\left| #1 \right\rangle}
\newcommand{\ket}[1]{\left\langle #1 \right|}

\usepackage{lipsum}

\title{Weyl laws for interacting particles}

\author[N.N. Nguyen]{Ngoc Nhi Nguyen}
\address{Ngoc Nhi Nguyen, Universit\`a degli Studi di Milano, Dipartimento di Matematica, Via Cesare Saldini 50, 20133 Milano, Italy}
\email{ngoc.nguyen@unimi.it}
\date{\today}

\begin{document}
	
	\maketitle
	

	\begin{abstract}
		We study grand-canonical interacting fermionic systems in the mean-field regime, in a trapping potential. 
		We provide the first order term of integrated and pointwise Weyl laws, but in the case with interaction. 
		More precisely, we prove the convergence of the densities of the grand-canonical Hartree-Fock ground state to the Thomas-Fermi ground state in the semiclassical limit $\hbar\to 0$. For the proof, we write the grand-canonical version of the results of [S. Fournais, M. Lewin, J.P. Solovej, \emph{The semi-classical limit of large fermionic systems}, Calculus of Variations and Partial Differential Equations, 57, p.105 (2018)], 
		and [J.G. Conlon, \emph{Semi-classical Limit Theorems for Hartree-Fock Theory}, Communications in Mathematical Physics, 88, p.133 (1983)].
	\end{abstract}
	
\tableofcontents

	\section{Introduction}\label{sec:intro}

For a continuous $V:\R^d\to\R$ which is \emph{confining}, i.e.\ such that $V(x)\to+\infty$ as $\abs{x}\to+\infty$, and an energy $E\in\R$, we denote by $\indicatrice{-\hbar^2\Delta+V\leq E}$ the spectral projector corresponding to eigenvalues less or equal to $E\in\R$ of the Schr\"odinger operators $P_\hbar:=-\hbar^2\Delta+V$  and by $(x,y)\mapsto\indicatrice{-\hbar^2\Delta+V\leq E}(x,y)$ its associated integral kernel. The \emph{Weyl law} provides asymptotics of this integral kernel at the semiclassical limit $\hbar\to 0$. Let us recall two particular forms and their associated first order terms.
\begin{itemize}
	\item The \emph{integrated} Weyl law is the version that provides the asymptotics of the number of eigenvalues of $P_\hbar$ in the interval $(-\infty,E]$. This counting function can be written as
	\begin{equation*}
		N_\hbar(E):=\tr_{L^2(\R^d)}\indicatrice{-\hbar^2\Delta+V\leq E}=\int_{\R^d} \indicatrice{-\hbar^2\Delta+V\leq E}(x,x) dx.
	\end{equation*}
	Then, we have (see for instance \cite{dimassi1999spectral} or \cite[Chap.6]{zworski2012semiclassical} for smooth confining potential $V$, and \cite[Thm. 4.28]{Frank-Laptev-Weidl2023} for continuous compact $V$) that
	\begin{equation}\label{eq:int-without-int_1ordterm}
	\lim_{\hbar\to 0}\hbar^dN_\hbar(E)= \frac{\abs{B_{\R^d}(0,1)}}{(2\pi)^d}\int_{\R^d}[(E-V(x))_+]^{d/2}dx.
	\end{equation}
	\item The \emph{pointwise} Weyl law is the version that provides the pointwise asymptotics of the integral kernel $\indicatrice{-\hbar^2\Delta+V\leq E}(x,y)$. For instance, we have (see \cite[Thm. I.1]{deleporte2021universality} for confining locally smooth potentials $V$) that
	\begin{equation}\label{eq:pw-without-int_1ordterm}
		\lim_{\hbar\to 0}\hbar^d\indicatrice{-\hbar^2\Delta+V\leq E}(x,x)=\frac{\abs{B_{\R^d}(0,1)}}{(2\pi)^d}[(E-V(x))_+]^{d/2}.
	\end{equation}
\end{itemize}

These asymptotics can be interpreted as the spatial equidistribution of free fermions trapped in a potential $V$ and associated to energies less or equal to $E$. Indeed, the fermionic ground state of non-interacting, grand-canonical fermions at zero temperature and energy $E$ in the external potential $V$ is uniquely characterized by its one-body density matrix $\indicatrice{-\hbar^2\Delta+V\leq E}$, and the function $x\mapsto\indicatrice{-\hbar^2\Delta+V\leq E}(x,x)$ then represents the spacial density of such a system.

Having this many-body picture in mind, we are interested in generalizing these Weyl laws to the case where interactions between fermions are introduced. We choose to work in the Hartree-Fock approximation where the state of the system is still characterized by its one-body density matrix $\gamma$, a self-adjoint operator on $L^2(\R^d)$ satisfying $0\le\gamma\le1$. The ground state of the system is then a minimizer of the \emph{$\hbar$-Hartree-Fock energy} functional
\begin{equation}\label{eq-def:HF-energy}
	\begin{split}
	\ER^{\rm HF}_{\hbar,V-E,w}(\gamma) := &\hbar^d\tr_{L^2(\R^d)}((-\hbar^2\Delta+V-E)\gamma) \\&\quad
	+\frac{\hbar^{2d}}2\left(\int_{\R^d}\int_{\R^d} \rho_\gamma(x)\rho_\gamma(y)w(x-y)dx dy
	-\int_{\R^d}\int_{\R^d} \abs{\gamma(x,y)}^2w(x-y)dx dy\right)
	,
	\end{split}
\end{equation}
where $w:\R^d\to\R$ is the pair interaction potential, $\gamma(x,y)$ denotes the integral kernel of $\gamma$ at $(x,y)\in\R^d\times\R^d$ and $\rho_\gamma(x):=\gamma(x,x)$ its associated density. Notice that since we work in the grand-canonical setting, we are interested in global minimizers of this functional (that is, without any constraint on $\tr_{L^2(\R^d)}(\gamma)$). Notice also that in the case without interactions ($w=0$), $\gamma=\indicatrice{-\hbar^2\Delta+V\le E}$ is indeed a minimizer.
	The Hartree-Fock functional is scaled with respect to the semiclassical parameter $\hbar$ in such a manner that the kinetic term and the potential terms are of the same order. To achieve this, one could consider trial operators $\gamma_\hbar$ with a semiclassical structure, for instance the Weyl quantization $\gamma_\hbar={\rm Op}^{\rm w}_\hbar(a)$ of phase-space function $a:\R^d\times\R^d\to\R$ with a sufficient decay at infinity \cite{zworski2012semiclassical} (heuristically the Wigner transform \cite{wigner1932quantum} of $\gamma_\hbar$, up to some error as $\hbar\to 0$). Its integral kernel $\gamma_\hbar(x,y)$ and its density $\rho_{\gamma_\hbar}$ are of order $\hbar^{-d}$. The linear term $\tr((-\hbar^2\Delta +V)\gamma_\hbar)$ is of order $\hbar^{-d}$. The interaction term is quadratic in $\rho_{\gamma_\hbar}$ and $\gamma_\hbar(x,y)$, therefore is of order $\hbar^{-2d}$. That is the reason why we retrieve an additional factor $\hbar$ in front of the direct and exchange terms. A further choice has been made in order to have $\ER^{\rm HF}_{\hbar,V,w}(\gamma_\hbar)$ of order $\OR(1)$ for (approximate) minimizers $\gamma_\hbar$, which are expected to have almost a semiclassical structure and to be of order $\OR(\hbar^{-d})$.

As is well-known since the works of Lieb-Simon \cite{lieb1973thomas,lieb1977thomas} in the case of Coulomb systems, minimizers of the Hartree-Fock functional should be related as $\hbar\to0$ to minimizers of the \emph{Thomas-Fermi functional}
\begin{equation}\label{eq-def:TF-energy}
	\begin{split}
	\ER^{\rm TF}_{V-E,w}(\rho) 
	&= \frac{d}{d+2}c_{\rm TF}\int_{\R^d}\rho(x)^{1+2/d}dx +\int_{\R^d}(V(x)-E)\rho(x)dx 
	\\&\quad
	+\frac 12\iint_{\R^d\times\R^d}w(x-y)\rho(x)\rho(y)dx dy
	,
	\end{split}
\end{equation}
on a set of non-negative densities $\rho:\R^d\to\R$ with additional good properties. Here, 
\begin{equation*}
	c_{\rm TF}:=\frac{4\pi^2}{\abs{B_{\R^d}(0,1)}^{2/d}}.
\end{equation*}
Minimizers of $\ER^{\rm TF}_{V-E,w}$ must satisfy, when they exist, the Thomas-Fermi integral equation
\begin{equation*}
	c_{\rm TF}\rho_{\rm TF}(x)^{2/d}= (E-V(x)-w\ast\rho_{\rm TF}(x))_+,
\end{equation*}
i.e.
\begin{equation}\label{eq:TF}
	\rho_{\rm TF}(x) = \frac{\abs{B_{\R^d}(0,1)}}{(2\pi)^d}(E-V(x)-w\ast\rho_{\rm TF}(x))_+^{d/2}.
\end{equation}
	The assumption on the coercivity of the potential $V$ implies that the density $\rho_{\rm TF}$ has compact support.
In the case $w=0$, one recovers the limiting density in the Weyl law \eqref{eq:pw-without-int_1ordterm} and we will see that $\rho_{\rm TF}$ is its natural generalization in the case with interactions. This formula also reflects the typical fact that mean-field particles behave as free particles in the mean-field effective potential $V+w*\rho_{\rm TF}$.

Our main result (see Theorems \ref{thm:int-WL_HF} and \ref{thm:pointwise-WL_HF} below for a more precise formulation) is then
\begin{thm}\label{thm:main-result}
	Let $d\ge1$. Under suitable assumptions on $V$ and $w$ (in particular, assuming that $w$ is repulsive), any minimizer $\gamma_\hbar$ of $\ER_{\hbar,V-E,w}^{\rm HF}$ satisfies (up to a subsequence) the integrated Weyl law
	\begin{equation*}
		 \lim_{\hbar\to 0}\hbar^d\tr_{L^2(\R^d)}(\gamma_\hbar)= \int_{\R^d}\rho_{\rm TF}(x) dx
		,
	\end{equation*}
	as well as the pointwise Weyl law
	\begin{equation*}
		\lim_{\hbar\to 0}\hbar^d\rho_{\gamma_\hbar}(x) =\rho_{\rm TF}(x) 
	\end{equation*}
	for all $x\in\R^d$, where $\rho_{\rm TF}$ is a minimizer of $\ER^{\rm TF}_{V-E,w}$.
\end{thm}

In the case where the minimizer of $\ER_{V-E,w}^{\rm TF}$ is unique (which is for instance the case when $\hat{w}\ge0$, that implies the convexity of the functional), the subsequence can be dropped and the convergence holds for the whole limit $\hbar\to0$.

To prove this result, we follow the approach of Fournais-Lewin-Solovej \cite{fournais2018semi} who treated the canonical case (that is, the asymptotics as $N\to\infty$ of a system of $N$ particles). The grand canonical case is similar, with the following differences: while building test functions is simpler due to the absence of a trace constraint (in	the canonical setting, one assumes that $\tr\gamma=N$), one has to use the confinement of the potential together with the repulsiveness of the interactions to infer that for minimizers, the trace has the right behaviour in $\hbar$ (that is, $\tr\gamma_\hbar	\lesssim \hbar^{-d}$). Let us notice that \cite{fournais2018semi} also proves the above convergences for full many-body ground states (that is, without the Hartree-Fock approximation). The grand canonical case requires a more significant change in the full many-body proof, see Section \ref{sec:lower-bound-whole-mb}. Another notable difference with \cite{fournais2018semi} is that we obtain a pointwise convergence as well, and for this we use an argument due to Conlon \cite{conlon1983semi}.

\subsection{Notation}

Let us fist sum up all notation of the objects used in the article. 
\bigskip

\begin{tabular}{lll}
	\hline
	$V$ & confining exterior potential on $\R^d$ & Assumption \ref{cond:confining}
	 \\
	$E\in\R$ & chemical potential on $\R^d$&\\
	$w$ & potential of interaction &\\
	\hline
	$\hbar>0$ &  semiclassical parameter, reduced Planck constant &\\
	\hline
	$\ER^{\rm HF}_{\hbar,V,w}$ & $\hbar$-Hartree-Fock energy functional & see \eqref{eq-def:HF-energy}\\
	$\ER^{\rm rHF}_{\hbar,V,w}$ & $\hbar$-reduced Hartree-Fock energy functional &\\
	$\ER^{\rm rHF}_{\hbar,V,w,\lambda}$ & $\lambda$-enriched $\hbar$-reduced Hartree-Fock energy functional & see \eqref{eq-def:TF-energy_enriched}\\
	$\ER^{\rm TF}_{\hbar,V,w}$ & Thomas-Fermi energy functional & see \eqref{eq-def:TF-energy}\\
	$\ER^{\rm TF}_{\hbar,V,w,\lambda}$ & $\lambda$-enriched Thomas-Fermi energy functional  & see \eqref{eq-def:TF-energy_enriched}\\
	$\ER^{\rm Vlas}_{\hbar,V,w}$ & Vlasov energy functional & see \eqref{eq-def:Vlas-energy}
\end{tabular}

\bigskip

\begin{tabular}{lll}
		\hline
	$e_{\hbar,V,w}$ & ground state energy of the $\hbar$-grand-canonical Hamiltonian & see \eqref{eq-def:gs-energy-grandcan} \\
	$e^{\rm HF}_{\hbar,V,w}$ & $\hbar$-Hartree-Fock ground state energy  & see \eqref{eq-def:gs-energy-HF}\\
	$e^{\rm rHF}_{\hbar,V,w}$ & $\hbar$-reduced Hartree-Fock ground state energy &\\
	$e^{\rm TF}_{\hbar,V,w}$ & Thomas-Fermi ground state energy  & see \eqref{eq-def:gs-energy-TF}\\
	$e^{\rm Vlas}_{\hbar,V,w}$ & Vlasov ground state energy & see \eqref{eq-def:gs-energy-Vlas}\\
	\hline
	$\XR$ & energy space of Hartree-Fock functional & see \eqref{eq-def:set-HF}\\
	$\KR$ & admissible states of Hartree-Fock energy & see \eqref{eq-def:set-HF_min}\\
	$\KR^{\rm Vlas}$ & admissible states of Vlasov energy& see \eqref{eq-def:set-Vlas}\\
	$\XR^{\rm TF}$ & energy space of Thomas-Fermi functional& see \eqref{eq-def:set-TF}\\
	\hline
	$H_{\gamma,\hbar}$, $H_\gamma$ & $\hbar$-semiclassical mean-field operator& see \eqref{eq-def:H_gamma}\\
	\hline
\end{tabular}
\medskip

	In order to simplify notation, we will now write $V$ instead of $V-E$.
	We denote by $\omega_x(\delta,V)$ the modulus of continuity of $V$ at $x\in\R^d$ for $\delta>0$
	\begin{equation*}
		\omega_x(\delta,V):=\sup_{\{y\in\R^d\: : \: \abs{x-y}\leq\delta\}}\abs{V(x)-V(y)}.
	\end{equation*}
\bigskip

We then introduce more rigorously the Hartree-Fock, Thomas-Fermi, Vlasov and the grand canonical many-body fermionic ground state problems. Then, we will state the assumptions on the potentials $V$ and $w$ and the main results of the paper.

\subsubsection{Hartree-Fock setting}

Let us recall the Hartree-Fock functional expression $\ER_{\hbar,V,w}^{\rm HF}$ (c.f. \eqref{eq-def:HF-energy}) and let us explain each term is defined
\begin{equation*}
	\ER_{\hbar,V,w}^{\rm HF}(\gamma)
	=\hbar^d\tr((-\hbar^2\Delta+V)\gamma)+\frac{\hbar^{2d}}2\left(D_w(\rho,\rho)-{\rm Ex}_w(\gamma)\right).
\end{equation*}
Here, $D_w$ and ${\rm Ex}_w$ respectively denote the \emph{direct term} and \emph{exchange term} defined by
\begin{equation*}
	D_w(\rho,\tilde{\rho}) :=\int_{\R^d}\int_{\R^d} \rho(x)\tilde{\rho}(y)w(x-y)dx dy
\text{ and }
{\rm Ex}_w(\gamma) := \int_{\R^d}\int_{\R^d} \abs{\gamma(x,y)}^2w(x-y)dx dy.
\end{equation*}
Let $\Sigma>0$ large enough such that $V+\Sigma\geq 0$, the expression $\tr((-\hbar^2\Delta+V)\gamma)$ means that
	\begin{equation*}
		\tr((-\hbar^2\Delta+V)\gamma) := \tr(\sqrt{-\hbar^2\Delta+V+\Sigma}\gamma\sqrt{-\hbar^2\Delta+V+\sigma})-\Sigma\tr(\gamma).
	\end{equation*}

We denote by $\schatten^1$ the set of trace-class operators on $L^2(\R^d)$, which is endowed by the norm
\begin{equation*}
	\norm{\gamma}_{\schatten^1}:=\tr(\sqrt{\gamma^*\gamma}).
\end{equation*}
We introduce the definition set of $\ER_{\hbar,V,w}^{\rm HF}$ for a good enough $w$
\begin{equation}\label{eq-def:set-HF}
	\XR := \{ \gamma\in\schatten^1 \: :\: \normSch{(-\Delta+V+1)^{1/2}\gamma (-\Delta+V+1)^{1/2}}{1}{}<+\infty \},
\end{equation}
endowed with the norm
\begin{equation*}
	\norm{\gamma}_\XR := \normSch{(-\Delta+V+1)^{1/2}\gamma (-\Delta+V+1)^{1/2}}{1}{}.
\end{equation*}
By abuse of notation, without any possible conflict, we sometimes denote by $\ER^{\rm HF}_\hbar$ the Hartree-Fock energy, instead of $\ER^{\rm HF}_{\hbar,V,w}$. We look at the variational problem for a fixed $\hbar>0$ 
\begin{equation}\label{eq-def:gs-energy-HF}
	\inf_{\gamma\in\KR}\ER^{\rm HF}_{\hbar,V,w}(\gamma)=:e^{\rm HF}_{\hbar,V,w},
\end{equation}
on the convex closed subset of $\XR$
\begin{equation}\label{eq-def:set-HF_min}
	\KR := \{\gamma\in\XR \: :\: 0\leq \gamma\leq 1\}.
\end{equation}
where we minimize the Hartree-Fock functional.
We denote by $e^{\rm HF}_\hbar$ this ground state energy (we explain in Sections \ref{subsec:coerc-HF} and \ref{subsec:min-HF} under what conditions on the considered potentials $V$, $w$ this problem is well-posed).

We will show that the integrated Weyl law's asymptotics remain true if we consider even density matrices $\gamma_\hbar\in\KR$ that approach minimizers of the Hartree-Fock functional. 
	We call \emph{almost-minimizers} or \emph{approximate minimizers}, operators $\gamma_\hbar\in\KR$ such that there exists $\varepsilon_\hbar=o_\hbar(1)$ such that
	\begin{equation}\label{eq-def:almost-min-HF}
	\ER^{\rm HF}_{\hbar,V,w}(\gamma_\hbar)=e^{\rm HF}_{\hbar,V,w}+\varepsilon_\hbar.
	\end{equation}
We will prove semiclassical limits (Theorems \ref{thm:weak-WL_HF} and \ref{thm:int-WL_HF}) on these objects.

	At the first glance, the interest in working with such objects is not obvious when the problems admit minimizers. This is the case in this paper, which deals only with external confining potentials. A reasonable perpective to this paper is to extend these semiclassical limits to more general external potentials and to excited states. This is far from being obvious and the question remains open for now. The existence of Hartree-Fock minimizers is indeed more subtle for unconfined potentials (see for instance \cite{lions1987solutions}). Even in this case, there are always approximate minimizers, and they would naturally be the objects to consider.

In the semiclassical limit $\hbar\to 0$, the exchange term of $\ER^{\rm HF}_\hbar$ becomes negligible (see for instance \cite{bach1992error} who proved it in a canonical setting. It is what we prove in Lemma \ref{lemma:HF-to-rHF}). Therefore, it will therefore be easier to work with the \emph{reduced Hartree-Fock} functional $\ER^{\rm rHF}_\hbar$ which is the Hartree-Fock functional without its exchange term. We denote by $e^{\rm rHF}_\hbar$ its associated ground state energy.

\subsubsection{Thomas-Fermi and Vlasov settings}

The Thomas-Fermi energy functional $\ER^{\rm TF}_{V,w}$ (c.f. \eqref{eq-def:TF-energy}) is defined on the set of trial states
\begin{equation}\label{eq-def:set-TF}
	\XR^{\rm TF}_V:=\left\lbrace \rho\in L^1(\R^d)\cap L^{1+2/d}(\R^d),\: \rho\geq 0:\: \int_{\R^d}V(x)\rho(x)dx<\infty\right\rbrace
	.
\end{equation}
We denote by $e^{\rm TF}$ its ground state energy of the Thomas-Fermi functional, defined by the formula \eqref{eq-def:TF-energy}
\begin{equation}\label{eq-def:gs-energy-TF}
	e^{\rm TF}_{V,w} := \inf_{\rho\in\XR^{\rm TF}_V}\ER^{\rm TF}_{V,w}(\rho)
	.
\end{equation}
Let us introduce also the \emph{Vlasov energy} functional

\begin{equation}\label{eq-def:Vlas-energy}
\begin{split}
	\ER^{\rm Vlas}_{V,w}(m)
	 &:=\frac 1{(2\pi)^d}\int_{\R^d}\int_{\R^d} \abs{\xi}^2m(x,\xi) dx d\xi + \int_{\R^d}V(x)\rho_m(x)dx
	\\&\quad\quad
	+ \frac 12 \int_{\R^d}\int_{\R^d} \rho_m(x)w(x-y)\rho_m(y)dx dy,
\end{split}
\end{equation}
defined on the phase-space densities set
\begin{equation}\label{eq-def:set-Vlas}
	\KR^{\rm Vlas}_V := \inf\left\lbrace m\in L^1(\R^d\times\R^d),\: 0\leq m\leq 1,\: \iint_{\R^d\times\R^d}(\abs{\xi}^2+V(x))m(x)dx<\infty\right\rbrace,
\end{equation}
for
\begin{equation*}
	\rho_m(x):=\frac 1{(2\pi)^d}\int_{\R^d} m(x,\xi)d\xi.
\end{equation*}
We define its ground state energy
\begin{equation}\label{eq-def:gs-energy-Vlas}
	e^{\rm Vlas}_{V,w}:= \inf_{m\in\KR^{\rm Vlas}_V}\ER^{\rm Vlas}_{V,w}(m).
\end{equation}
As we will see in Section \ref{subsec:link-Vlas-TF}, Vlasov and Thomas-Fermi functionals are closely related. In particular, the Vlasov energy is used to study the Hartree-Fock ground state and thus relate it to the Thomas-Fermi model.

\subsection{Main results}

The aim of this paper consists in investigating the first leading term of the Weyl law with interactions, i.e.\ the one of Hartree-Fock minimizers' densities
\begin{itemize}
	\item into an integrated form: see Theorem \ref{thm:int-WL_HF} (proved in Section \ref{sec:proof-int-WL}),
	\item a pointwise form: see Theorem \ref{thm:pointwise-WL_HF} (proved in Section \ref{sec:proof-pointwise-WL}).
\end{itemize}
But in fact, as mentioned in the informal introduction, we provide also
\begin{itemize}
	\item the structure of $\hbar$-Hartree-Fock minimizers: see Theorem \ref{thm:EL_minimiz} (proved in Section \ref{subsec:EL-eq-min-HF}),
	\item the semiclassical weak asymptotics of $\hbar$-Hartree-Fock minimizers: see Theorem \ref{thm:weak-WL_HF} (proved in Section \ref{sec:proof-weak-WL}),
	\item  the semiclassical convergence of Hartree-Fock and the whole system ground state to Thomas-Fermi one: see Theorem \ref{thm:sc-lim-HF} (proved in Section \ref{sec:conv-gs-HF}) and Theorem \ref{thm:lower-bound-whole-mb} (proved in Section \ref{sec:lower-bound-whole-mb}).
\end{itemize}
We present below the precise assumptions, the statements of these Theorems and their interconnections.

\subsubsection*{Assumptions on the potentials}

Without loss of generality,  let us assume assume here that
\begin{assump}\label{cond:confining}
	$V:\R^d\to\R$ is a continuous non-negative potential such that $V(x)\to+\infty$ as $\abs{x}\to+\infty$.
\end{assump}

We detail below definitions of repulsivity on the interaction potential.

\begin{assump}[Repulsive potential]\label{cond:w-Dterm}
	Let $w$ be an even real-valued function $L^1_{\loc}(\R^d)$ such that
	\begin{equation*}
	\forall\rho\in L^\infty_c(\R^d,\R_+)\quad D_w(\rho,\rho)\geq 0.
	\end{equation*}
\end{assump}

\begin{rmk}\label{rmk:w-Dterm}
	Assumption \ref{cond:w-Dterm} ensures that the lower semicontinuous Thomas-Fermi functional is non-negative on $\XR_V^{\rm TF}$. In particular, the associated ground state energy should be finite.
	Note that when $w$ does not satisfy Assumption \ref{cond:w-Dterm}, the associated Thomas-Fermi ground state energy is not bounded from below for dimensions $d\geq 3$ where the contribution of the direct term is more important than the one of the kinetic one
	\begin{equation*}
		e^{\rm TF}_{V,w}=-\infty.
	\end{equation*}
	For instance, by choosing $\rho$ such that $D_w(\rho,\rho)<0$, and defining $\rho_n:=n\rho$ for any $n\in\N$, the energy $\ER^{\rm TF}_{V,w}(\rho_n)\to -\infty$ as $n\to \infty$.
	In this case, if we prove that the ground state behaves as the Thomas-Fermi model at the first order term, then the ground state of the system diverges.
\end{rmk}

We have restricted our study to repulsive potentials since we treat grand-canonical systems and then we want to bound by below the Thomas-Fermi ground state energy. It would also be interesting to understand what happens in the attractive case, since that is what we do in dimensions 1 and 2.

\begin{assump}[Alternative of attractive potential for dimensions $d=1,2$]\label{cond:w-Dterm-d=12}
	Let $p\in(1,\infty)$.
	Let $w\in L^1\cap L^p(\R^d)$ an even real-valued function such that
	\begin{equation}\label{eq:cond-E_HF_bound_below-d=12}
		\normLp{(\hat{w})_-}{\infty}{(\R^d)} <\frac 12\begin{cases}
		(2\sqrt{\pi})^{-1}
		&\text{if}\ d=1,\\
		(2\pi)^{-1}C_{\rm LT}^{-2} &\text{if}\ d=2
		.
	\end{cases}
	\end{equation}
	Here $C_{\rm LT}=C_{{\rm LT},d}>0$ denotes the constant in the kinetic Lieb-Thirring inequality (see for instance \cite[Theorem 3.2]{lieb-thirring1975}, \cite{lieb1976inequalities} and see \cite[Prop. 4]{frank-hunder-jex-nam2021} for an uppper bound)
	that holds for any operator $\gamma\geq 0$
	\begin{equation}\label{eq:LT-ineq}
		\normLp{\rho_\gamma}{1+2/d}{(\R^d)}\leq C_{{\rm LT,}d}\tr((-\Delta)\gamma)^{\frac d{d+2}}\norm{\gamma}_{L^2\to L^2}^{\frac 2{d+2}}
		.
	\end{equation}
\end{assump}

\begin{rmk}
	When the potential $w$ is non-negative, the direct term $D_w(\rho_\gamma,\rho_\gamma)$ is always non-negative.
	Moreover, $D_w(\rho_\gamma,\rho_\gamma)\geq {\rm Ex}_w(\gamma) $ and $V$ is trapping.
	Thus, since $V$ is bounded from below, there exists such that for any $\gamma\in\KR$ and any $\hbar>0$,
	\begin{equation*}
		\ER^{\rm HF}_{\hbar,V,w}(\gamma)\geq \hbar^d\tr((-\hbar^2\Delta+V)\gamma)\geq \hbar^d\tr((-\hbar^2\Delta+V)_-)\geq -C.  
	\end{equation*}
	Therefore, the energy $\inf_{\gamma\in\KR}\ER^{\rm HF}_{\hbar,V-E,w}(\gamma)$ is always bounded from below for any $\hbar>0$.
	We will show that it remains true under more general assumptions. We prove it in Lemma \ref{lemma:E_HF_bound_below} of Section \ref{subsec:about-prop-HF}.
\end{rmk}

In the statement of our results, there appears also the following assumption

\begin{defi}\label{def:w_Lp+Linfty-eps}
	Let $p\in[1,\infty]$.
	A function $w\in L^p(\R^d)+L^\infty_\varepsilon(\R^d)$ if for any $\varepsilon>0$, there exist $w_1\in L^p(\R^d)$ and $w_2\in L^\infty(\R^d)$ such that $\normLp{w_2}{\infty}{(\R^d)}\leq\varepsilon$ and $w=w_1+w_2$.
\end{defi}

	We will see that it is needed for the lower bound on the limit of the ground state energy in Theorem \ref{thm:sc-lim-HF}.
We will state later conditions on $p$ that guarantee the well-definition of the Hartree-Fock energy functional and the well-posedness of the related minimization problem (see Section \ref{subsec:about-prop-HF}).

\subsubsection*{Statement of the results}

The integral and pointwise Weyl laws actually rely on a weaker version.

\begin{thm}[Weak semiclassical limit of the density]\label{thm:weak-WL_HF}
	Let $V:\R^d\to\R$ be continuous and such that $V(x)\to+\infty$ as $\abs{x}\to+\infty$, and let $w\in L^{1+d/2}(\R^d)+L^\infty_\varepsilon(\R^d)$  satisfying Assumption \ref{cond:w-Dterm} (or in dimensions $d=1,2$, $w\in L^1(\R^d)\cap L^{1+d/2}(\R^d)$ with Assumption \ref{cond:w-Dterm-d=12}). Let $\{\gamma_\hbar\}_{\hbar>0}\subset\KR$ be a sequence such that for any $\hbar>0$, $\gamma_\hbar$ is an approximate minimizer of the $\hbar$-Hartree-Fock energy $\ER^{\rm HF}_{\hbar,V,w}$, defined in \eqref{eq-def:almost-min-HF}).
	Then, there exist $\rho_{\rm TF}$ a minimizer of the Thomas-Fermi energy $\ER^{\rm TF}_{V,w}$ and a decreasing subsequence $\{\hbar_n\}_{n\in\N}\subset\R_+^*$ such that $\hbar_n\to 0$ as $n\to+\infty$ and such that $\hbar_n^d\rho_{\gamma_{\hbar_n}}\rightharpoonup\rho_{\rm TF}$ weakly in $L^1(\R^d)\cap L^{1+2/d}(\R^d)$ as $n\to +\infty$.
\end{thm}

	Again, once we have the appropriate bounds of Lemma \ref{lemma:trace_bd_h^d-almost-min} for (almost-)minimizers of HF, the proof of Theorem \ref{thm:int-WL_HF} will follow from well-known	arguments from \cite{lieb1977thomas,fournais2018semi} using Husimi transforms and weak-* lower semi-continuity of the Vlasov functional.

\begin{rmk}[Assumptions on $w$]\label{rmk:hyp-p-w-conv-dens_HF}
	We restrict the exponent $p$ that appears in Theorem \ref{thm:sc-lim-HF}  
	\begin{equation*}
	w\in L^p(\R^d)+L^\infty_\varepsilon(\R^d) \quad(\text{ or } w\in L^1(\R^d)\cap L^p(\R^d))
	\end{equation*}
	to $p=1+d/2$ so that the sequence of almost Hartree-Fock minimizers' densities $\rho_{\gamma_\hbar}\in L^{1+2/d}(\R^d)$. This ensures they are trial states of the Thomas-Fermi functional and that the limit is well-defined.
\end{rmk}


From this weak convergence, we deduce an equivalent of the Weyl law \eqref{eq:int-without-int_1ordterm} for the Hartree-Fock approximation.

\begin{thm}[Integrated Weyl law]\label{thm:int-WL_HF}
	Let $V:\R^d\to\R$ be a continuous function such that $V(x)\to+\infty$ as $\abs{x}\to+\infty$ and $w\in L^{1+d/2}(\R^d)+L^\infty_\varepsilon(\R^d)$ which satisfies Assumption \ref{cond:w-Dterm} (or in dimensions $d=1,2$, $w\in L^1(\R^d)\cap L^{1+d/2}(\R^d)$ satisfies Assumption \ref{cond:w-Dterm-d=12}). Let $\{\gamma_\hbar\}_{\hbar>0}\subset\KR$ be a sequence of almost-minimizers of the $\hbar$-Hartree-Fock energy $\ER^{\rm HF}_{\hbar,V,w}$.
	Then, we have the semiclassical asymptotic (up to a decreasing subsequence $\{\hbar_n\}_n\subset\R_+^*$, $\hbar_n\to 0$)
		\begin{equation*}
		\lim_{\hbar\to 0}\hbar^d\tr(\gamma_\hbar)=\lim_{\hbar\to 0}\hbar^d\int_{\R^d}\rho_{\gamma_\hbar}(x)dx=\int_{\R^d}\rho_{\rm TF}(x)dx ,
		\end{equation*}
	for some minimizer $\rho_{\rm TF}$ of the Thomas-Fermi energy $\ER^{\rm TF}_{V,w}$ (the one in Theorem \ref{thm:weak-WL_HF}).
\end{thm}

	The proof of Theorem \ref{thm:int-WL_HF} will follow from the weak convergence proved in	Theorem \ref{thm:weak-WL_HF} together with the fact that V is confining.


The proofs of Theorems \ref{thm:weak-WL_HF} and \ref{thm:int-WL_HF} actually rely on the convergence of the ground state energies.

\begin{thm}\label{thm:sc-lim-HF}
	Let $d\geq1$ and $p\in(\max(d/2),\infty)$.
	For any $V\in\CR(\R^d,\R)$ such that $V(x)\to+\infty$ as $\abs{x}\to+\infty$, and any $w\in L^p(\R^d)+L^\infty_\varepsilon(\R^d)$ which satisfies Assumption \ref{cond:w-Dterm} (or $w\in L^1(\R^d)\cap L^p(\R^d)$ with Assumption \ref{cond:w-Dterm-d=12} for $d=1,2$), we have
		\begin{equation*}
		\lim_{\hbar\to 0} e^{\rm HF}_{\hbar,V,w} =	\lim_{\hbar\to 0} e^{\rm rHF}_{\hbar,V,w}= e^{\rm TF}_{V,w}.
		\end{equation*}
\end{thm}

	The proof of Theorem \ref{thm:sc-lim-HF} will follow from the same arguments as \cite[Prop. 2.5]{fournais2018semi} (upper and lower bounds) once we show our Lemma 2.5 (that is, while
	we do not have the condition $\tr\gamma=N$ as in \cite{fournais2018semi}, we will show that for (almost-)minimizers, we do have 	$\tr((-\hbar^2\Delta+V+1)\gamma_\hbar)\lesssim\hbar^{-d}$. We then show that it is enough to conclude the proof.

\begin{rmk}\label{rmk-intro:cond-repuls}
Note that the repulsivity conditions \ref{cond:w-Dterm} and \ref{cond:w-Dterm-d=12} just ensure that the fundamental Thomas-Fermi energy is well-defined and thus that the semiclassical limit is finite. However, if they are not verified, we have $e^{\rm TF}_{V,w}=-\infty$, so the limit is still true if we replace the conclusion by
	\begin{equation*}
	\limsup_{\hbar\to 0} e^{\rm HF}_{\hbar,V,w} =	\limsup_{\hbar\to 0} e^{\rm rHF}_{\hbar,V,w}=-\infty.
	\end{equation*}
\end{rmk}

\begin{rmk}\label{rmk-intro:WL-withoutint-energy}
	In the non-interacting case, by the Weyl law, the Hartree-Fock ground state energy $\ER^{\rm HF}_{\hbar,V,w=0}$ satisfies
	\begin{align*}
	\lim_{\hbar\to 0}e^{\rm HF}_{\hbar,V,w=0} 	= \lim_{\hbar\to 0}\hbar^d\tr((-\hbar^2\Delta+V)_-)
	&= 
	\frac{1}{(2\pi)^d}\iint_{\R^d\times\R^d}(\abs{\xi}^2+V(x))_- dxd\xi
	%
	%
	\\&
	= \frac 2{d+2}\frac{\abs{B_{\R^d}(0,1)}}{(2\pi)^d}\int_{\R^d}[(-V(x))_+]^{1+d/2}dx
	.
	\end{align*}
	The right-hand side term to the factor $\hbar^d$ corresponds to the Thomas-Fermi ground state energy $e^{\rm TF}_{V,w=0}$.
	In fact, due to the Thomas-Fermi equation \eqref{eq:TF}, the unique minimizer is $\rho_{\rm TF}(x)=c_{\rm TF}^{-d/2}(-V(x))_+^{d/2}$ when $w=0$.
\end{rmk}

Let's state now our pointwise version of the Weyl law.

\begin{thm}[Pointwise semiclassical limit of the density]\label{thm:pointwise-WL_HF}
	Let $V:\R^d\to\R$ be continuous and such that $V(x)\to+\infty$ when $\abs{x}\to+\infty$.
		Let	$w:\R^d\to\R$ be a continuous even function such that
		\begin{itemize}
			\item $w\in L^{1+d/2}(\R^d)+L^\infty_\varepsilon(\R^d)$ with Assumption \ref{cond:w-Dterm}, for $d\geq 2$
			\item $w\in L^{1+d/2}(\R^d) \cap L^2(\R^d)+L^\infty_\varepsilon(\R^d)$ with Assumption \ref{cond:w-Dterm}, for $d\geq 1$,
			\item or alternatively $w\in L^1(\R^d)\cap L^{1+d/2}(\R^d)$ with Assumption \ref{cond:w-Dterm-d=12} for $d=2$,
			\item  $w\in L^1(\R^d)\cap L^{1+d/2}(\R^d) \cap L^2(\R^d)$ with Assumption \ref{cond:w-Dterm-d=12} for $d=1,2$.
		\end{itemize}
		and such that
		\begin{equation}\label{cond:equicont_conlon}
		\nabla w\in L^{1+d/2}(\R^d)+L^\infty(\R^d).
		\end{equation}
	Let $\{\gamma_\hbar\}_{\hbar>0}\subset\KR$ be a sequence of minimizers of the $\hbar$-Hartree-Fock energy $\ER^{\rm HF}_{\hbar,V,w}$ such that for any $\hbar>0$
	\begin{equation}\label{eq:pointwise-WL-EL}
	\gamma_\hbar = \indicatrice{-\hbar^2\Delta +V+\hbar^d\rho_{\gamma_\hbar}\ast w
		-\hbar^d X_w(\gamma_\hbar)
		\leq 0}.
	\end{equation}
	Then, we have the pointwise limit (up to a decreasing subsequence $\{\hbar_n\}_n\subset\R_+^*$, $\hbar_n\to 0$ as $n\to+\infty$), for any $x\in\R^d$
	\begin{equation*}
	\lim_{\hbar\to 0} \hbar^d\rho_{\gamma_\hbar}(x)= \rho_{\rm TF}(x),
	\end{equation*}
	where $\rho_{\rm TF}$ is a minimizer of the Thomas-Fermi energy $\ER^{\rm TF}_{V,w}$. 
\end{thm}

	We emphasize that Theorem \ref{thm:pointwise-WL_HF} does not follow from Theorem \ref{thm:weak-WL_HF} and \ref{thm:int-WL_HF}, since weak convergence in $L^1$ and convergence in $L^1$-norm does not imply
	convergence in $L^1$ in general. If it were true, one could then argue	that convergence in $L^1$ implies convergence almost everywhere up to a subsequence. Let us also mention that, even if one has convergence in $L^1$ of the	densities, this would only imply convergence almost everywhere of the densities up to	a subsequence, while our result implies for instance that if the Thomas-Fermi minimizer is unique, one even has convergence almost everywhere as $\hbar\to0$.

These asymptotics are another consequence of Theorem \ref{thm:weak-WL_HF} and of the pointwise Weyl law without interaction (see later Theorem \ref{thm:WL-without-int_conlon}), which only requires $\CR^0$ regularity of confining external potentials. This is less restrictive than the assumption mentioned in the regular statements of the Weyl law. As mentioned in the introduction, the proof follows Conlon's approach \cite{conlon1983semi}. It is therefore simpler than the usual one, which use semiclassical tools. However, we have to keep in mind that it provides only the leading term with non-optimal reminder.

We above assume that Hartree-Fock minimizers exist and that some of them have the form \eqref{eq:pointwise-WL-EL}. This is in fact true given Theorem \ref{thm:EL_minimiz}, which is stated and proved in Section \ref{subsec:EL-eq-min-HF}.	It is a variation of arguments of \cite{bach-lieb-loss-sol1994,lenzmann2010minimizers} in the	grand canonical setting. It will be important for the proof of Theorem \ref{thm:pointwise-WL_HF}.	

	As explained below, Theorem \ref{thm:sc-lim-HF} provides the semiclassical convergence of the ground state of the Hartree-Fock energy. In the canonical setting \cite{fournais2018semi}, when the number of particles $N$ of the system is fixed, the ground state energy per particle for the whole fermionic system converges also to the Thomas-Fermi ground state energy at the effective limit $\hbar\to 0$ ($\hbar=N^{-1/d}$). As expected, the result remains the same in the grand-canonical setting, that we now introduce.
\subsubsection{Whole many-body grand-canonical setting} 
A grand-canonical ensemble is composed of $\N$ identical systems of finite particles with a fixed chemical potential $E$, which is shared with the other particles and the energy.
When the particles are fermionic, the system of grand-canonical states is described by the \emph{fermionic Fock space}
\begin{equation*}
\mathcal{F}:=\C\oplus\bigoplus_{N=1}^\infty L^2_a(\R^{dN}),
\end{equation*}
composed of sequences $\Psi=(\Psi_0,\Psi_1,\cdots,\Psi_N,\cdots)$ such that $\Psi_N\in L^2_a(\R^{dN})$ for any $N\in\N^*$. This Fock space $\mathcal{F}$ is usually endowed by the scalar product
\begin{equation*}
\prodscal{\Psi}{\Phi}_{\mathcal{F}}=\sum_{N\in\N}\prodscal{\Psi_N}{\Phi_N}_{L^2(\R^{dN})}
.
\end{equation*}
Moreover, the quantum Hamiltonian is
\begin{equation*}
\mathbf{P}=\bigoplus_{N=1}^\infty P_N,
\end{equation*}
where $P_N$ is the Hamiltonian on $L^2(\R^{dN})$ defined by
\begin{equation*}
P_N:= \sum_{j=1}^N (-\hbar^2\Delta_{x_j}+V(x_j)) +\hbar^d\sum_{1\leq i<j\leq N}w(x_i-x_j).
\end{equation*}
The ground state energy $e_{\hbar,V,w}$ of $\mathbf{P}$ per particle
\begin{equation}\label{eq-def:gs-energy-grandcan}
e_{\hbar,V,w}:= \hbar^d\inf\spec(\mathbf{P})
\end{equation}
is also defined as a function of the ground state energies of the canonical Hamiltonians $P_N$, up to a factor $\hbar^d$
\begin{equation*}
	e_{\hbar,V,w}=\inf_{N\geq 0}e_{N,V,w} =: \hbar^d\inf_{N\geq 0} \spec(P_N).
\end{equation*}
If we add a regularity condition of $w$, one can prove the grand-canonical equivalent of Proposition \cite[Prop. 3.5]{fournais2018semi}:
	\begin{equation*}
	\lim_{\hbar\to 0} e_{\hbar,V,w}= e^{\rm TF}_{V,w},
	\end{equation*}
that we deduce easily from Theorem \ref{thm:sc-lim-HF} since we always have the upper bound
\begin{equation*}
	e_{\hbar,V,w} \leq e^{\rm HF}_{\hbar,V,w}.
\end{equation*}
\begin{thm}\label{thm:lower-bound-whole-mb}
	Let $d\geq 1$. Let $V:\R^d\to\R$ be continuous and such that $V(x)\to +\infty$ as $\abs{x}\to+\infty$.
	Let $w$ an even function such that $\hat{w}\in L^1(\R^d)$ such that
	\begin{itemize}
		\item  $w\in L^{1+d/2}(\R^d)+L^\infty_\varepsilon(\R^d)$ that satisfies Assumption \ref{cond:w-Dterm},
		\item or alternatively, for $d\in\{1,2\}$, we can assume $w\in L^1(\R^d)\cap L^{1+d/2}(\R^d)$ with Assumption \ref{cond:w-Dterm-d=12} . 
	\end{itemize}
	Then, one has
		\begin{equation*}
		\liminf_{\hbar\to 0} e_{\hbar,V,w}\geq e^{\rm TF}_{V,w}.
		\end{equation*}
\end{thm}

	\section{Useful preliminary properties}\label{sec:preliminary}
	
	\subsection{About the Hartree-Fock functional}\label{subsec:about-prop-HF}
	
	We state in this section some basic properties on the Hartree-Fock functional, namely that it is well-defined in a suitable space, and that minimizers exist (for the proofs we refer for instance to \cite{lenzmann2010minimizers}). Moreover, we prove the result on the coercivity of the Hartree-Fock energy and on the structure of the minimizers.

	\subsubsection{Setting and well-posedness of the Hartree-Fock energy}\label{subsec:HF-basis}
	
	By definition of $\XR$, the linear term $\gamma\mapsto\tr((-\hbar^2\Delta+V)\gamma)$ is defined and bounded on $\XR$.
	We make explicit the conditions on $w$ for which the direct term $ D_w(\rho_\gamma,\rho_\gamma)$ and the exchange term ${\rm Ex}_w(\gamma)$ are controlled by the $\XR$-norm, and for which the kinetic term of the Hartree-Fock energy can absorb the exchange term.

	\begin{prop}[Bound on the direct term]\label{fact:control_direct}
		Let $d\geq 1$ and 
		\begin{equation*}\label{eq:control_direct-p}
		p\in\begin{cases}
		[1,\infty]&\text{for}\ d=1,3,4,\\
		\{1\}&\text{for}\ d=2,\\
		\left[\frac d4,\infty\right] &\text{for}\ d\geq 5.
		\end{cases}
		\end{equation*}
		\begin{itemize}
			\item[(i)] Let  $w\in L^p(\R^d)+L^\infty(\R^d)$. Then, there exists $C>0$ such that for any $0\leq\gamma\leq1$
			\begin{equation}\label{eq:control_direct-d>2}
			\abs{D_w(\rho_\gamma,\rho_\gamma)}\leq C\left[\tr((-\Delta)\gamma)^2+\tr(\gamma)^2\right].
			\end{equation}
			\item[(ii)] Assume that $d\in\{1,2\}$ and let $w\in L^p(\R^d)$.  Then, for any $\hbar>0$ and any $\gamma\in\schatten^1(L^1(\R^d))$ such that $0\leq\gamma\leq1$ 
			\begin{equation}\label{eq:control_direct-d=12}
			-C_{\gamma,\hbar}(d)\normLp{(\hat{w})_-}{\infty}{(\R^d)}\leq \hbar^dD_w(\rho_\gamma,\rho_\gamma)\leq C_{\gamma,\hbar}(d)\normLp{(\hat{w})_+}{\infty}{(\R^d)},
			\end{equation}
			with
			\begin{equation*}
			C_{\gamma,\hbar}(d)=(2\pi)^{d/2}\begin{cases}
			2\tr((1-\hbar^2\Delta)\gamma)&\text{if}\ d=1,\\
			C_{{\rm LT},2}^2\tr((-\hbar^2\Delta)\gamma)&\text{if}\ d=2.
			\end{cases}
			\end{equation*}
		\end{itemize}
	\end{prop}
	
	The proof of the two bounds above \eqref{eq:control_direct-d>2} and \eqref{eq:control_direct-d=12} are an adaptation of the one of \cite[Eq. (2.10)]{lenzmann2010minimizers} and \cite[Lem. 1]{lewin2015hartree}.

	\begin{prop}[Bound on the exchange term]\label{fact:control_exch}
		Let $d\geq 1$ and
		\begin{equation*}\label{eq:control_exch-p}
		p\in
		\begin{cases}
		[1,+\infty] &\text{if}\ d=1,\\
		(1,+\infty] &\text{if}\ d=2,\\
		\left[\frac d2,+\infty\right] &\text{if}\ d\geq 3.
		\end{cases}
		\end{equation*}
		For any $w\in L^p(\R^d)+L^\infty(\R^d)$, there exists $C>0$ such that for any $0\leq\gamma\leq 1$ and any $\varepsilon>0$ 
		\begin{equation}\label{eq:control_exch}
		\abs{{\rm Ex}_w(\gamma)}\leq C\left[\varepsilon\tr((-\Delta)\gamma)+\big(1+\varepsilon^{-\frac d{2p-d}}\big)\tr(\gamma)\right].
		\end{equation}
		As a consequence, assuming $p>1$ when $d=1$, there exists $\varepsilon_\hbar,\tilde{\varepsilon}_\hbar>0$ such that $\varepsilon_\hbar,\tilde{\varepsilon}_\hbar=o_\hbar(1)$, such that for any $\gamma\in\XR$
		\begin{equation}\label{eq:control_exch-err}
			\hbar^d\abs{{\rm Ex}_w(\gamma)}\leq C\left[\varepsilon_\hbar\tr((-\hbar^2\Delta)\gamma)+\tilde{\varepsilon}_\hbar \tr(\gamma)\right].
		\end{equation}
	\end{prop}

	The proof of Proposition \ref{fact:control_exch} is similar to one of \cite[Prop. 3.1]{fournais2018semi}.

	\subsubsection{Coercivity of the Hartree-Fock energy}\label{subsec:coerc-HF}
	
	First, we prove that for any fixed $\hbar>0$, the energy $\ER^{\rm HF}_\hbar$ is coercive on $\XR$: there exists $C_\hbar,c_\hbar>0$ such that for any $0\leq\gamma\leq 1$, we have $\ER^{\rm HF}_\hbar(\gamma)\geq C_\hbar\norm{\gamma}_{\XR}-c_\hbar$. In particular, $\ER^{\rm HF}_\hbar$ is bounded from below on $\XR$, which implies that the ground state energy $e_\hbar^{\rm HF}$ is finite and all minimizing sequences of $\ER^{\rm HF}_\hbar$ are bounded with respect to the $\XR$-norm.
	We also often use this relation of coercivity to give a semiclassical bound on the almost-minimizers of $\hbar$-(reduced)-Hartree-Fock functional and the whole system density matrices of trial states. This bound is crucial since we use it in several proofs below.
	
	\begin{lemma}\label{lemma:E_HF_bound_below}
		Let $V:\R^d\to\R$ that satisfies Assumption \ref{cond:confining}.
		Let $p\in[1,\infty]$ such that
		\begin{equation}\label{eq:E_HF_bound_below-p}
		p\in\begin{cases}
		(1,\infty]&\text{if}\ d=1,2,\\
		\left[\frac d2,\infty\right]&\text{if}\ d\geq 3. \\
		\end{cases}
		\end{equation}
		\begin{itemize}
			\item Let $w\in L^p(\R^d)+L^\infty(\R^d)$ satisfies Assumption \ref{cond:w-Dterm}.
			\item Otherwise, a second alternative for $d=1,2$ is to take  $w\in L^1\cap L^p(\R^d)$  satisfies Assumption \ref{cond:w-Dterm-d=12}.
		\end{itemize}
		Then, there exists $C>0$ and $h_0>0$ such that for any $\hbar\in(0,\hbar_0]$ and any $\gamma\in\KR$
		\begin{equation*}
			\ER^{\rm HF}_{\hbar,V,w}(\gamma)\geq \frac{\hbar^d} 4\tr((-\hbar^2\Delta+V+1)\gamma)-C.
		\end{equation*}
	\end{lemma}

	\begin{rmk}
		The assumption \eqref{eq:E_HF_bound_below-p} on $p$ takes in account both \eqref{eq:control_direct-p} and \eqref{eq:control_exch-p} so that the direct and the exchange terms are well-defined in $\XR$. It allows also \eqref{eq:control_exch-err} in order to control the exchange term by the kinetic term.
	\end{rmk}

	\begin{proof}[Proof of Lemma \ref{lemma:E_HF_bound_below}]
		Since we have Assumption \ref{cond:w-Dterm}, we only need to control the linear and the exchange term of the Hartree-Fock energy.
		Moreover, Assumption \ref{eq:cond-E_HF_bound_below-d=12}  for $d=1,2$ ensures that the direct term to be controlled by the linear term. Indeed, 
		\begin{align*}
			\hbar^d	\ER^{\rm HF}_{\hbar,V,w}(\gamma)
			&\geq \tr((-\hbar^2\Delta+V)\gamma)-\frac{\hbar^d}{2}\abs{{\rm Ex}_w(\gamma)}-\frac{\hbar^d}{2}\normLp{\hat{w}_-}{\infty}{(\R^d)}C_{\gamma,\hbar}(d)
			%
			%
			\\&\geq -\frac{\hbar^d}{2}\abs{{\rm Ex}_w(\gamma)}+
			\frac 12\begin{cases}
			\tr((-\hbar^2\Delta+V-1))\gamma)&\text{if}\ d=1,\\
			\tr((-\hbar^2\Delta+V)\gamma) &\text{if}\ d=2.
		\end{cases}
		\end{align*}
		Now, let us provide a lower bound of the Hartree-Fock energy without its direct term and exchange term. For convenience, let us treat only $\ER^{\rm rHF}_{\hbar,V,w}$ (we recover the other cases by adding multiplicative constants and changing the value of E).
		%
		For any $M>0$, we denote the spectral projectors $\Pi_{M,\hbar}^\pm$ by
		\begin{equation*}
		\Pi_{M,\hbar}^+ := \indicatrice{-\hbar^2\Delta+V> M}
		,\quad
		\Pi_{M,\hbar}^- := \indicatrice{-\hbar^2\Delta+V\leq M}
		.
		\end{equation*}
		%
		Notice that $\Pi_{M,\hbar}^++\Pi_{M,\hbar}^-=1$.
		Furthermore, for any $\gamma\geq 0$ and $M\geq 0$, 
		\begin{align*}
			\tr((-\hbar^2\Delta+V)\Pi_{M,\hbar}^+\gamma)
			&= \tr((-\hbar^2\Delta+V-M)\Pi_{M,\hbar}^+\gamma)
			+M\tr((1-\Pi_{M,\hbar}^-)\gamma)
			\\&\geq M\tr(\gamma) -M\tr(\Pi_{M,\hbar}^-).
		\end{align*}
		By the min-max principle and the integrated Weyl law (see \cite[Thm. 4.28]{Frank-Laptev-Weidl2023}), applied to the continuous and compactly supported potential $-(V-M)_-$
		\begin{align*}
		\tr(\Pi_{M,\hbar}^-)\leq \tr&(\indicatrice{-\hbar^2\Delta-(V-M)_-\leq 0})
		\leq
		\frac{\abs{B_{\R^d}(0,1)}}{(2\pi\hbar)^d}\int_{\R^d}(M-V(x))_-^{d/2}.
		\end{align*}
		Thus, there exists $C_M\geq 0$ such that
		\begin{equation}\label{eq:bound_PM+}
			\tr((-\hbar^2\Delta+V)\Pi_{M,\hbar}^+\gamma)\geq M\tr(\gamma)-\hbar^{-d}C_M.
		\end{equation}
		On the one hand, 
		\begin{equation*}
			\tr((-\hbar^2\Delta+V)\gamma)= \tr((-\hbar^2\Delta+V+1)\gamma)-\tr(\gamma).
		\end{equation*}
		On the other hand, by the Weyl law and \eqref{eq:bound_PM+}, the H\"older inequality and since $0\leq\gamma\leq 1$,
		\begin{align*}
		\tr((-\hbar^2\Delta+V)\gamma)
		&=  \tr((-\hbar^2\Delta+V)\Pi_{M,\hbar}^+\gamma) + \tr((-\hbar^2\Delta+V)\Pi_{M,\hbar}^-\gamma)
		\\&\geq M\tr(\gamma)-\hbar^{-d}C_M 
		\:+ (\min V)\tr(\Pi_{M,\hbar}^-\gamma)
		\\&\geq M\tr(\gamma)-\hbar^{-d}C_M
		.
		\end{align*}
		We obtain the lower bound of the linear term
		\begin{align*}
		\tr((-\hbar^2\Delta+V)\gamma)
		%
		&\geq \frac 12\tr((-\hbar^2\Delta+V+1)\gamma) +\frac 12(M-1)\tr(\gamma)-\hbar^{-d}C_M.
		\end{align*}
		Furthermore, taking into account the bound on the exchange term \eqref{eq:control_exch-err} and $M=1$, there exist $\varepsilon_\hbar,\tilde{\varepsilon}_\hbar=o_\hbar(1)$ and $C'>0$ such that
		\begin{equation*}
		\hbar^{-d}
		\ER^{\rm HF}_{\hbar,V,w}(\gamma)\geq \frac 12 \tr((-\hbar^2\Delta+V+1)\gamma)
		-\hbar^{-d}C'
		-\varepsilon_\hbar\tr((-\hbar^2\Delta)\gamma)-\tilde{\varepsilon}_\hbar\tr(\gamma).
		\end{equation*}
		We impose now that $\hbar\in(0,\hbar_0]$ for $h_0>0$ such that $1/2-\max(\varepsilon_{\hbar_0},\tilde{\varepsilon}_{\hbar_0})>1/4$. This ends the proof of Lemma \ref{lemma:E_HF_bound_below}.
	\end{proof}
	
	\subsubsection{Properties on almost-minimizers and link with reduced-Hartree-Fock model}
	
	We begin to state uniform bounds on almost-minimizers of Hartree-Fock functional, that are a direct consequence of their definition and of the coercivity bound (Lemma \ref{lemma:E_HF_bound_below}).

	\begin{lemma}
		\label{lemma:trace_bd_h^d-almost-min}
		Let $V$ and $w$ be functions that satisfy the same assumptions of Lemma \ref{lemma:E_HF_bound_below}:
		\begin{itemize}
			\item Let $w\in L^p(\R^d)+L^\infty_\varepsilon(\R^d)$ satisfies Assumption \ref{cond:w-Dterm}.
			\item Otherwise, a second alternative for $d=1,2$ is to take  $w\in L^1\cap L^p(\R^d)$  satisfies Assumption \ref{cond:w-Dterm-d=12},
		\end{itemize}
		with $p\in[1,\infty]$ that satisfies \eqref{eq:E_HF_bound_below-p}.
		Then, for any almost-minimizers $\gamma_\hbar$ of the $\hbar$-Hartre-Fock energy $\ER_{\hbar,V,w}^{\rm HF}$, there exists $C>0$ such that for any $\hbar\in(0,\hbar_0]$
		\begin{equation*}
			\tr((-\hbar^2\Delta+V+1)\gamma_\hbar)\leq C\hbar^{-d}.
		\end{equation*}
	\end{lemma}
		
	The almost-minimizers of Hartree-Fock functional satisfy more generally the following condition.
	 
	\begin{assump}\label{cond:trace_bd_h^d}
		Let a family $\{\gamma_\hbar\}_{\hbar\in(0,\hbar_0]}\subset\KR$ be such that there exists $C>0$ so that for any $\hbar\in(0,\hbar_0]$, one has $\tr(-\hbar^2\Delta\gamma_\hbar),\tr(\gamma_\hbar)<C\hbar^{-d}$.
	\end{assump}
	
	In particular, this assumption is necessary to have the next lemma, that states that the limits of $\ER^{\rm HF}_\hbar(\gamma_\hbar)$ and $\ER^{\rm rHF}_\hbar(\gamma_\hbar)$ are the same when $\hbar\to 0$.
	
	\begin{lemma}\label{lemma:HF-to-rHF}
		\label{exple:trace_bd_h^d}
		Let $V$ and $w$ be functions that satisfy the same assumptions as in Lemma \ref{lemma:E_HF_bound_below}.
		Then, there exists a sequence $\{r_\hbar\}_{\hbar>0}\subset\R^*_+$ such that $r_\hbar\to 0$ as $\hbar\to 0$ and such that for any 			$\{\gamma_\hbar\}_\hbar\in\KR$ which satisfies Assumption \ref{cond:trace_bd_h^d}
		\begin{equation*}\label{eq:proof-HF-to-rHF}
			\abs{\ER^{\rm HF}_{\hbar,V,w}(\gamma_\hbar)-\ER^{\rm rHF}_{\hbar,V,w}(\gamma_\hbar)}
				\leq r_\hbar.
		\end{equation*}
	\end{lemma}
	
	\begin{proof}[Proof of Lemma \ref{lemma:HF-to-rHF}]
		The bound  is a  direct consequence of Assumption \ref{cond:trace_bd_h^d} and the bound \eqref{eq:control_exch-err}: there exists $\varepsilon_\hbar,\tilde{\varepsilon}_\hbar=o_\hbar(1)$
		\begin{align*}
		\abs{\ER^{\rm HF}_{\hbar,V,w}(\gamma)-\ER^{\rm rHF}_{\hbar,V,w}(\gamma)}
		&\leq \frac{\hbar^{2d}}2\abs{{\rm Ex}_w(\gamma)}
		\leq \hbar^d(
		\underset{=o_\hbar(\hbar^{-d})}{\underbrace{
				\varepsilon_\hbar\tr(-\hbar^2\Delta\gamma)
		}}
		+
		\underset{=o_\hbar(\hbar^{-d})}{\underbrace{
				\tilde{\varepsilon}_\hbar\tr(\gamma)
		}}
	)
		.
		\end{align*}
	\end{proof}
	
	As well, the asymptotics of the ground state energies are the same.
	 
	 \begin{cor}\label{cor:HF-to-rHF_gs}
	 	Let $V$ and $w$ be functions that satisfy the same assumptions as in Lemma \ref{lemma:E_HF_bound_below}. Then, one has the equality on the ground state energies
	 	\begin{equation*}
		 	e^{\rm HF}_{\hbar,V,w}=e^{\rm rHF}_{\hbar,V,w}+ o_\hbar(1).
	 	\end{equation*}
	 \end{cor}
	 
	 \begin{proof}[Proof of Corollary \ref{cor:HF-to-rHF_gs}]
	Let $\{\gamma_\hbar\}\subset\KR$ such that $\ER^{\rm rHF}_{\hbar,V,w}(\gamma_\hbar)\leq 0$.
	 In particular, $\ER^{\rm rHF}_{\hbar,V,w}(\gamma)\leq C$ for any $C>0$ and any $\hbar>0$.
	 By the coercivity bound (Lemma \ref{lemma:E_HF_bound_below}), $\{\gamma_\hbar\}_\hbar$ satisfies Assumption \ref{cond:trace_bd_h^d}. Then, by Lemma \ref{lemma:HF-to-rHF}
	 \begin{equation*}
	 	\ER^{\rm rHF}_{\hbar,V,w}(\gamma_\hbar)=\ER^{\rm HF}_{\hbar,V,w}(\gamma_\hbar)+o_\hbar(1) \geq e^{\rm HF}_{\hbar,V,w}+o_\hbar(1).
	 \end{equation*}
	 	Minimizing the left-hand term on all $\gamma$, one has
		\begin{equation*}
		e^{\rm rHF}_{\hbar,V,w}\geq e^{\rm HF}_{\hbar,V,w}+o_\hbar(1).
		\end{equation*}
	 	Conversely, the inequality holds if we exchange $e^{\rm rHF}_{\hbar,V,w}$ and $e^{\rm HF}_{\hbar,V,w}$.
	 \end{proof}
	 
	 As a consequence, almost-minimizers of Hartree-Fock functional are the ones of the reduced-Hartree-Fock functional and conversely.

	\subsubsection{Existence of Hartree-Fock minimizers}\label{subsec:min-HF}
	
	Given the continuity and the coercivity of $\ER_{\hbar,V,w}^{\rm HF}$ in $\KR$, an other crucial ingredient for the existence of minimizers in \eqref{eq-def:gs-energy-HF}, is the weakly-$\ast$ lower semi-continuity of the functional $\ER_{\hbar,V,w}^{\rm HF}$ in $\XR$.
	
	\begin{lemma}\label{lemma:HF-wlsci}
		Let $p\in[1,\infty]$ such that
		\begin{equation*}
		p\in	\begin{cases}
		(1,\infty) &\text{if}\ d=1,2,\\
		\left[\frac d 2,\infty\right]&\text{if}\ d\geq 3.
		\end{cases}
		\end{equation*}
		Let $w$ be an even real-valued function on $\R^d$ such that $w\in L^p(\R^d)+ L^\infty(\R^d)$.
		Then, for any $\hbar>0$, if $\gamma_n\rightharpoonup\gamma$ weakly-$\ast$ in $\XR$, then
		\begin{equation*}
			\liminf_{n\to+\infty}\ER_{\hbar,V,w}^{\rm HF}(\gamma_n)\geq \ER_{\hbar,V,w}^{\rm HF}(\gamma).
		\end{equation*}
	\end{lemma}

	In our case, one can show that each term is weakly-$\ast$ lower semi-continuous on $\KR$. We will not detail the proof, which has the same structure as \cite[Cor. 4.1]{lenzmann2010minimizers} and \cite[Lem. 2.4]{lewin2011geometric}.
	
	We can therefore formulate the result of existence of minimizers for the Hartree-Fock functional energy as follows.
	
	\begin{prop}\label{prop:exist-min-HF}
		Let $p\in[1,\infty]$ such that
		\begin{equation*}
		p\in	\begin{cases}
		(1,\infty) &\text{if}\ d=1,2,\\
		\left[\frac d 2,\infty\right]&\text{if}\ d\geq 3.
		\end{cases}
		\end{equation*}
		Let $w$ be an even real-valued function on $\R^d$ such that $w\in L^p(\R^d)+ L^\infty(\R^d)$.
		Then, for any $\hbar>0$, the Hartree-Fock problem \eqref{eq-def:gs-energy-HF} admits a minimizer in $\KR$.
	\end{prop}

	\subsubsection{Nonlinear equation of the minimizers}\label{subsec:EL-eq-min-HF}

	We provide in the section the form of the minimizers of the $\hbar$-Hartree-Fock functional. They are closely related to the semiclassical \emph{mean-field operator} $H_\gamma=H_{\gamma,\hbar,V,w}$, defined by
	\begin{equation}\label{eq-def:H_gamma}
	H_\gamma:=(-\hbar^2\Delta+V)+\hbar^d\left[(\rho_\gamma\ast w)(x)-X_w(\gamma)\right]
	,
	\end{equation}
	where $X_w$ refers to the integral operator on $L^2(\R^d)$ defined by the kernel $X_w(\gamma)(x,y) :=w(x-y)\gamma(x,y)$. 
	\begin{thm}[Structure of Hartree-Fock minimizers]\label{thm:EL_minimiz}
		Assume that $V:\R^d\to\R$ is continuous and such that $V(x)\to+\infty$ as $\abs{x}\to+\infty$. 
		Let $p>\max(1,d/2)$.
		Assume that $w\in L^p(\R^d)+L^\infty(\R^d)$ is even and positive. Then, for any minimizer $\gamma_\hbar$ of $\ER^{\rm HF}_{\hbar,V,w}$ in $\KR$, there exists a self-adjoint operator $0\leq Q_\hbar\leq 1$ on $L^2(\R^d)$ such that
		\begin{equation*}
		\gamma_\hbar = \indicatrice{H_{\gamma_\hbar}<0}+Q_\hbar,
		\end{equation*}
		with $\range (Q_\hbar)\subset\ker(H_{\gamma_\hbar})$.
		Furthermore, there exists a projector $P_\hbar$ that minimizes $\ER^{\rm HF}_{\hbar,V,w}$ and that satisfies
		\begin{equation*}
		P_\hbar=\indicatrice{H_{P_\hbar}\leq 0} \text{ or } \indicatrice{H_{P_\hbar}<0}.
		\end{equation*}
	\end{thm}
	Let us provide a proof of this result.

	
	\paragraph[Step 1]{\it$\triangleright$ Step 1.} 
	We first prove that any minimizer $\gamma_\hbar$ of $\ER^{\rm HF}_\hbar$ minimizes the energy $\tr(H_{\gamma_\hbar}\cdot)$ on the set of self-adjoint operators of $\KR$.
	By explicit computation, using that $w$ is even, one has the following equality
	
	\begin{lemma}\label{lemma:taux-accr_hF}
		For any self-adjoint operators $\gamma,\tilde{\gamma}\in\XR$
		\begin{equation*}
		\ER^{\rm HF}_\hbar(\gamma+\tilde{\gamma})-\ER^{\rm HF}_\hbar(\gamma)
		= \tr(H_\gamma\tilde{\gamma})+  \frac{\hbar^d}2\left[D_w(\rho_{\tilde{\gamma}},\rho_{\tilde{\gamma}})-{\rm Ex}_w(\tilde{\gamma})\right]
		.
		\end{equation*}
	\end{lemma}
	
	
	Then, we deduce the following corollary on the rate of increase of the Hartree-Fock functional.
	
	\begin{cor}\label{cor:taux-accr_hF}
		For any self-adjoint $\gamma_\hbar,\tilde{\gamma}\in\XR$
		\begin{equation*}
		\lim_{t\to 0,\: t\in(0,1]}\frac{\ER^{\rm HF}_\hbar(\gamma_\hbar+t(\tilde{\gamma}-\gamma_\hbar))-\ER^{\rm HF}_\hbar(\gamma_\hbar)}t = \tr(H_{\gamma_\hbar}(\tilde{\gamma}-\gamma_\hbar))
		.
		\end{equation*}
	\end{cor}
	
	
	Since $\gamma_\hbar$ is a minimizer of $\ER^{\rm HF}_\hbar$ on $\KR$, for any $\tilde{\gamma}\in\KR$ and $t\in[0,1]$, we have $\ER^{\rm HF}_\hbar(\gamma_\hbar+t(\tilde{\gamma}-\gamma_\hbar))\geq \ER^{\rm HF}_\hbar(\gamma_\hbar)$.
	By Corollary \ref{cor:taux-accr_hF} and linearity of the trace, for any self-adjoint $\tilde{\gamma}\in\KR$,
	\begin{equation*}
	\tr(H_{\gamma_\hbar}\tilde{\gamma})\geq \tr(H_{\gamma_\hbar}\gamma_\hbar).
	\end{equation*}
		\paragraph[Step 2]{\it$\triangleright$ Step 2.}  Let us prove that $P_-^{\gamma_\hbar} := \indicatrice{H_{\gamma_\hbar}< 0}$ is a minimizer of $\tr(H_{\gamma_\hbar}\cdot)$. Note that the same argument works for $\indicatrice{H_{\gamma_\hbar}\leq 0}$.
	Let us also define $P_+^{\gamma_\hbar}:= 1-P_-^{\gamma_\hbar}$.
	For any operator $\gamma$ on $L^2(\R^d)$, let us denote $\gamma_{\pm \pm} := P_{\pm}^{\gamma_\hbar}\gamma P_{\pm}^{\gamma_\hbar}$.
	
	Let $\tilde{\gamma}\in\KR$. Using the relation $
	P_\pm^{\gamma_\hbar}H_{\gamma_\hbar}P_\pm^{\gamma_\hbar} = \pm P_\pm^{\gamma_\hbar}\abs{H_{\gamma_\hbar}}P_\pm^{\gamma_\hbar}
	$,
	the trace's cyclicity, the non-negativity of $\abs{H_{\gamma_\hbar}}\tilde{\gamma}_{++}$ and $\tilde{\gamma}_{--}\leq P_-^{\gamma_\hbar}$ (since $0\leq \tilde{\gamma}\leq 1$)
	\begin{align*}
	\tr(H_{\gamma_\hbar}\tilde{\gamma}) 
	&\geq  -\tr(\abs{H_{\gamma_\hbar}}\tilde{\gamma}_{--})
	\geq -\tr(\abs{H_{\gamma_\hbar}}P_-^{\gamma_\hbar})
	= \tr(H_{\gamma_\hbar}P_-^{\gamma_\hbar})
	.
	\end{align*}
	Since, this bound is true for any $\tilde{\gamma}\in\KR$, that therefore yields the claimed result.

	\paragraph[Step 3]{\it$\triangleright$ Step 3. Properties of $Q_\hbar:=\gamma_\hbar-P_-^{\gamma_\hbar}$.} Let us now explain why the range of $Q_\hbar:=\gamma_\hbar-P_-^{\gamma_\hbar}$ is included in the kernel of $H_{\gamma_\hbar}$.
	For this, we need to prove the inclusion $\range H_{\gamma_\hbar}\subset \ker Q_\hbar$, which follows by orthogonality.
	
	Since $\gamma_\hbar$ and $P_-^{\gamma_\hbar}$ are minimizers of $\tr(H_{\gamma_\hbar}\cdot)$ on $\KR$, we have $\tr(H_{\gamma_\hbar}Q_\hbar)=0$.
	Now, let us bound by below $\tr(H_{\gamma_\hbar}Q_\hbar)$ by using the non-negativity of $\abs{H_{\gamma_\hbar}}$ and $Q_{\hbar,++}-Q_{\hbar,--}\geq Q_\hbar^2$
	\begin{equation*}
		\tr(H_{\gamma_\hbar}Q_\hbar) = \tr(\abs{H_{\gamma_\hbar}}(Q_{\hbar,++}-Q_{\hbar,--}))\geq \tr(\abs{H_{\gamma_\hbar}}Q_\hbar^2).
	\end{equation*}
	Since the operator $H_{\gamma_\hbar}$ is self-adjoint and has a compact resolvent on $L^2(\R^d)$, there exists a $L^2$-orthonormal basis of eigenfunctions $\{\varphi_j^\hbar\}_{j\in\N}$ of  $\abs{H_{\gamma_\hbar}}$. By denoting $\{\lambda_j^\hbar\}_{j\in\N}\subset\R_+$ the sequence of the associated eigenvalues, one has for any $j\in\N$,
	\begin{align*}
		\tr(\abs{H_{\gamma_\hbar}}Q_\hbar^2) 
		&\geq \prodscal{\varphi_j^\hbar}{\abs{H_{\gamma_\hbar}}Q_\hbar^2\varphi_j^\hbar}_{L^2}
		=  \lambda_j^\hbar\normLp{Q_\hbar\varphi_j^\hbar}{2}{(\R^d)}^2.
	\end{align*}
	We obtain $Q_\hbar\varphi_j^\hbar=0$ for any $j\in\N$, which implies the claimed inclusion $\range (H_{\gamma_\hbar})\subset\ker(Q_\hbar)$.

	\paragraph[Step 4]{\it$\triangleright$ Step 4.} Assume that $w$ is positive almost everywhere in $\R^d$. Let us show that the number of eigenvalues of $Q_\hbar$ in $(0,1)$ is at most 1.
	We will follow the same argument as \cite[Cor. 1]{bach1992error}.
	
	Let $u_\hbar,v_\hbar\in L^2(\R^d)$ an orthonormal pair in $L^2(\R^d)$ and $\lambda_\hbar,\mu_\hbar\in(0,1)$ such that $Q_\hbar u_\hbar=\mu_\hbar u_\hbar$ and $Q_\hbar v_\hbar=\lambda_\hbar v_\hbar$.
	%
	From their definition and the previous step, we have $u_\hbar, v_\hbar\in \ker(H_{\gamma_\hbar})$. This implies $\tr(H_{\gamma_\hbar}(\bra{u_\hbar}\ket{u_\hbar}-\bra{v_\hbar}\ket{v_\hbar})) =0$.
	Combining it with Lemma \ref{lemma:taux-accr_hF}, one has for any $\delta>0$
	\begin{align*}
		&  \ER^{\rm HF}_\hbar(\gamma_\hbar+\delta\bra{u_\hbar}\ket{u_\hbar}-\delta\bra{v_\hbar}\ket{v_\hbar})-\ER^{\rm HF}_\hbar(\gamma_\hbar)
		%
		%
		\\&\quad
		= -\frac{\delta^2\hbar^2}2\iint_{\R^d\times\R^d} w(x-y)\abs{u_\hbar(x)v_\hbar(y)-u_\hbar(y)v_\hbar(x)}^2 dx dy
		,
	\end{align*}
	which is strictly negative because of the positivity of $w$ and of the non-proportionality of $u_\hbar$ and $v_\hbar$.
	By writting $\gamma_\hbar=\sum_{j\in J_\hbar}\nu_j^\hbar\bra{\varphi_j^\hbar}\ket{\varphi_j^\hbar}$ with $\{\nu_j^\hbar\}_{j\in J_\hbar}\subset(0,1]$ and with orthonormal basis $\{\varphi_j^\hbar\}_{j\in J_\hbar}$ on $L^2(\R^d)$, that includes $u_\hbar, v_\hbar$, one can write
	\begin{align*}
		&\gamma_\hbar+\delta\bra{u_\hbar}\ket{u_\hbar}-\delta\bra{v_\hbar}\ket{v_\hbar}
		\\&\quad
		= \sum_{j:\: \varphi_j^\hbar\neq u_\hbar,\: \varphi_j^\hbar\neq v_\hbar }\nu_j^\hbar\bra{\varphi_j^\hbar}\ket{\varphi_j^\hbar} 
		+(\lambda_\hbar+\delta)	\bra{u_\hbar}\ket{u_\hbar} +(\mu_\hbar-\delta)	\bra{v_\hbar}\ket{v_\hbar}
		.
	\end{align*}
	Taking any $\delta=\min(\mu_\hbar,1-\lambda_\hbar)$, the operator $\gamma_\hbar+\delta\bra{u_\hbar}\ket{u_\hbar}-\delta\bra{v_\hbar}\ket{v_\hbar}$ belongs to $\KR$.

	That leads to a contradiction since $\gamma_\hbar$ is a minimizer of the Hartree-Fock energy on $\KR$. Hence, the operator $Q_\hbar$ admits at most one eigenfunction in the interval $(0,1)$.

	\paragraph[Step 5]{\it $\triangleright$ Step 5.} Let us now explain why for any minimizer $\gamma_\hbar$ of $\ER^{\rm HF}_\hbar$ on $\KR$, one can build a projector $P_\hbar$ that minimizes also $\ER^{\rm HF}_\hbar$. Let $P_1^\hbar$ be the projector on the eigenspace $\ker(Q_\hbar-\id)$. By the previous step, there exists $\theta\in\{0,1\}$, $\lambda_\hbar\in(0,1)$ and $u_\hbar$ a $L^2$-normalized eigenfunction of $Q_\hbar$ such that $Q_\hbar = P_1^\hbar+\theta\lambda_\hbar\bra{u_\hbar}\ket{u_\hbar}$. The operator $P_\hbar=\gamma_\hbar-\theta\lambda_\hbar\bra{u_\hbar}\ket{u_\hbar}=\indicatrice{H_{\gamma_\hbar}<0}+P_1^\hbar$ is a projector since $\range(P_1^\hbar)\subset\ker(H_{\gamma_\hbar})$.
	Let us explain why $P_\hbar$ is a minimizer of $\ER^{\rm HF}_\hbar$.
	\begin{itemize}
		\item[$\bullet$] If $\theta=0$, it is deal since $\gamma_\hbar=\indicatrice{H_{\gamma_\hbar}<0}+P_1^\hbar$.
		\item[$\bullet$] Assume that $\theta=1$. Using $u_\hbar\in\range(Q_\hbar)\subset\ker(H_{\gamma_\hbar})$, $P_\hbar$ satisfies the relation $\ER^{\rm HF}_\hbar(\gamma_\hbar-\lambda_\hbar\bra{u_\hbar}\ket{u_\hbar})=\ER^{\rm HF}_\hbar(\gamma_\hbar)$, and thus is a minimizer of $\ER^{\rm HF}_\hbar$.
	\end{itemize}
	
	\paragraph[Step 6]{\it$\triangleright$ Step 6.} Finally, it remains to prove why $P_\hbar$ satisfies the Euler-Lagrange equation. Since $w>0$, one can write $P_\hbar=\indicatrice{H_{P_\hbar}<0}+\tilde{Q}_\hbar$ with $0\leq \tilde{Q}_\hbar\leq 1$, which is also a projector, and $\range(\tilde{Q}_\hbar)\subset\ker(P_\hbar)$. Let us prove that $\tilde{Q}_\hbar=0$ or $\tilde{Q}_\hbar$ is the projector on $\ker(H_{P_\hbar})$. Assume that there exists an orthonormal family functions $\{u_\hbar,v_\hbar\}$ on $L^2(\R^d)$ such that $\tilde{Q}_\hbar u_\hbar=0$ and $\tilde{Q}_\hbar v_\hbar=v_\hbar$. Recall that $u_\hbar,v_\hbar\in\ker(P_\hbar)$ since they are in the image of $\tilde{Q}_\hbar$. Notice that $0\leq P_\hbar+\bra{u_\hbar}\ket{u_\hbar}-\bra{v_\hbar}\ket{v_\hbar}\leq 1$. Similarly as in Step 4, 	$P_\hbar+\bra{u_\hbar}\ket{u_\hbar}-\bra{v_\hbar}\ket{v_\hbar}\in\KR$ (it is also a projector),
	and
	\begin{equation*}
	\ER^{\rm HF}_\hbar(P_\hbar+\bra{u_\hbar}\ket{u_\hbar}-\bra{v_\hbar}\ket{v_\hbar})<\ER^{\rm HF}_\hbar(P_\hbar),
	\end{equation*}
	which is a contradiction with the fact that $P_\hbar$ minimizes $\ER^{\rm HF}_\hbar$ on $\KR$. Therefore, we have either $\tilde{Q}_\hbar=0$ or $\tilde{Q}_\hbar$ must be the spectral projector on the kernel of $P_\hbar$.

	This concludes the proof of Theorem \ref{thm:EL_minimiz}.

	\subsection{Links between the Vlasov and the Thomas-Fermi energy functionals}\label{subsec:link-Vlas-TF}
	
	We give in this section some relations between  Vlasov and Thomas-Fermi energy functionals. Their minimization problems are equivalent
		\begin{equation*}
				e^{\rm TF}_{V,w} = e^{\rm Vlas}_{V,w}.
		\end{equation*}
		This can be proved the bathtub principle \cite[Theorem 1.14]{lieb2001analysis}. We state below, the main steps of the proof.
	
	\begin{lemma}\label{claim:link-Vlas-TF}
		Let $\rho$ be a trial function for the Thomas-Fermi energy $\ER^{\rm TF}_{V,w}$. Then, $m(x,\xi):=\indicatrice{\abs{\xi}^2\leq c_{\rm TF}\rho(x)^{2/d}}$ is a trial function for the Vlasov energy $\ER^{\rm Vlas}_{V,w}$ and we have the following equalities
		\begin{equation*}
		\rho=\rho_m
		\end{equation*}
		and
		\begin{equation*}
		\frac 1{(2\pi)^d}\iint_{\R^d\times\R^d}\abs{\xi}^2m(x,\xi) dx d\xi = \frac{d}{d+2}c_{\rm TF}\int_{\R^d}\rho(x)^{1+2/d}dx.
		\end{equation*}
		As a consequence,
		\begin{equation*}
		\ER^{\rm TF}_{V,w}(\rho)=\ER^{\rm Vlas}_{V,w}(\indicatrice{(x,\xi) \: :\:\abs{\xi}^2\leq c_{\rm TF}\rho(x)^{2/d}}),
		\end{equation*}
		and we have the lower bound of the Thomas-Fermi energy
		\begin{equation*}
		e^{\rm TF}_{V,w}\geq e^{\rm Vlas}_{V,w}.
		\end{equation*}
	\end{lemma}

	
	\begin{lemma}\label{lemma:optVlas-TF}
		Let $m$ an admissible state for the Vlasov energy functional $\ER^{\rm Vlas}_{V,w}$. Then $\rho:=\rho_m$ is an admissible state for the Thomas-Fermi energy functional $\ER^{\rm TF}_{V,w}$.
		Moreover, setting
		\begin{equation*}
		\tilde{m}(x,\xi):=\indicatrice{(x,\xi) \: : \: \abs{\xi}^2\leq c_{\rm TF}\rho(x)^{2/d}},
		\end{equation*}
		we have
		\begin{equation*}
		\ER^{\rm Vlas}_{V,w}(m)\geq \ER^{\rm Vlas}_{V,w}(\tilde{m})=\ER^{\rm TF}_{V,w}(\rho),
		\end{equation*}
		and the equality holds if and only if $m=\tilde{m}$. In particular, the ground state energies
		\begin{equation*}
		e^{\rm Vlas}_{V,w}=e^{\rm TF}_{V,w},
		\end{equation*}
		and if $m$ is a minimizer of $\ER^{\rm Vlas}_{V,w}$, then $m=\tilde{m}$ and $\rho_m$ is a minimizer of $\ER^{\rm TF}_{V,w}$.
	\end{lemma}

	\subsection{Semiclassical tools}\label{subsec:semiclass-tools}
	
	Let us introduce in this section definitions of coherent states, Husimi measures and some properties which will be useful for proving the desired asymptotics for ground state energies and densities. Indeed, the main idea is to look instead at the limit of the sequences of Husimi measures associated with the sequences of minimizers or approximate minimizers of the $\hbar$-Hartree-Fock functional. We will see later that they are in fact minimizing sequences for the Vlasov energy.
	
	\begin{defi}\label{def:coherentstate-Husimi}
		Let $f\in H^1(\R^d)$ be a real-valued even function such that $\normLp{f}{2}{(\R^d)}=1$.
		Denote by $f^\hbar$, the normalized function 
		\begin{equation*}
			f^\hbar(y) := \hbar^{-d/4}f\left(\frac{y}{\sqrt{\hbar}}\right).
		\end{equation*}
		For any $x,\xi\in\R^d$, denote by $f_{x,\xi}^\hbar$ the \emph{coherent state}
		\begin{equation*}\label{eq-def:coherent-state}
			f_{x,\xi}^\hbar(y) := \hbar^{-d/4}f\left(\frac{x-y}{\sqrt{\hbar}}\right)e^{-i\frac{\xi\cdot y}h}.
		\end{equation*}
		Let us introduce the \emph{Husimi transform} defined for any operator $0\leq \gamma\leq 1$ on $L^2(\R^d)$ by
		\begin{equation*}\label{eq-def:husimi-transf}
			m_{\hbar,\gamma,f}(x,\xi):=\prodscal{f_{x,\xi}^\hbar}{\gamma f_{x,\xi}^\hbar}.
		\end{equation*}
		Fo any $m:\R^d\times\R^d\to[0,1]$, let us denote by $\rho_m$ the density
		\begin{equation*}\label{eq-def:rho_m}
			\rho_m(x):=\frac{1}{(2\pi)^d}\int_{\R^d}m(x,\xi)d\xi.
		\end{equation*}
	\end{defi}

	We first state useful formulas between operators and their associated Husimi measure, namely that under some assumption it is an integrable measure on the phase-space $\R^d\times\R^d$ with the Pauli exclusion constraint $0\leq m_\hbar\leq 1$.
	We provide also relations which link the kinetic part (resp. the direct term) of the $\hbar$-Hartree-Fock energy $\ER^{\rm HF}_{\hbar,V,w}$ of $\gamma_\hbar$ to the kinetic part (resp. the direct term) of the Vlasov energy $\ER^{\rm Vlas}_{V,w}$ of the associated Husimi transform $m_\hbar$. 
	\begin{lemma}\label{lemma:prop-husimi} Denote by  $m_\hbar:=m_{\gamma_\hbar,\hbar,f}$ the operator associated to an operator $\gamma_\hbar$ on $L^2(\R^d)$.
		\begin{itemize}
			\item[(i)]
			For any $u\in L^2(\R^d)$
			\begin{equation*}\label{eq:useful-form-<fxp,fxpu>}
				\prodscal{f_{x,\xi}^\hbar}{u}=(2\pi\hbar)^{d/2}\fourierh[f^\hbar(\cdot-x)u](\xi)
				.
			\end{equation*}
			Furthermore, for any $0\leq\gamma_\hbar\leq 1$ 
			\begin{equation*}\label{eq:rho_{m_h}}
				\rho_{m_\hbar} = \hbar^d \rho_{\gamma_\hbar}\ast(|{f^\hbar}|^2).
			\end{equation*}
			\item[(ii)]
			Assume that $0\leq\gamma_\hbar\leq 1$ and $\hbar^d\tr(\gamma_\hbar)\lesssim 1$. Then,
			the associated sequence of Husimi transforms $\{m_\hbar\}_\hbar$ is bounded on $L^1(\R^d\times\R^d)$ and $0\leq m_\hbar\leq 1$ for all $\hbar>0$.
			\item[(iii)] 	For any $\hbar>0$, any $0\leq\gamma_\hbar\leq 1$ and any $V:\R^d\to\R$ that satisfies Assumption \eqref{cond:confining}
			\begin{equation}\label{lemma:relation_m_h_gamma_h-lin_term}
			\begin{split}
				\frac{1}{(2\pi)^d}\iint_{\R^d\times\R^d}\abs{\xi}^2m_\hbar(x,\xi)dx d\xi
				&= \hbar^d\tr(-\hbar^2\Delta\gamma_\hbar)+\hbar^{d+1}\tr(\gamma_\hbar)\normLp{\nabla f}{2}{(\R^d)}^2
				\\\frac{1}{(2\pi)^d}\iint_{\R^d\times\R^d} V(x)m_\hbar(x,\xi)dx d\xi
				&= \hbar^d \int_{\R^d} \rho_{\gamma_\hbar}(x)\big(V\ast(|f^\hbar|^2)(x) dx
				.
			\end{split}
			\end{equation} 
			\item[(iv)]
			For any $\hbar>0$, $w\in L^\infty(\R^d)$, $\gamma_\hbar\in\schatten^1(L^2(\R^d))$ such that $0\leq\gamma_\hbar\leq 1$ 
				\begin{equation}\label{lemma:relation_m_h_gamma_h-dir_term}
				D_w(\rho_{\gamma_\hbar},\rho_{\gamma_\hbar})= \hbar^{-2d}D_w(\rho_{m_\hbar},\rho_{m_\hbar})+D_{w-w\ast|f^\hbar|^2\ast|f^\hbar|^2}(\rho_{\gamma_\hbar},\rho_{\gamma_\hbar}).
			\end{equation}
		\end{itemize}
	\end{lemma}

	Let us now state weak continuity of the application $m\mapsto\rho_m$ and a consequence of it. We refer to \cite[Prop. 3.13]{lewin-sabin2020hartree} for the proof.
	\begin{lemma}\label{claim:wlim-rho_mh-rho_m}
		Let $m_\hbar:\R^d\times\R^d\to[0,1]$ be such that $\iint\abs{\xi}^2m_\hbar(x,\xi)dx dx$ is uniformly bounded as $\hbar\to 0$. Assume that there exists $m:\R^d\times\R^d\to[0,1]$ such that $m_\hbar\rightharpoonup m$ weakly-$\ast$ in $L^\infty(\R^d\times\R^d)$ as $\hbar\to 0$.	Then, we have the convergence of the density $\rho_{m_\hbar} \to\rho_m$ in $\DR'(\R^d)$.
	\end{lemma}
	%
	%
	\begin{rmk}\label{rmk:lim-rho_mh-rho_m}
		Lemma \ref{claim:wlim-rho_mh-rho_m} can be applied to any bounded sequence $\{m_\hbar\}_{\hbar>0}$ (up to a subsequence $\{\hbar_n\}_n\subset\R_+^*$ that $\hbar_n\to 0$ as $n\to+\infty$) of $\KR^{\rm Vlas}_V$ such that $m_\hbar\rightharpoonup m$ weakly-$\ast$ on $L^\infty(\R^d\times\R^d)$.
	\end{rmk}

	The following lemma gives some consequences on the weak convergence of $\{\rho_{m_\hbar}\}_\hbar$ in the case where $\{\hbar^d\rho_{\gamma_\hbar}\}_\hbar$ has a weak limit. We refer to \cite[Cor. 3.12]{lewin-sabin2020hartree} for the proof.
	\begin{lemma}\label{claim:wlim-rho_mh-rho}
		For $1< q\leq\infty$, let a sequence $\{\gamma_\hbar\}_\hbar\subset\schatten^1(L^2(\R^d))$ such that $0\leq\gamma_\hbar\leq 1$ for all $\hbar>0$ and let $\rho:\R^d\to\R_+$ be a density such that $\hbar^d\rho_{\gamma_\hbar}\rightharpoonup\rho$
		\begin{itemize}
			\item[$\bullet$] weakly on $L^q(\R^d)$ if $q\in[1,\infty)$,
			\item[$\bullet$] or weakly$-\ast$ on $L^\infty(\R^d)$ if $q=\infty$.
		\end{itemize}
		Then, 
		\begin{itemize}
			\item[$\bullet$] if $q\in(1,\infty)$, the sequence $\rho_{m_\hbar}\rightharpoonup\rho$ weakly on $L^q(\R^d)$,
			\item[$\bullet$] if $q=\infty$, the sequence $\rho_{m_\hbar}\rightharpoonup\rho$ weakly-$\ast$ on $L^\infty(\R^d)$.
		\end{itemize}
	\end{lemma}

	\section{Proof of the semiclassical integrated Weyl law (Theorem \ref{thm:int-WL_HF})}\label{sec:proof-int-WL}
	
	Assume that we have Theorem \ref{thm:weak-WL_HF} and let us explain how to deduce Theorem \ref{thm:int-WL_HF}.
	
	Let $\{\gamma_\hbar\}_{\hbar>0}$ be a sequence of almost-minimizers of $\ER^{\rm HF}_{\hbar,V,w}$. By Lemma \ref{lemma:trace_bd_h^d-almost-min}), there exists $C>0$ such that for any $\hbar\in(0,\hbar_0]$,
	\begin{equation*}
		\tr\left((-\hbar^2\Delta+V+1)\gamma_\hbar\right)\leq C\hbar^{-d}.
	\end{equation*}
	Let us write 
	\begin{align*}
		&h^d\int_{\R^d}\rho_{\gamma_\hbar}(x)dx-\int_{\R^d}\rho_{\rm TF}(x)dx
		\\&\quad
		=  \int_{B_R} (\hbar^d\rho_{\gamma_\hbar}(x)-\rho_{\rm TF}(x)) dx+\int_{\R^d\setminus B_R} (\hbar^d\rho_{\gamma_\hbar}(x)-\rho_{\rm TF}(x)) dx.
	\end{align*}
	Let us explain why we can restrict ourselves to the limit of the integral on a ball.
	Let $\varepsilon>0$. Since $V$ is confining, we can chose $R>0$ large enough so that $V(x)\geq \varepsilon^{-1}$ for any $x$ outside of $B_R$. Then, for any $\hbar>0$,
	\begin{align*}
		&\abs{\int_{\R^d\setminus B_R} (\hbar^d\rho_{\gamma_\hbar}(x)-\rho_{\rm TF}(x)) dx}
		= \abs{\int_{\R^d}\indicatrice{\abs{x}\geq R}V(x)^{-1} \: V(x) (\hbar^d\rho_{\gamma_\hbar}(x)-\rho_{\rm TF}(x)) dx}
		\\&\qquad
		\leq C\varepsilon\normLp{V(\hbar^d\rho_{\gamma_\hbar}-\rho_{\rm TF})}{1}{(\R^d)}
		\leq C'\varepsilon
		.
	\end{align*}
	
	Morerover, Theorem \ref{thm:weak-WL_HF} yields that $\hbar^d\rho_{\gamma_\hbar}\rightharpoonup\rho_{\rm TF}$ in $L^1(\R^d)$, which implies that
	\begin{equation*}
		\lim_{\hbar\to 0}\int_{B_R} (\hbar^d\rho_{\gamma_\hbar}(x)-\rho_{\rm TF}(x)) dx =0.
	\end{equation*}
	Thus, we deduce the integrated Weyl law
	\begin{equation*}
		\lim_{\hbar\to 0}\hbar^d\int_{\R^d}\rho_{\gamma_\hbar}(x)dx=\int_{\R^d}\rho_{\rm TF}(x)dx .
	\end{equation*}
	This proves Theorem  \ref{thm:int-WL_HF}.

	\section{Proof of the semiclassical limit of the Hartree-Fock ground state energy (Theorem \ref{thm:sc-lim-HF})}\label{sec:conv-gs-HF}
	
	 We prove in this section the convergence of the Hartree-Fock ground state energy to the Thomas-Fermi ground state energy.

	\subsection{Reduction to the reduced Hartree-Fock energy}
	
	The main idea of the proof of Theorem \ref{thm:sc-lim-HF} consists in getting back to the reduced Hartree-Fock ground state. Indeed, we can deal the asymptotics in the semiclassical limit. Here, we do not fix the trace as in \cite{fournais2018semi}, but it is relevant to consider operators that satisfy the Assumption \ref{cond:trace_bd_h^d}. According to Lemma \ref{lemma:trace_bd_h^d-almost-min}, this is the case for the Hartree-Fock almost-minimizers.
	
	\subsection{The upper bound}
	
	We bound by above the $\hbar$-Hartree-Fock functional by the $\hbar$-reduced Hartree-Fock functional up to an error depending on $h$, and
	take a suitable element $\tilde{\gamma}_\hbar\in\KR$ that satisfies, for a given trial state  $\rho\in\test{\R^d,\R_+}$ of the Thomas-Fermi functional, the limit for the reduced Hartree-Fock functional 
	$\lim_{\hbar\to 0}\ER^{\rm rHF}_{\hbar,V,w}(\tilde{\gamma}_\hbar) = \ER^{\rm TF}_{V,w}(\rho)$. Then, one has
	\begin{align*}
		\limsup_{\hbar\to 0} e^{\rm HF}_{\hbar,V,w}\leq \limsup_{\hbar\to 0}\ER^{\rm HF}_{\hbar,V,w}(\tilde{\gamma}_\hbar) \leq \limsup_{\hbar\to 0} (\ER^{\rm rHF}_{\hbar,V,w}(\tilde{\gamma}_\hbar)+o_\hbar(1)) =\ER^{\rm TF}_{V,w}(\rho).
	\end{align*}
	Hence, one obtains the upper bound by minimizing on all the trial states $\rho$.
	
	This is essentially the same proof as in \cite[Sec. 3.1]{fournais2018semi}, by taking $\tilde{\gamma}_\hbar$ the extension by 0 of the spectral projector $\indicatrice{-\hbar^2\Delta_{C_R}-c_{\rm TF}\rho(x)^{2/d}\leq 0}$, where $C_R:=(-R/2,R/2)^d$ is such that it contains $\supp\rho$ and $-\Delta_{C_R}$ is the Dirichlet Laplace operator on $C_R$.
	More explicitly, $\tilde{\gamma}_\hbar$ can be written as a spectral projector on a $L^2(\R^d)$-orthonormal family $\{\tilde{u}_j^\hbar\}_{1\leq j\leq N_\hbar}\subset H^1(\R^d)$, that we define as an extension by 0 out of $C_R$ of an orthonormal family $\{u_j^\hbar\}_{1\leq j\leq N_\hbar}\subset H^2(C_R)\cap H^1_0(C_R)$ of eigenfunctions of $-\hbar^2\Delta_{C_R}-c_{\rm TF}\rho(x)^{2/d}$, associated to negative eigenvalues.

	\subsection{The lower bound}
	
	Let us now prove the bound
	\begin{equation*}
		\liminf_{\hbar\to 0} e^{\rm HF}_{\hbar,V,w} \geq \liminf_{\hbar\to 0} e^{\rm rHF}_{\hbar,V,w} 	=  e^{\rm TF}_{V,w}
	.
	\end{equation*}
	By Corollary \ref{cor:HF-to-rHF_gs}, we only have to prove the semiclassical limit of the reduced Hartree-Fock ground state energy to the Thomas-Fermi ground state  energy
	\begin{equation}\label{eq:proof-low-slim}
		\lim_{\hbar\to 0} e^{\rm rHF}_{\hbar,V,w} =  e^{\rm TF}_{V,w}.
	\end{equation}
	For any $\hbar>0$, let $\gamma_\hbar\in\KR$ be an almost-minimizer of the $\hbar$-reduced-Hartree-Fock energy
	\begin{equation*}
		\ER^{\rm  rHF}_{\hbar,V,w}(\gamma_\hbar)=  e^{\rm rHF}_{\hbar,V,w}+o_\hbar(1).
	\end{equation*}
	By  Lemma \ref{lemma:trace_bd_h^d-almost-min}, one has $h_0>0$ such that for any $\hbar\in(0,\hbar_0]$, the operator $\gamma_\hbar$ satisfies the uniform bound $\tr((-\hbar^2\Delta+V+1)\gamma_\hbar)<C$.
	Let $f_{x,\xi}^\hbar$ be a coherent state defined above with the additional assumption $f\in\CR_c^1(\R^d)$ and even. Denote by $m_\hbar$ the Husimi measure associated to $\gamma_\hbar$ and $f$.

	\begin{claim}\label{claim:bound-m_h}
		The sequence $\{m_\hbar\}_\hbar$ is a bounded sequence of trial functions for the Vlasov energy $\ER^{\rm Vlas}_{V,w}$.
		There exists $C>0$ such that for any $\hbar>0$
		\begin{equation*}
			0\leq m_\hbar\leq 1,\quad
			\iint_{\R^d\times\R^d}(\abs{\xi}^2+ V(x)+1)m_\hbar(x,\xi)dx d\xi\leq C.
		\end{equation*}
	\end{claim}
	\begin{proof}[Proof of Claim \ref{claim:bound-m_h}]
		 By the Eq. \eqref{lemma:relation_m_h_gamma_h-lin_term}
		\begin{align*}
			&\frac 1{(2\pi )^d} \iint_{\R^d\times\R^d}(\abs{\xi}^2+ V(x)+1)m_\hbar(x,\xi)dx d\xi
			\\&\quad
			=  \hbar^d\tr(-\hbar^2\Delta\gamma_\hbar)  +\hbar^d\int_{\R^d}(1+ V(x)+(V\ast|f^\hbar|^2)(x))\rho_{\gamma_\hbar}(x)dx 
			+\hbar^{d+1}\tr(\gamma_\hbar)\normLp{\nabla f}{2}{(\R^d)}^2 
			.
		\end{align*}
		Let us prove that there exist $C>0$, such that for any $\hbar>0$
		\begin{align*}
			\hbar^d\int_{\R^d}(V\ast|f^\hbar|^2)(x)\rho_{\gamma_\hbar}(x)dx \leq C.
		\end{align*}
		Combining it to \ref{cond:trace_bd_h^d}, one gets the desired bound.
		Indeed, by the definition of $f^\hbar$, that $f$ is normalized in $L^2(\R^d)$ and has a compact support, one has
		\begin{align*}
			\hbar^d\abs{\int_{\R^d} (V-V\ast|f^\hbar|^2)(x)\rho_{\gamma_\hbar}(x)dx} &= \abs{\int_{\R^d}\int_{\R^d} (V(x)-V(x-\sqrt{\hbar}y))\abs{f(y)}^2\rho_{\gamma_\hbar}(x) dx dy}
			\\&\leq\int_{\R^d}  \sup_{y\in\supp f}\abs{V(x)-V(x-\sqrt{\hbar}y)}\hbar^d\rho_{\gamma_\hbar}(x) dx			.
		\end{align*}
		Let $R_0>0$ such that $\supp f\subset B_{R_0}$. 
		Splitting the last integral into two parts on $B_R$ and the other on $\R^d\setminus B_R$ for $R>R_0$, on the one hand, by Assumption \ref{cond:trace_bd_h^d}, one has for any $\hbar>0$
		\begin{align*}
			\int_{B_R}  \sup_{y\in\supp f}\abs{V(x)-V(x-\sqrt{h}y)} \hbar^d\rho_{\gamma_\hbar}(x) dx&\leq  \sup_{x\in B_R}\omega_{x}(\sqrt{\hbar}R_0,V) \int_{\R^d}\hbar^d\rho_{\gamma_\hbar}(x) dx
			\\&\leq C\sup_{x\in B_R}\omega_{x}(\sqrt{\hbar}R_0,V).
		\end{align*}
		Since $V$ is continuous on $\R^d$, then $V$ is uniformly continuous on $B_R$ for any $R>0$. One has that $\sup_{x\in B_R}\omega_{x}(\sqrt{\hbar}R_0,V)\to 0$ as $\hbar\to 0$. On the other hand, for any $\hbar\in(0,1]$
		\begin{align*}
			\int_{\R^d\setminus B_R}  \sup_{y\in\supp f}\abs{V(x)-V(x-\sqrt{\hbar}y)} \hbar^d\rho_{\gamma_\hbar}(x) dx
			&\leq 2\int_{\R^d\setminus B_{R-R_0}}V(x)\hbar^d\rho_{\gamma_\hbar}(x)dx
			.
		\end{align*}
		The right hand-side of the inequality is uniformly bounded by $2\int_{\R^d}V(x)\hbar^d\rho_{\gamma_\hbar}(x)dx$, which is, by Assumption \ref{cond:trace_bd_h^d}, uniformly bounded with respect to $\hbar\in (0,1]$.
		Eventually, by taking the limit $\hbar\to 0$, and $R\to +\infty$, we obtain
		\begin{align*}
			\lim_{\hbar\to 0}\hbar^d\int_{\R^d}(V(x)-(V\ast|f^\hbar|^2)(x))\rho_{\gamma_\hbar}(x)dx =0.
		\end{align*}
		This ends up the proof of Claim \ref{claim:bound-m_h}.
	\end{proof}

	To obtain it, we have to prove that
	\begin{equation}\label{eq:lowbound_E-rHF}
		\ER^{\rm rHF}_{\hbar,V,w}(\gamma_\hbar)=\ER^{\rm Vlas}_{V,w}(m_\hbar)+\varepsilon_\hbar,
	\end{equation}
	Let us express the linear term and the direct term of $\ER^{\rm rHF}_\hbar$ with respect to $m_\hbar$ and $\rho_{\gamma_\hbar}$.
	By the equation \eqref{lemma:relation_m_h_gamma_h-lin_term},
	\begin{align*}
		\hbar^d\tr((-\hbar^2\Delta+V)\gamma_\hbar)
		&= \frac{1}{(2\pi )^d}\iint_{\R^d\times\R^d}(\abs{\xi}^2+V(x))m_\hbar(x,\xi)dx d\xi
		\\&\quad- \hbar^{d+1}\tr(\gamma_\hbar)\normLp{\nabla f}{2}{(\R^d)}^2 
		+\hbar^d \int_{\R^d} (V-V\ast (|f^\hbar|^2))(x)\rho_{\gamma_\hbar}(x)dx
		.
	\end{align*}
	Thus, using \eqref{lemma:relation_m_h_gamma_h-dir_term}, one has for any $\hbar\in(0,\hbar_0]$
	\begin{align*}
		\ER^{\rm rHF}_{\hbar,V,w}(\gamma_\hbar)
		&= \hbar^d\tr((-\hbar^2\Delta+V)\gamma_\hbar) +\frac{\hbar^{2d}}2 D_w(\rho_{\gamma_\hbar},\rho_{\gamma_\hbar})
		\\&= \frac{1}{(2\pi)^d}\iint_{\R^d\times\R^d}(\abs{\xi}^2+V(x)))m_\hbar(x,\xi) dxd\xi +\frac 12 D_w(\rho_{m_\hbar},\rho_{m_\hbar}) +\varepsilon_\hbar
		\\&
		= \ER^{\rm Vlas}_{V,w}(m_\hbar) +\varepsilon_\hbar
		,
	\end{align*}
	where
	\begin{align*}
		\varepsilon_\hbar
		&= -\hbar^{d+1}\tr(\gamma_\hbar)\normLp{\nabla f}{2}{(\R^d)}^2
		\\&\qquad
		+\hbar^d\int_{\R^d}(V-V\ast (|f^\hbar|^2))(x)\rho_{\gamma_\hbar}(x)dx
		+\frac{\hbar^{2d}}{2}D_{w\ast|f^\hbar|^2\ast|f^\hbar|^2-w}(\rho_{\gamma_\hbar},\rho_{\gamma_\hbar}).
	\end{align*}
	It remains to show why $\varepsilon_\hbar=o_\hbar(1)$.
	%
	Assumption \ref{cond:trace_bd_h^d} and $f\in H^1(\R^d)$ imply that
	\begin{equation*}
		\lim_{\hbar\to 0}\hbar^{d+1}\tr(\gamma_\hbar)\normLp{\nabla f}{2}{(\R^d)}^2 =0.
	\end{equation*}
	Moreover, one has seen in the proof of Claim \ref{claim:bound-m_h} that
		\begin{align*}
		\lim_{\hbar\to 0}\hbar^d\int_{\R^d}(V(x)-(V\ast|f^\hbar|^2)(x))\rho_{\gamma_\hbar}(x)dx =0.
		\end{align*}
	Since $w\in L^p(\R^d)+L^\infty_\varepsilon(\R^d)$, for a fixed $\varepsilon>0$, there exist $w_1\in L^p(\R^d)$ and $w_\infty\in L^\infty(\R^d)$ such that  $w=w_1+w_\infty$ and $\normLp{w_\infty}{\infty}{(\R^d)}\leq\varepsilon$.
	Since $\{|f^\hbar|^2\}_{\hbar>0}$ is an unit approximation, $w_1-w_1\ast|f^\hbar|^2\ast|f^\hbar|^2 \to 0$ strongly in $L^p(\R^d)$. Furthermore, one has $\normLp{w_\infty-w_\infty\ast|f^\hbar|^2\ast|f^\hbar|^2}{\infty}{(\R^d)}\leq 2\varepsilon$.
	By Proposition \ref{fact:control_direct}
	\begin{align*}
		&\abs{D_{w-w\ast|f^\hbar|^2\ast|f^\hbar|^2}(\rho_{\gamma_\hbar},\rho_{\gamma_\hbar})}
		\\&\quad
		\leq C \left(\norm{w_1-w_1\ast|f^\hbar|^2\ast|f^\hbar|^2}_{L^p(\R^d))}+2\varepsilon\right)\left(\tr(-\hbar^2\Delta\gamma_\hbar)^2+\tr(\gamma_\hbar)^2\right)
		\\&\quad
		\leq C'' \hbar^{-2d}(o_\hbar(1)+ \varepsilon)
		.
	\end{align*}
	Thus, letting $\varepsilon\to 0$, we have
	\begin{equation*}
		\lim_{\hbar\to 0}\hbar^{2d}D_{w-w\ast|f^\hbar|^2\ast|f^\hbar|^2}(\rho_{\gamma_\hbar},\rho_{\gamma_\hbar})=0.
	\end{equation*}
	This ends the proof of the semiclassical asymptotics of the ground state energy \eqref{eq:proof-low-slim}.

	\section{Proof of the weak semiclassical limit of the density (Theorem \ref{thm:weak-WL_HF})}\label{sec:proof-weak-WL}

	Before proving Theorem \ref{thm:weak-WL_HF}, let us state and prove a crucial ingredient: the weak lower semi-continuity of the Vlasov functional.
	
	\begin{lemma}\label{lemma:vlas-wlsci}
		Let $V:\R^d\to\R$ which satisfies Assumption \ref{cond:confining}
		 and $w\in L^{1+d/2}(\R^d)+L^\infty_\varepsilon(\R^d)$  satisfying Assumption \ref{cond:w-Dterm} (or in dimensions $d=1,2$, $w\in L^1(\R^d)\cap L^{1+d/2}(\R^d)$ that satisfies Assumption \ref{cond:w-Dterm-d=12}). For any $E\in\R$ and any bounded sequence $\{m_\hbar\}_\hbar\subset\KR^{\rm Vlas}_V$ such that $m_\hbar\rightharpoonup m$ weakly-$\ast$ on $L^\infty(\R^d\times\R^d)$,
		\begin{equation*}
		\liminf_{\hbar\to 0}	\ER^{\rm Vlas}_{V,w}(m_\hbar)\geq\ER^{\rm Vlas}_{V,w}(m).
		\end{equation*}
	\end{lemma}

	\begin{proof}[Proof of Lemma \ref{lemma:vlas-wlsci}] We treat the linear and the quadratic term separately as for the proof for the Hartree-Fock functional.
		\item[\quad 1)] We first show that the kinetic energy $m\mapsto \iint_{\R^d\times\R^d} (\abs{\xi}^2+V(x)-E) m(x,\xi)dx d\xi$ is weakly-$\ast$ lower semi-continuous in $L^\infty(\R^d\times\R^d)$ 
		(it is actually also true in $L^q(\R^d\times\R^d)$).
		Let us treat the two terms of the functional 
		\begin{equation*}
			\iint_{\R^d\times\R^d} (\abs{\xi}^2+V(x))_\pm m_\hbar(x,\xi)dx d\xi.
		\end{equation*}
		Let us introduce a radial decreasing function $\chi\in\test{\R^d\times\R^d,[0,1]}$, which is equal to 1 in the ball $B_{\R^d\times\R^d}(0,1)$ and for any $R>0$ the cut-off function $\chi_R(x):=\chi(x/R)$.
		On the one hand,
		\begin{align*}
			\iint_{\R^d\times\R^d} (\abs{\xi}^2+V(x))_+ m_\hbar(x,\xi)dx d\xi
			&\geq \iint_{\R^d\times\R^d} \chi_R(x,\xi)(\abs{\xi}^2+V(x))_+ m_\hbar(x,\xi)dx d\xi
			.
		\end{align*}
		By taking first the limit $\hbar\to 0$ of the left-hand side term and then $R\to +\infty$ by the monotone convergence theorem
		\begin{equation*}
			\liminf_{\hbar\to 0} \iint_{\R^d\times\R^d} (\abs{\xi}^2+V(x))_+ m_\hbar(x,\xi)dx d\xi\geq 	\iint_{\R^d\times\R^d} (\abs{\xi}^2+V(x))_+ m(x,\xi)dx d\xi.
		\end{equation*}
		On the other hand, since $V$ is a confining potential, the function $(x,\xi)\mapsto (\abs{\xi}^2+V(x))_-$ has a compact support and then is in $L^1(\R^d\times\R^d)$. Thus, using that $m_\hbar\rightharpoonup m$ weakly-$\ast$ in $L^\infty(\R^d\times\R^d)$
		\begin{align*}
			\lim_{\hbar\to 0}&\left(-\iint_{\R^d\times\R^d} (\abs{\xi}^2+V(x))_- m_\hbar(x,\xi)dx d\xi\right)
			= -\iint_{\R^d\times\R^d}(\abs{\xi}^2+V(x))_- m(x,\xi)dx d\xi.
		\end{align*}
		Finally, we recover the lower semi-continuity by adding the two limits
		\begin{equation*}
			\liminf_{\hbar\to 0}\iint_{\R^d\times\R^d} (\abs{\xi}^2+V(x)) m_\hbar(x,\xi)dx d\xi \geq \iint_{\R^d\times\R^d} (\abs{\xi}^2+V(x)) m(x,\xi)dx d\xi.
		\end{equation*}
		\item[\quad 2)] We prove now that the direct term $m\mapsto D_w(\rho_m,\rho_m)$ is strongly continuous. By Lemma \ref{claim:wlim-rho_mh-rho_m}, one has $\rho_{m_\hbar}\rightharpoonup \rho_m$ weakly in $L^1(\R^d)\cap L^{1+2/d}(\R^d)$.
		The choice of $w$ makes the sequence $\{w\ast\rho_{m_\hbar}\}_\hbar$ be bounded in $L^\infty(\R^d)$, $w\ast\rho_{m_\hbar}\to w\ast\rho_m$ a.e.\ and the  weak convergence $w\ast\rho_{m_\hbar}\rightharpoonup w\ast\rho$ in $L^1_{\loc}(\R^d)\cap L^{1+2/d}_{\loc}(\R^d)$.
		Besides, we write
		\begin{equation*}
			D_w(\rho_{m_\hbar},\rho_{m_\hbar})
			=\int_{\R^d}(w\ast\rho_{m_\hbar})(x)\rho_{m_\hbar}(x)dx,
		\end{equation*}
		which can be split into two parts $\int_{B_R}(w\ast\rho_{m_\hbar})(x)\rho_{m_\hbar}(x)dx$ and $\int_{\R^d\setminus B_R}(w\ast\rho_  {m_\hbar})(x)\rho_{m_\hbar}(x)dx$.
		For any fixed $R>0$, the first part tends to $\int_{B_R}(w\ast\rho_m)(x)\rho_m(x)dx$ as $\hbar\to 0$.
		Let $\varepsilon>0$.
		The Assumption \ref{cond:confining} implies that $V(x)\geq \varepsilon$ out of $B_R$, for $R>0$ large enough.
		Besides, since the sequences $\{w\ast\rho_{m_\hbar}\}_\hbar\subset L^\infty(\R^d\times\R^d)$ and $\{V\rho_{m_\hbar}\}_\hbar\subset L^1(\R^d)$ are bounded as $\hbar\to 0$
		\begin{align*}
			\abs{\int_{\R^d\setminus B_R}(w\ast\rho_{m_\hbar})(x)\rho_{m_\hbar}(x)dx}
			&\leq \abs{\int_{\R^d}(w\ast\rho_{m_\hbar})(x) \: \indicatrice{\abs{x}\geq R}V(x)^{-1} \: V(x)\rho_{m_\hbar}(x)dx}
			\\&\leq 	\varepsilon\normLp{w\ast\rho_{m_\hbar}}{\infty}{(\R^d)}\int_{\R^d}V(x)\rho_{m_\hbar}(x)dx
			\\&	<C\varepsilon.
	\end{align*}
		Therefore, as $R\to +\infty$, we obtain the limit $D_w(\rho_{m_\hbar},\rho_{m_\hbar})\to D_w(\rho_m,\rho_m)$ when $\hbar\to 0$.
		This ends the proof of Lemma \ref{lemma:vlas-wlsci}.
	\end{proof}
	
	We now prove the weak limit of the  almost-minimizers' densities.

	\begin{proof}[Proof of Theorem \ref{thm:weak-WL_HF}]
	
	Let $\{\gamma_\hbar\}_{\hbar>0}\subset\KR$ be an approximate minimizing sequence of the $\hbar$-Hartree-Fock energies associated to the potentials $V$ and $w$, at the energy $E$.
	By Lemma \ref{lemma:trace_bd_h^d-almost-min}, the energy
	$\hbar^d\tr((-\hbar^2\Delta+V+1)\gamma_\hbar)$ is uniformly bounded in $h$. This inequality combined with the Lieb-Thirring inequality \eqref{eq:LT-ineq}
	implies that the associated sequence of densities $\{\hbar^d\rho_{\gamma_\hbar}\}_\hbar$ is bounded in $ L^1(\R^d)\cap L^{1+2/d}(\R^d)$ and in $L^1(\R^d, V(x)dx)$.
	Then, there exists an integrable $\rho\geq 0$ such that $V\rho\in L^1(\R^d)$ and $\hbar^d\rho_{\gamma_\hbar}\rightharpoonup\rho$ weakly in $L^1(\R^d)\cap L^{1+2/d}(\R^d)$. In particular, $\rho$ is a trial function of the Thomas-Fermi energy $\ER^{\rm TF}_{V,w}$.
	Let us explain now why $\rho$ is a minimizer of the Thomas-Fermi energy. To do so, we link it again to the associated Vlasov energy $\ER^{\rm Vlas}_{V,w}$.
	Let $m_\hbar$ be the Husimi transform associated to $\gamma_\hbar$ and the $L^2$-normalized function $f\in\schwartz(\R^d)$.
	The equality \eqref{eq:rho_{m_h}} yields that $\{\rho_{m_\hbar}\}_\hbar\subset L^1(\R^d)\cap L^{1+2/d}(\R^d)$ is bounded. 
	Furthermore, by Lemma \ref{claim:wlim-rho_mh-rho}, one has that $\rho_{m_\hbar}\rightharpoonup\rho$ weakly in $L^1(\R^d)\cap L^{1+2/d}(\R^d)$.
	We have seen in the proof of the lower bound in Theorem \ref{thm:sc-lim-HF} (see  Claim \ref{claim:bound-m_h} and \eqref{eq:lowbound_E-rHF}) that $\{m_\hbar\}_{\hbar>0}$ is a minimizing sequence of the Vlasov energy:
	\begin{equation*}
		\lim_{\hbar\to 0}\ER^{\rm Vlas}_{V,w}(m_\hbar)
		=\lim_{\hbar\to 0}\hbar^d\ER^{\rm HF}_{\hbar,V,w}(\gamma_\hbar)
		= e^{\rm TF}_{V,w}
		= e^{\rm Vlas}_{V,w}.
	\end{equation*}
	In addition since $\{m_\hbar\}_{\hbar>0}$ is bounded, there exists $ 0\leq m\leq 1$ such that $m_\hbar\rightharpoonup m$ weakly-$\ast$ in $L^\infty(\R^d\times\R^d)$.
	By Lemma \ref{lemma:vlas-wlsci}, we deduce that the limit $m$ minimizes the Vlasov energy. Then, by Lemma \ref{lemma:optVlas-TF}, there exists a minimizer $\rho_{\rm TF}$ of the Thomas-Fermi energy such that
	\begin{equation*}
		m(x,\xi)=\indicatrice{(x,\xi) \: :\: \abs{\xi}^2\leq c_{\rm TF}\rho_{\rm TF}(x)^{2/d}}.
	\end{equation*}
	Notice, that this minimizer is exactly the density $\rho_m$.
	Eventually, since the limit of $\{\hbar^d\rho_{\gamma_\hbar}\}_{\hbar>0}$ and $\{\rho_{m_\hbar}\}_{\hbar>0}$ have the same weak limit on $L^1(\R^d)\cap L^{1+2/d}(\R^d)$ (and that the limits are respectively $\rho$ and $\rho_m$), one has $\rho=\rho_m$. Thus, the weak limit of $\{\hbar^d\rho_{\gamma_\hbar}\}_{\hbar>0}$ is $\rho=\rho_{\rm TF}$.
	This ends the proof of Theorem \ref{thm:weak-WL_HF}.
\end{proof}
		
	\section{Proof of the semiclassical pointwise Weyl law (Theorem \ref{thm:pointwise-WL_HF})}\label{sec:proof-pointwise-WL}

	\subsection{Asymptotics without interaction}
	
	Let us prove first the Weyl law in the non-interacting case, with a Conlon type proof (see \cite[Thm 3.6]{conlon1983semi}).
	
	\begin{thm}[Weyl law]\label{thm:WL-without-int_conlon}
		Let $V:\R^d\to\R$ which satisfies Assumption \ref{cond:confining} and let $E\in\R$.
		Then, for any $x\in\R^d$, we have
		\begin{equation*}
			\lim_{\hbar\to 0} \hbar^d\indicatrice{-\hbar^2\Delta+V\leq E}(x,x) = \frac{\abs{B_{\R^d}(0,1)}}{(2\pi )^d}(E-V(x))_+^{d/2}.
		\end{equation*}
	\end{thm}
	
	We will see later in Lemma \ref{lemma:pointwise-WL_HF-h} that the proof can be adapted to the case where $V$ depends on the semiclassical parameter $h$. Let us first state the two main ingredients: the Hardy-Littlewood  Tauberian theorem and Lemma \ref{lemma:pointwise-WL_HF-exp}, which will allow us to deduce our limits from the Tauberian theorem.

	\begin{lemma}[{Hardy-Littlewood Tauberian theorem \cite[Thm. 10.3]{simon2005functionalint}}]\label{lemma:tauberian-thm}
		Let $E\in\R$, $\{m_\gamma\}_{\gamma>0}$ and $m_\infty$ be a non-negative measures on $[-E,+\infty)$ such that their Laplace transform is well-defined
		\begin{equation*}
			\forall \gamma\in(0,+\infty],\quad
			\int_{-E}^{+\infty}e^{-\alpha u}dm_\gamma(u)<\infty,
		\end{equation*}
		and such that for any $\alpha>0$
		\begin{equation}\label{eq:tauberian-thm_ass}
			\lim_{\gamma\to +\infty}\int_{-E}^{+\infty}e^{-\alpha u}dm_\gamma(u) = \int_{-E}^{+\infty}e^{-\alpha u}dm_\infty(u).
		\end{equation}
		Then, one has
		\begin{equation}\label{eq:tauberian-thm_res}
			 \lim_{\gamma\to +\infty} dm_\gamma([-E,0]) = dm_\infty([-E,0]).
		\end{equation}
	\end{lemma}
	
	\begin{lemma}[{\cite[Lem. 3.5]{conlon1983semi}}]\label{lemma:pointwise-WL_HF-exp}
		Let $d\geq 1$ and $W:\R^d\to\R$ be continuous and such that $W(x)\to +\infty$ as $\abs{x}\to +\infty$. Then, one has for all $x\in\R^d$
		\begin{equation}\label{eq:{lemma:pointwise-WL_HF-exp}}
			\lim_{t\to 0^+} (4\pi t)^{d/2}e^{-t\big(-\Delta+(1/t) W(x)\big)}(x,x)=e^{-W(x)}.
		\end{equation}
	\end{lemma}

	Let us recall Conlon's proof of Lemma \ref{lemma:pointwise-WL_HF-exp}, which is exactly what we extend to the interacting case. It uses results on Brownian motions (see \eqref{eq:Feynman-Kac_simple} and \eqref{eq:reflexion-princ-conlon} below).
	
	\begin{proof}[Proof of Lemma \ref{lemma:pointwise-WL_HF-exp}]
		The main argument is the use of the Feynman-Kac formula. It states that for any $x,y\in\R^d$ and any $t>0$
		\begin{equation}\label{eq:Feynman-Kac_simple}
			(4\pi t)^{d/2}e^{-t(-\Delta+(1/t)W)}(x,y)= (4\pi t)^{d/2}e^{-t\Delta}(x,y)  \int \exp\left(-\frac 1 t\int_0^t W(\beta(s))ds \right)d\mu_{x,y,t}(\beta).
		\end{equation}
		Here, $	e^{-t\Delta}(x,y)$ is the integral kernel of the propagator $e^{-t\Delta}$ of the heat equation
		\begin{equation*}
			e^{-t\Delta}(x,y)=\frac{e^{-\frac{1}{4t}\abs{x-y}^2}}{(4\pi t)^{d/2}},
		\end{equation*} 
		and $d\mu_{x,y,t}$ is the conditional Wiener measure on the continuous path $\beta:[0,t]\to\R^d$ such that $\beta(0)=x$ and $\beta(t)=y$ (see for instance \cite[Chap.2]{simon2005functionalint}). 
		When $x=y$, we can restrict our study to
		\begin{equation*}
			\int \exp\left(-\frac 1 t\int_0^t W(\beta(s))ds \right)d\mu_{x,x,t}(\beta).
		\end{equation*}
		Conlon's idea is to estimate this integral in the localized set of path $N_\delta$ defined for any $\delta>0$
		\begin{equation}\label{eq-def:N_delta-conlon}
			N_\delta=\{ \beta\in\CR([0,t],\R^d),\: \beta(0)=x=\beta(t),\quad \sup_{0\leq s\leq t}\abs{\beta(s)-x}\leq\delta  \}
		\end{equation}
		and its complementary $N_\delta^c$. We then write 
		\begin{align*}
			\int \exp&\left(-\frac 1 t\int_0^t W(\beta(s))ds \right)d\mu_{x,x,t}(\beta)
			\\&= 
			e^{-W(x)}d\mu_{x,x,t}(N_\delta)
			\\&\quad
			+e^{-W(x)}\int_{N_\delta}\left[\exp\left(-\frac 1t\int_0^t (W(\beta(s))-W(x)) ds\right)-1\right] d\mu_{x,x,t}(\beta)
			\\&\quad
			+\int_{N_\delta^c}\exp\left(-\frac 1t\int_0^t W(\beta(s))ds\right)d\mu_{x,x,t}(\beta)
			.
		\end{align*}
		On the one hand, one has the bounds
		\begin{equation*}
			e^{-W(x)}\leq e^{\normLp{W_-}{\infty}{}},\quad d\mu_{x,x,t}(N_\delta) \leq 1,
		\end{equation*}
		and for any $\delta>0$
		\begin{equation*}
			\lim_{t\to 0^+}d\mu_{x,x,t}(N_\delta)=1
			.
		\end{equation*}		
		On the other hand, denoting by $\omega_x(\delta,W)$ the modulus of continuity of $w$ at $x$
		\begin{equation*}
			\omega_x(\delta,W):=\sup_{\{y\in\R^d\: : \: \abs{x-y}\leq\delta\}}\abs{W(x)-W(y)},
		\end{equation*}
		and using that $d\mu_{x,x,t}(N_\delta)\leq 1$, one has
		\begin{align*}
			e^{-W(x)}&\int_{N_\delta}\left[\exp\left(-\frac 1t\int_0^t (W(\beta(s))-W(x)) ds\right)-1\right] d\mu_{x,x,t}(\beta)
			\\&\leq e^{\normLp{W_-}{\infty}{}}\int_{N_\delta}\left(\frac 1t\int_0^t \abs{W(\beta(s))-W(x)}ds\right)\exp\left(\frac 1t\int_0^t \abs{(W(\beta(s))-W(x))} ds\right) d\mu_{x,x,t}(\beta)
			\\&\leq e^{\normLp{W_-}{\infty}{}} \omega_x(\delta,W)e^{\omega_x(\delta,W)}
			.
		\end{align*}
		Furthermore, let us explain more precisely
		why
		there exists $C>0$ such that for any $t\in(0,1]$
		\begin{equation}\label{eq:reflexion-princ-conlon}
			\mu_{x,x,t}(N_\delta^c)\leq C e^{-\frac{\delta^2}{8t}}
			.
		\end{equation}
	The bound \eqref{eq:reflexion-princ-conlon} provided below is a bit more precise than \cite[Eq. (3.44)]{conlon1983semi} in term of dependance on the parameter $\delta$. However, it does not matter to have $\mu_{x,x,t}(N_\delta^c)=\OR(e^{-\alpha/ t})$ with $\alpha=\alpha_\delta>0$ regardless of the dependance in $\delta$. In fact, we eventually apply the limit $t\to 0^+$ before $\delta\to 0^+$.
	\begin{proof}
		Let is call $\{X_t\}_{t\geq 0}$ the Wiener process with values in the Euclidian space $\R^d$ and on a probability space with probability $\mathds{P}$. The continuous probability measure $\mathds{P}_x$  of the process that starts at $x\in\R^d$, is associated to the density $y\mapsto e^{-\frac t2\Delta}(x,y)=(2\pi t)^{d/2}e^{-\frac{1}{2t}\abs{y-x}^2}$ with respect to the Lebesgue measure. 
		The measure $d\mu_{x,x,t}$ corresponds to the law of the modified Brownian bridge $\{B_{2s}\}_{s\in[0,t]}$ which is the normalized Wiener process (admitting the centered Gaussian of covariance $2t$ as law) starting and finishing at the point $x$, i.e.\ with the conditioning $B_0=B_{2t}=x$
		\begin{equation*}\forall\Omega\subset\R^d,\quad
			\mu_{x,x,t}(\Omega)=\mathds{P}\left(\{B_{2s}\}_{0\leq s\leq t}\in\Omega\right)
			.
		\end{equation*}
		In our case, $\mu_{x,x,t}(N_\delta^c)$ is equal to $\mathds{P}\left(\sup_{0\leq s\leq t}\abs{B_{2s}-x}\geq \delta\right)$. This quantity is bounded by $\mathds{P}_x\left(\sup_{0\leq s\leq t}\abs{X_{2s}-x}\geq \delta\right)$.
		The reflection principle (see for instance \cite[Thm. 3.6.5 and Eq. (7.6')]{simon2005functionalint}) states that for any $t>0$,
		\begin{equation*}
			\mathds{P}_x\left(\sup_{0\leq s\leq t}\abs{X_{2s}-x}\geq \delta\right)\leq 2\mathds{P}_x(\abs{X_{2t}-x}\geq \delta).
		\end{equation*}
		Moreover,
		\begin{align*}
			\mathds{P}_x(\abs{X_{2t}-x}\geq\delta)
			&= \frac{1}{(2\pi t)^{d/2}}\int_{\abs{y}\geq \delta} e^{-\frac 1{4t}\abs{y}^2}dy
			\leq e^{-\frac{\delta^2}{8t}}\left(\frac{1}{(2\pi t)^{d/2}}\int_{\abs{y}\geq \delta} e^{-\frac 1{8t}\abs{y}^2}dy\right)
			\leq e^{-\frac{\delta^2}{8t}}.
		\end{align*}
		As a consequence, one has the uniform bound
		\eqref{eq:reflexion-princ-conlon} for any $x\in\R^d$.
	\end{proof}
		%
		Then, we can bound the last term
		\begin{equation*}
			\int_{N_\delta^c}\exp\left(-\frac 1t\int_0^t W(\beta(s))ds\right)d\mu_{x,x,t}(\beta) \leq e^{\normLp{W_-}{\infty}{}}\mu_{x,x,t}(N_\delta^c)
			\leq C e^{\normLp{W_-}{\infty}{}}  e^{-\frac{\delta^2}{8t}}.
		\end{equation*}
		Eventually,
		\begin{align*}
			(4\pi t)^{d/2}e^{-t(-\Delta+(1/t)W)}(x,x)
			&= e^{-W(x)}(1+o_{t\to 0^+}(1)) + \OR\left(\omega_x(\delta,W)e^{\omega_x(\delta,W)}\right) + 
			\OR\left(e^{-\frac{\delta^2}{8t}}\right)
			.
		\end{align*}
		We obtain the desired limit by letting first $t\to 0^+$, then $\delta\to 0^+$
		\begin{equation*}
			\lim_{t\to 0^+}(4\pi t)^{d/2}e^{t(-\Delta+(1/t)W)}(x,x)=e^{-W(x)}.
		\end{equation*}
		This ends the proof of Lemma \ref{lemma:pointwise-WL_HF-exp}.
	\end{proof}

	Let us explain now how Lemmata \ref{lemma:tauberian-thm} and \ref{lemma:pointwise-WL_HF-exp} imply the desired asymptotics.
	Recall that we have that $V$ is non-negative and that $E>0$. By writing $W:=V-E$ and $\gamma:=\hbar^{-2}$, the semiclassical pointwise Weyl law is equivalent to prove the pointwise limit for any $x\in\R^d$
	\begin{equation}\label{eq-demo:WL-without-int_mod}
		\lim_{\gamma\to +\infty}\gamma^{-d/2}\indicatrice{-\Delta+\gamma W\leq 0}(x,x)= \frac{\abs{B_{\R^d}(0,1)}}{(2\pi)^d}(-W(x))_+^{d/2},
	\end{equation}
	for any $W:\R^d\to\R$ continuous such that $W(x)\to+\infty$ as $\abs{x}\to+\infty$ and such that $W\geq -E$.  	
	
	Let us fix $x\in\R^d$. Let us write now this limit into the form \eqref{eq:tauberian-thm_res}. 
	To do so, we introduce the sequence of measures $\{m_\gamma\}_{\gamma\in(0,+\infty])} =\{m_\gamma[x]\}\}_{\gamma\in(0,+\infty])}$ associated to the non-decreasing function for any $\gamma>0$
	\begin{equation*}
		m_\gamma: u\in\R\mapsto \gamma^{-d/2}\indicatrice{-\Delta+\gamma W\leq \gamma u}(x,x)\in\R,
	\end{equation*}
	and to the function
	\begin{equation*}
		m_\infty: u\in\R\mapsto \frac d2\frac{\abs{B_{\R^d}(0,1)}}{(2\pi)^d}(u-W(x))_+^{d/2-1}\in\R.
	\end{equation*}
	In particular, one has for any $\gamma>0$
	\begin{equation*}
		dm_\gamma([-E,0])= m_\gamma(0)-m_\gamma(-E)=\gamma^{-d/2}\indicatrice{-\Delta+\gamma W\leq 0}(x,x)
		.
	\end{equation*}
	Furthermore, since $W\geq -E$, then $dm_\gamma(u)=0$  and $dm_\infty(u)=0$ for any $u\leq -E$ and any $\gamma>0$.
	We can write
	\begin{equation*}
		dm_\infty([-E,0]) =dm_\infty((-\infty,0]) = \frac d 2\frac{\abs{B_{\R^d}(0,1)}}{(2\pi)^d}\int_{-\infty}^0 (u-W(x))_+^{d/2-1}du.
	\end{equation*}
	If $W(x)\geq 0$, we have that $=\int_{-\infty}^0 (u-W(x))_+^{d/2-1}du = 0$.
	If $W(x)<0$, by the change of variable $v=-W(x)u$
	\begin{align*}
		\int_{-\infty}^0 (u-W(x))_+^{d/2-1}du &
		= (-W(x))_+^{d/2}\int_{-\infty}^0 (v+1)_+^{d/2-1}dv
		=\frac 2 d (-W(x))_+^{d/2}.
	\end{align*}
	As a consequence,
	\begin{equation*}
		dm_\infty([-E,0]) = \frac{\abs{B_{\R^d}(0,1)}}{(2\pi)^d}(-W(x))_+^{d/2}
		.
	\end{equation*}
	Let us now rewrite \eqref{eq:tauberian-thm_ass} by computing the Laplace transforms of $dm_\gamma$ and $dm_\infty$.
	Indeed, for any $\gamma>0$
	\begin{equation*}
		\int_{-E}^{+\infty} e^{-\alpha u}dm_\gamma(u)
		= \gamma^{-d/2}e^{-\frac\alpha\gamma(-\Delta+\gamma W)}(x,x).
	\end{equation*}
	Furthermore
	\begin{align*}
		\int_{-E}^{+\infty} e^{-\alpha u}dm_\infty(u)
		&= \frac d 2\frac{\abs{B_{\R^d}(0,1)}}{(2\pi)^d}  \int_\R e^{-\alpha u} (u-W(x))_+^{d/2-1}du
		\\&=  \frac d 2\frac{\abs{B_{\R^d}(0,1)}}{(2\pi)^d} \alpha^{-d/2}e^{-\alpha W(x)}\int_\R e^{-w}w_+^{d/2-1} du
		\\&= (2\pi)^{-d/2} \alpha^{-d/2}e^{-\alpha W(x)}.
	\end{align*}
	Since we want to prove \eqref{eq:tauberian-thm_ass} for any $\alpha>0$, by taking $t=\alpha/\gamma$ and replacing $W$ by $\alpha^{-1}W$, it is equicalent to prove that
	\begin{equation*}
		\lim_{t\to 0^+}(4\pi t)^{d/2}e^{-t\left(-\Delta+\frac 1t W\right)}(x,x)=e^{-W(x)}
		.
	\end{equation*}
	Since, this is what we have assumed, this concludes the proof of Theorem  \ref{thm:WL-without-int_conlon}.

	\subsection{Interacting case}
	
	Let us now deal with the case with interaction. We generalize now the previous asymptotics.
	
	\begin{lemma}\label{lemma:pointwise-WL_HF-h}
		Let $V:\R^d\to\R$ be a continuous function. Assume that we have the decomposition
		\begin{equation*}
			V=V_{\rm trap}+V_\hbar,
		\end{equation*}
		such that
		\begin{itemize}
			\item $V_{\rm trap}:\R^d\to\R$ satisfies Assumption \ref{cond:confining};
			\item $V_\hbar:\R^d\to\R$ is continuous for any $\hbar>0$, and such that there exists $V_0:\R^d\to\R$ continuous and bounded such that $V_\hbar$ converges uniformly to $V_0$ in all compacts of $\R^d$ as $\hbar\to 0$.
		\end{itemize}
	Let $E\in\R^d$.
		Then, one has for any $x\in\R^d$,
		\begin{equation*}
			\lim_{\hbar\to 0} \hbar^d\indicatrice{-\hbar^2\Delta+V_{\rm trap}+V_\hbar\leq E}(x,x) = \frac{\abs{B_{\R^d}(0,1)}}{(2\pi)^d}(E-V_{\rm trap}(x)-V_0(x))_+^{d/2}.
		\end{equation*}
	\end{lemma}

		\begin{proof}[Proof of Lemma \ref{lemma:pointwise-WL_HF-h}]
		Let us explain how the proof of Lemma \ref{lemma:pointwise-WL_HF-exp} can be adapted in this framework, with $W_\hbar=V_{\rm trap}+V_\hbar-E$ instead of $W=V_{\rm trap}-E$. Let us show that
		\begin{equation}\label{eq-demo:pointwise-WL_HF-h}
			\lim_{t\to 0^+}\int \exp\left(-\frac 1t\int_0^t W_t(\beta(s)) ds\right) d\mu_{x,x,t}(\beta) = e^{-W_0(x)}.
		\end{equation}
		Recall that $t=\hbar^2/\alpha$. By assumption, one has the uniform limit $W_t\to W_0:=V_{\rm trap}+V_0-E$ on any compact of $\R^d$, as $t\to 0^+$ . On one hand, this implies the same asymptotics for $e^{-W_t}\to e^{-W_0}$ as $t\to 0^+$.
		On the other hand, for any continuous path $\beta:[0,t]\to\R^d$ such that $\beta(0)=\beta(t)=x$ and any $s\in[0,t]$, we have
		\begin{align*}
			\abs{W_t(\beta(s))-W_t(x)}
			&\leq\abs{W_0(\beta(s))-W_0(x)}+\abs{W_t(\beta(s))-W_0(\beta(s))}+\abs{W_t(x)-W_0(x)}
			\\&\leq \omega_x(\delta,W_0)+2\normLp{W_t-W_0}{\infty}{\left(B_{\R^d}(x,\delta)\right)}
			.
		\end{align*}
		Thus,
		\begin{align*}
			e^{-W_t(x)}&\int_{N_\delta}\left[\exp\left(-\frac 1t\int_0^t W_t((\beta(s))-W_t(x)) ds\right)  -1 \right]d\mu_{x,x,t}(\beta)
			\\&\leq \left(\abs{e^{-W_t(x)}-e^{-W_0(x)}}+e^{\normLp{(W_0)_-}{\infty}{}}\right)\times
			\\&\quad\times \left(\omega_x(\delta,W_0)+2\normLp{W_t-W_0}{\infty}{\left(B_{\R^d}(x,\delta)\right)}\right)e^{\omega_x(\delta,W_0)+2\normLp{W_t-W_0}{\infty}{\left(B_{\R^d}(x,\delta)\right)}}.
		\end{align*}
		The last term is bounded by
		\begin{equation*}
			\int_{N_\delta^c}\exp\left(-\frac 1t\int_0^t W_t(\beta(s))ds\right)d\mu_{x,x,t}(\beta) 
			\leq \left(\abs{e^{-W_t(x)}-e^{-W_0(x)}}+e^{\normLp{(W_0)_-}{\infty}{}}\right) e^{-\alpha/ t}.
		\end{equation*}
		Once again, by $t\to 0^+$, then $\delta\to 0^+$, we obtain \eqref{eq-demo:pointwise-WL_HF-h} and then
		\begin{equation*}
			\lim_{t\to 0^+}(4\pi t)^{d/2}e^{t(-\Delta+(1/t)W_t)}(x,x) = e^{-W_0(x)}.
		\end{equation*}
	By Lemma \ref{lemma:tauberian-thm}, one obtains the desired limit.
	\end{proof}
	
	Let us finally prove the pointwise Weyl law in the case with interactions.

	Note that if one can prove that we have the strong limit $\hbar^d\rho_{\gamma_\hbar}\to\rho_{\rm TF}$ in $L^p(\R^d)$, for some $p\in[1,\infty)$, one would have the pointwise convergence almost everywhere on $\R^d$, up to a subsequence. In this case, Lemma \ref{lemma:pointwise-WL_HF-h} would not be bery useful. However, this strong limit is not obvious and all we have is a weak convergence.
		
	Moreover, it turns out that the Weyl law will not change if we add an exchange term perturbation.
	\begin{lemma}\label{lemma:pointwise-WL_kill-exchange}
		Let $V:\R^d\to\R$ which satisfies Assumption \ref{cond:confining} and $E\in\R$. Let $w:\R^d\to\R$ even function such that
		\begin{itemize}
			\item $w\in L^{1+d/2}(\R^d)+L^\infty_\varepsilon(\R^d)$ with Assumption \ref{cond:w-Dterm}, for $d\geq 2$
			\item $w\in L^{1+d/2}(\R^d) \cap L^2(\R^d)+L^\infty_\varepsilon(\R^d)$ with Assumption \ref{cond:w-Dterm}, for $d\geq 1$,
			\item or alternatively $w\in L^1(\R^d)\cap L^{1+d/2}(\R^d)$ with Assumption \ref{cond:w-Dterm-d=12} for $d=2$,
			\item  $w\in L^1(\R^d)\cap L^{1+d/2}(\R^d) \cap L^2(\R^d)$ with Assumption \ref{cond:w-Dterm-d=12} for $d=1,2$.
		\end{itemize}
		Then, for any $\{\gamma_\hbar\}_\hbar\subset\KR$ such
		\begin{equation*}\forall\hbar\in(0,\hbar_0],\quad
		\tr((-\hbar^2\Delta+V+1)\gamma_\hbar)\leq C\hbar^{-d},
		\end{equation*}
		one has that for any $x\in\R^d$
		\begin{align*}
		\lim_{\hbar\to 0}&
		\hbar^d\indicatrice{-\hbar^2\Delta+V+\hbar^d\rho_{\gamma_\hbar}\ast w-\hbar^dX_w(\gamma_\hbar) \leq E}(x,x) \\&=\lim_{\hbar\to 0}\hbar^d
		\indicatrice{-\hbar^2\Delta+V+\hbar^d\rho_{\gamma_\hbar}\ast w\leq E}(x,x).
		\end{align*}
	\end{lemma}
	\begin{proof}[Idea of proof of Lemma \ref{lemma:pointwise-WL_kill-exchange}]
		The proof is similar to the same as the one of Theorem \ref{thm:WL-without-int_conlon}. We only change the measure $m_\gamma$ by adding the exchange term into it. Implicitly $W$ can depend on $t$ as in the proof of Lemma \ref{lemma:pointwise-WL_HF-h} (such that the sequence of functions $W_t$ converges uniformly in any compact towards a continuous and trapping function $W_0$). The crucial point is to prove
		an equivalent of \cite[(3.33)]{conlon1983semi}. We prove that for any $x\in\R^d$
		\begin{equation}\label{eq-demo:exp-kill-exchange}
		\abs{e^{-t\big(-\Delta+(1/t) W(x)-t^{d/2-1}X_w(\gamma_t)\big)}(x,x)-e^{-t\big(-\Delta+(1/t) W(x)\big)}(x,x)} = \OR_{t\to 0^+}(1)
		.
		\end{equation}
		Then, we conclude by the Tauberian theorem (Lemma  \ref{lemma:tauberian-thm}) and Lemma \ref{lemma:pointwise-WL_HF-exp}. Let us now prove the limit \eqref{eq-demo:exp-kill-exchange}.
		We provide a proof of \eqref{eq-demo:exp-kill-exchange} in the end of this section. This version is different that the one in \cite{conlon1983semi} that we have corrected the arguments to make it work.
	\end{proof}

\begin{proof}[Proof of Theorem \ref{thm:pointwise-WL_HF}]
	Let $\{\gamma_\hbar\}_{\hbar>0}$ a sequence of minimizers of the $\hbar$-Hartree-Fock energy $\ER_{\hbar,V,w}^{\rm HF}$ such that for any $\hbar>0$, the operator $\gamma_\hbar$ satisfies the relation \eqref{eq:pointwise-WL-EL}.
	Our goal is to apply Lemma \ref{lemma:pointwise-WL_HF-h} to $V_{\rm trap}=V-E$ and $V_\hbar=\hbar^d\rho_{\gamma_\hbar}\ast w$. Let us check first why the hypothesis of the Lemma hold.
		The sequence $\{\hbar^d\rho_{\gamma_\hbar}\}_\hbar$ is uniformly bounded is $ L^1(\R^d)\cap L^{1+2/d}(\R^d)$, $\rho_{\rm TF}\ast w\in L^1(\R^d)\cap L^{1+2/d}(\R^d)$.
		The functions $\hbar^d\rho_{\gamma_\hbar}\ast w,\rho_{\rm TF}\ast w\in L^\infty$ are convolutions of functions in $L^p$ with conjugated exponents, and are thus in $L^\infty(\R^d)$ and continuous on $\R^d$.
	
	Given the assumptions on $V$ and $w$, and the nature of $\gamma_\hbar$, one has by Lemma \ref{lemma:trace_bd_h^d-almost-min}
	\begin{equation*}\forall \hbar>0,\quad
	\hbar^d\tr((-\hbar^2\Delta+V+1)\gamma_\hbar)\leq C
	.
	\end{equation*}	
	By Theorem \ref{thm:weak-WL_HF}, up to a subsequence $\hbar^d\rho_{\gamma_\hbar}\rightharpoonup \rho_{\rm TF}$ weakly in $L^1\cap L^{1+2/d}(\R^d)$ as $\hbar\to 0$. This implies the convergence $\hbar^d\rho_{\gamma_\hbar}\ast w\to \rho_{\rm TF}\ast w$ almost everywhere on $\R^d$.
	
		By the H\"older inequality and Eq. \eqref{cond:equicont_conlon}, it turns out that there exists $C>0$ such that for any $\hbar>0$
		\begin{equation*}
		\normLp{\nabla V_\hbar}{\infty}{}\leq \normLp{\hbar^d\rho_{\gamma_\hbar}}{1}{}\normLp{\nabla w}{\infty}{} <C
		\text{ or } \normLp{\nabla V_\hbar}{1+d/2}{}\leq \normLp{\hbar^d\rho_{\gamma_\hbar}}{1+2/d}{}\normLp{\nabla w}{\infty}{} <C
		.
		\end{equation*}
		By the Ascoli Theorem, up to a subsequence, one has $\hbar^d\rho_{\gamma_\hbar}\ast w\to\rho_{\rm TF}\ast w$ uniformly on all compacts of $\R^d$. 
		Applying Lemma \ref{lemma:pointwise-WL_kill-exchange} and Lemma \ref{lemma:pointwise-WL_HF-h}, one has the pointwise convergence (up to the same subsequence)
		\begin{align*}\forall x\in\R^d,\quad
			\lim_{\hbar\to 0}\hbar^d\rho_{\gamma_\hbar}(x)
				&=	\lim_{\hbar\to 0}\hbar^d\indicatrice{-\hbar^2\Delta+V-\hbar^d\rho_{\gamma_\hbar}\ast w\leq E}(x,x)
			\\&=(E-V(x)-\rho_{\rm TF}\ast w(x))_+^{d/2}.
		 \end{align*}
	Furthermore, since $\rho_{\rm TF}$ minimizes Thomas-Fermi energy, it also satisfies the Thomas-Fermi equation \eqref{eq:TF}. This concludes the proof of Theorem \ref{thm:pointwise-WL_HF}.
\end{proof}

\begin{proof}[Proof of the bound \eqref{eq-demo:exp-kill-exchange}]
	For any $t>0$, let us denote by 
	\begin{align*}
	R(t)&:= e^{-t\big(-\Delta+(1/t) W(x)-t^{d/2-1}X_w(\gamma_t)\big)}-e^{-t\big(-\Delta+(1/t) W(x)\big)}
	\\&=e^{-t\big(-\Delta+ (1/\hbar^2) W(x)-\hbar^{d-2} X_w(\gamma_\hbar)\big)}-e^{-t\big(-\Delta+(1/\hbar^2) W(x)\big)},
	\end{align*}
	with $t=\hbar^2$.
	By Duhamel formula, for any $t\in(0,1]$
	\begin{align*}
	R(t)&= \hbar^{d-2} \int_{0}^t e^{-(t-\tau)(-\Delta+(1/\hbar^2)W)}X_w(\gamma_\hbar)e^{-\tau\big(-\Delta+(1/\hbar^2) W(x)-\hbar^{d-2}X_w(\gamma_\hbar)\big)}d\tau.
	\end{align*}
	We reintegrate this formula, so that $R(t)=\sum_{n\geq1} R_n(t)$ where for any $n\geq 1$,
	\begin{align*}
	R_n(t) =
	& \big(\hbar^{d-2}\big)^n\int_{0}^{t}e^{-(t-\tau_1)(-\Delta+(1/\hbar^2)
		W)}X_w(\gamma_\hbar)\int_0^{\tau_1}e^{-(\tau_1-\tau_2)(-\Delta+(1/\hbar^2) W)}X_w(\gamma_\hbar)\int_0^{\tau_2}\cdots\times
	\\&\times
	X_w(\gamma_\hbar)\int_0^{\tau_{n-1}} e^{-(\tau_{n-1}-\tau_n)(-\Delta+(1/\hbar^2) W)}X_w(\gamma_\hbar) e^{-\tau_n(-\Delta+(1/\hbar^2) W)} d\tau_n d\tau_{n-1} \cdots d\tau_1
	\\= & \big(t^{d/2}\big)^n \int_0^1 s_1^{n-1} e^{-t(1-s_1)(-\Delta+(1/t) W)} X_w(\gamma_t)\int_0^1 s_2^{n-2} e^{-ts_1(1-s_2)(-\Delta+(1/t) W)}\int_0^1\cdots\times
	\\&\times
	X_w(\gamma_t)\int_0^1 e^{-ts_1\ldots s_{n-1}(1-s_n)(-\Delta+(1/t) W)}X_w(\gamma_t) e^{-ts_1\ldots s_n(-\Delta+(1/t) W)} ds_n ds_{n-1} \cdots ds_1
	.
	\end{align*}
	As in \cite{conlon1983semi}, the main idea is to prove that exists $C,\tilde{C}>0$ such that for any $n\geq 1$ and any $t\in(0,1)$, one has
	\begin{equation}\label{eq-demo:exp-kill-exchange_Rn-Linfty}
		\sup_{x\in\R^d}\abs{R_n(t)(x,x)}\leq \tilde{C}\frac{C^n}{(n-1)!}.
	\end{equation}
	\begin{rmk}\label{rmk:proof-exp-kill-exchange_compColon}
		We provide a different proof of this bound of the Conlon's one, which does not rely on the induction relation between $R_n(t)$ and $R_{n-1}(t)$,
		on bilinear estimates of the type $\prodscal{u}{X_w(\gamma_t) v}_{L^2}\leq C_t\normLp{u}{r_1}{}\normLp{v}{r_2}{}$ and on bounds of the $L^r$ norm of $R_n(t)(\cdot,x)$ for fixed $x\in\R^d$. The equivalent statement of our Lemma \ref{lemma:exp-kill-exchange_Rn} was given by \cite[Lemma 3.2]{conlon1983semi}. However, this proof cannot be adapted in our case since it does not work for any dimension. Furthermore, it uses the $L^\infty$ norm of $R_n(t)(\cdot,x)$, which is not allowed by the rest of the proof (for instance between the constraints \cite[(3.15)]{conlon1983semi} and \cite[(3.16)]{conlon1983semi} for the case $n=1$). It also seems to us that it is incomplete because of a divergent integral on $[0,1]$ on the bound \cite[(3.15)]{conlon1983semi} for the case $n=1$. In this paper, we rather estimate $L^r$ bound of the density of the operators $R_n(t)$.
	\end{rmk}
	The bound \eqref{eq-demo:exp-kill-exchange_Rn-Linfty} is a consequence of the following lemma and the $\normLp{\rho_{\gamma_t}}{q}{(\R^d)}=\OR(t^{d/2})$ which holds for any $q=1$ or $q=1+2/d$, in particular for any $q\in[1,1+2/d]$ by interpolation. 
	\begin{lemma}\label{lemma:exp-kill-exchange_Rn}
		Let $q \geq 1$ and $p\geq 2$ such that $p> d/4$. For $w\in L^p(\R^d)$ and any $r\geq q$ such that
		\begin{equation}\label{cond:kill-exch_well-def-int}
		\frac 1p+\frac 1q-\frac 1r< \frac 4d.
		\end{equation}
		there exists $C>0$ such that for any $n\geq 1$,
		\begin{equation*}
			\norm{R_n(t)(x,x)}_{L^r_x(\R^d)}\leq \frac{e^{\normLp{W_-}{\infty}{}}}{(n-1)!}  t^{-\tfrac d2\left(\tfrac 1p+\tfrac 1q-\tfrac 1r\right)} \left(C\normLp{w}{p}{(\R^d)}  t^{d/2}\right)^n \normLp{\rho_{\gamma_t}}{q}{(\R^d)}\normLp{\rho_{\gamma_t}}{p'}{(\R^d)}^{n-1}
		.
		\end{equation*}
	\end{lemma}
	\begin{rmk}\label{rmk:proof-exp-kill-exchange}
		It turns out that when the assumptions of Lemma \ref{lemma:exp-kill-exchange_Rn} are satisfied and when $\normLp{\rho_{\gamma_t}}{q}{}=\OR_{t\to 0_+}(t^{-d/2})$, the norm $L^r$ of the density of $R_n(t)$ is $\OR_{t\to 0_+}(t^{d/2})$ if and only if $t^{-\frac d2\left(\frac 1p+\frac 1q-\frac 1r\right)}=\OR_{t\to 0_+}(t^{-d/2})$, i.e.\
		\begin{equation}\label{cond:kill-exch_small-err}
		\frac 1p+\frac 1q-\frac 1r\leq 1.
		\end{equation}
		Ideally, it is tempting to deduce \eqref{eq-demo:exp-kill-exchange_Rn-Linfty} for from $p>d/2$ and $r=q=+\infty$. However, it is not obvious that $\normLp{\rho_{\gamma_t}}{\infty}{}=\OR_{t\to 0_+}(t^{-d/2})$. It is what we will deduce from several iterations of Lemma \ref{lemma:exp-kill-exchange_Rn} for suitable sequences $\{q_j\}_j$ and $ \{r_j\}_j$ by using the fact the definition $\gamma_t=\indicatrice{-t\Delta+W-t^{d/2}X_w(\gamma_t)\leq 0}$.
	\end{rmk}
	\begin{rmk}\label{rmk:exp-kill-exchange_Rn}
		The assumption $p\geq 2$ is purely technical and is due to the proof of Lemma \ref{lemma:exp-kill-exchange_Rn}. This is satisfied for case $p=1+d/2$ for $d\geq 2$, but for $d=1$ this is unfornately non covered. This is the reason that we add the assumption $w\in L^{1+d/2}\cap L^2(\R^d)+L^\infty_\varepsilon(\R^d)$ in the statement of Lemma \ref{lemma:pointwise-WL_kill-exchange} and of Theorem \ref{thm:pointwise-WL_HF}. Besides, this hypothesis has the merit of including Coulomb potential in dimension $d\geq 2$.
	\end{rmk}
	\paragraph{\bf Proof of Lemma \ref{lemma:exp-kill-exchange_Rn}.}
	Let $x\in\R^d$ and let $t\in(0,1]$.
	Note that for almost every $z,y\in\R^d$, 
	\begin{equation*}
		\abs{X_w(\gamma_t)}\leq \sqrt{\rho_{\gamma_t}(z)}\abs{w(z-y)}\sqrt{\rho_{\gamma_t}(y)}.
	\end{equation*}
	Furthermore, for any $s\in(0,1]$, one has by Kato-Trotter formula
	\begin{equation}\label{eq:Kato-Trotter-gauss}
		e^{-s(-t\Delta+W)}(x,y) \leq e^{-st\Delta}(x,y)e^{s\normLp{W_-}{\infty}{}}
		.
	\end{equation}
	Let us denote by $g_{ts}$ the Gaussian $g_{st}= (4\pi st)^{-d/2}e^{-\frac 1{4ts}{\abs{\cdot}^2}}$.
	In particular, for any $1\leq b\leq\infty$
	\begin{equation}\label{eq-demo:bounds-Lq-gaussian}
		\normLp{g_{ts}}{b}{(\R^d)}\leq C(ts)^{-\tfrac d{2b'}}.
	\end{equation}
	One has
	\begin{equation}\label{eq-demo:exp-kill-exchange_first-bound-Rn}
	\begin{split}
		\abs{R_n(t)(x,x)} 
		&\leq  e^{\normLp{W_-}{\infty}{}} (t^{d/2})^n\int_0^1\ldots\int_0^1
		\left(	\int_{\R^d}\ldots 	\int_{\R^d} g_{t(1-s_1)}(x-y_1)\sqrt{\rho_{\gamma_t}(y_1)}
		\right.\\&\qquad
		\abs{w(y_1-y_2)} \sqrt{\rho_{\gamma_t}(y_2)}  g_{t s_2(1-s_1)}(y_2-y_3) \ldots \abs{w(y_{2n-1}-y_{2n})}
		\\&\left.
		\qquad  \sqrt{\rho_{\gamma_t}(y_{2n})}g_{ts_1\ldots s_n}(y_{2n}-x) \quad dy_{2n}\ldots dy_1 \right) \:
		ds_n \: s_{n-1} ds_{n-1}\: \ldots\: s_1^{n-1} ds_1.
	\end{split}
	\end{equation}
	For $n=1$, by the Young inequality
	\begin{align*}
	\abs{R_1(t)(x,x)} 
	&\leq C  e^{\normLp{W_-}{\infty}{}} t^{d/2} \int_0^1 \norm{g_{t(1-s_1)}(x-y_1)\sqrt{\rho_{\gamma_t}(y_1)}}_{L^{(2p)'}_{y_1}} 
	\\&\qquad\normLp{w}{p}{}  \norm{\sqrt{\rho_{\gamma_t}(y_2)} g_{t(1-s_1)}(y_2-x)}_{L^{(2p)'}_{y_2}} ds_1.
	\end{align*}
	Notice that $\norm{g_{ts}(x-y)\sqrt{\rho_{\gamma_t}(y)} }_{L^{(2p)'}_y} = \big((g_{ts}^{(2p)'}\ast\rho_{\gamma_t}^{(2p)'/2})(x)\big)^{1/(2p)'}$. Then, by the Young inequality applied to the exponents $r/b\geq 1$ such that $2r/(2p)' , b/(2p)', 2q/(2p)'\geq 1$ and such that $1+\frac{(2p)'}{2r}=\frac{(2p)'}{b}+\frac{(2p)'}{2q}$, one has
	\begin{align*}
		\norm{ g_{ts}(x-y)\sqrt{\rho_{\gamma_t}(y)}}_{L^{2r}_xL^{(2p)'}_y} 
		&
		=  \normLp{(g_{ts}^{(2p)'}\ast\rho_{\gamma_t}^{(2p)'/2}}{2r/(2p)'}{(\R^d)}
		\lesssim \normLp{g_{ts}}{b}{(\R^d)}\normLp{\sqrt{\rho_{\gamma_t}}}{2q}{(\R^d)}
		\\&\lesssim (ts)^{-\frac d{2b'}}\normLp{\rho_{\gamma_t}}{q}{(\R^d)}^{1/2}
		.
	\end{align*}
	Notice the conditions $2r/(2p)'\geq 1$ and $2q/(2p)'\geq$ are always satisfied and that
	\begin{equation}\label{cond:kill-exch_young}
	\begin{split}
		1+\frac{(2p)'}{2r}=\frac{(2p)'}{b}+\frac{(2p)'}{2q} & \quad\Longleftrightarrow\quad \frac 1r = \frac 1p+\frac 1q-\frac 2{b'},
		\\
		\frac b{(2p)'}\geq 1 & \quad\Longleftrightarrow\quad \frac 2{b'}\geq \frac 1p,\quad \text{ this implies that } r\geq q, 
		\\ \frac 1r\geq 0 & \quad\Longleftrightarrow\quad \frac 2{b'}\leq \frac 1p+\frac 1q.
	\end{split}
	\end{equation}
	By the H\"older inequality, we deduce that
	\begin{align*}
		\norm{R_1(t)(x,x)}_{L^r_x(\R^d)} 
		&\leq C e^{\normLp{W_-}{\infty}{}}\normLp{w}{p}{(\R^d)} t^{\tfrac d2-\tfrac{d}{b'}}\normLp{\rho_{\gamma_t}}{q}{(\R^d)} \int_0^1 (1-s)^{-\tfrac d{2b'}}s^{-\tfrac d{2b'}}ds
		.
	\end{align*}
	The last integral is finite if and only if $\frac d{2b'}<1$, which is the condition \eqref{cond:kill-exch_well-def-int}.
	When $p>1$, the conditions \eqref{cond:kill-exch_young} and \eqref{cond:kill-exch_well-def-int} can be satisfied in the same time.
	\\\\
	Let $n\geq 2$. By the Young inequality in \eqref{eq-demo:exp-kill-exchange_first-bound-Rn}
	\begin{align*}
	\abs{R_n(t)(x,x)} &\lesssim  e^{\normLp{W_-}{\infty}{}} (t^{d/2})^n 
		\int_0^1\ldots \int_0^1 d s_n \: s_{n-1} ds_{n-1}\ldots s_1^{n-1}ds_1
		\\&\qquad
		\norm{\int_{\R^d}  g_{t(1-s_1)}(x-y_1)\sqrt{\rho_{\gamma_t}(y_1)}w(y_1-y_2)\sqrt{\rho_{\gamma_t}(y_2)}dy_1}_{L^2_{y_2}} 
		\\&\qquad
		\normLp{g_{ts_1(1-s_2)}}{1}{(\R^d)}\norm{Q_{t,s_1,\ldots,s_n}^{[n-2]}(y_3,x)}_{L^2_{y_3}} 
		,
	\end{align*}
	where
	\begin{align*}
		Q_{,t,s_1,\ldots,s_n}^{[n-2]}(y_3,x) &= \begin{cases}\displaystyle
		\sqrt{\rho_{\gamma_t}(y_3)}\int_{\R^d} w(y_3-y_4)\sqrt{\rho_{\gamma_t}(y_4)} g_{ts_1s_2}(y_4-x) dy_4&\text{if}\ d=2,\\
		\\
		\displaystyle
		\sqrt{\rho_{\gamma_t}}(y_3)w(y_3-y_4) \int_{\R^d}\sqrt{\rho_{\gamma_t}}(y_4)  \int_{\R^d} g_{ts_1s_2(1-s_3)}(y_4-y_5) \quad dy_4 &\text{if}\ d\geq 3.
		\\\quad\displaystyle \qquad\qquad\qquad \ldots
		\\\displaystyle
		\sqrt{\rho_{\gamma_t}(y_{2n})} \int_{\R^d}w(y_{2n-1}- y_{2n})\sqrt{\rho_{\gamma_t}(y_{2n})}g_{t s_1\ldots s_n}(y_{2n}-x)  \quad dy_{2n-1} dy_{2n}
		\end{cases}
	\end{align*}
	When $n\geq 3$, we iterate the Young inequalities so that the gaussians are in the $L^1$ norm
	\begin{align*}
		&\norm{Q_{t,s_1,\ldots,s_n}^{[n-2]}(y_3,x)}_{L^2_{y_3}}
		\\&\quad
		\lesssim  \norm{ \sqrt{\rho_{\gamma_t}(y_3)} \norm{w(y_3-y_4)\sqrt{\rho_{\gamma_t}(y_4)}}_{L^2_{y_4}}  \normLp{g_{ts_1 s_2(1-s_3)}}{1}{} \norm{Q_{t,s_1,\ldots,s_n}^{[n-3]}(y_5,x)}_{L^2_{y_5}}}_{L^2_{y_3}}
		\\&\quad\qquad= \norm{ \sqrt{\rho_{\gamma_t}(y_3)}w(y_3-y_4)\sqrt{\rho_{\gamma_t}(y_3)}  }_{L^2_{y_3,y_4}} \normLp{g_{ts_1 s_2(1-s_3)}}{1}{} \norm{Q_{t,s_1,\ldots,s_n}^{[n-3]}(y_5,x)}_{L^2_{y_5}}
		\\&\quad
		\lesssim \prod_{j=1}^{n-2}\norm{\sqrt{\rho_{\gamma_t}(y_{2j+1})}w(y_{2j+1}-y_{2j+2})\sqrt{\rho_{\gamma_t}(y_{2j+2})}}_{L^2_{y_{2j+1},y_{2j+2}}}\normLp{g_{ts_1 \ldots s_{2j+1}(1-s_{2j+2})}}{1}{}
		\\&\quad\qquad
		\norm{
		\sqrt{\rho_{\gamma_t}(y_{2n-1})} \int_{\R^d}w(y_{2n-1}- y_{2n})\sqrt{\rho_{\gamma_t}(y_{2n})}g_{t s_1\ldots s_n}(y_{2n}-x)   dy_{2n} }_{L^2_{y_{2n-1}}}
	.
	\end{align*}
	Then, for any $n\geq 2$ 
	\begin{align*}
		&\abs{R_n(t)(x,x)} 
		\lesssim 	e^{\normLp{W_-}{\infty}{}} (t^{-d/2})^n \int_0^1\ldots \int_0^1 d s_n \: s_{n-1} ds_{n-1}\ldots s_1^{n-1}ds_1
		\\&\qquad 
		\norm{\int_{\R^d}  g_{t(1-s_1)}(x-y_1)\sqrt{\rho_{\gamma_t}(y_1)}w(y_1-y_2)\sqrt{\rho_{\gamma_t}(y_2)}dy_1}_{L^2_{y_2}} \normLp{g_{ts_1(1-s_2)}}{1}{(\R^d)} 
		\\&\qquad \prod_{j=1}^{n-2}\norm{\sqrt{\rho_{\gamma_t}(y_{2j+1})}w(y_{2j+1}-y_{2j+2})\sqrt{\rho_{\gamma_t}(y_{2j+2})}}_{L^2_{y_{2j+1},y_{2j+2}}}\normLp{g_{ts_1 \ldots s_{2j+1}(1-s_{2j+2})}}{1}{}
		\\&\qquad
		\norm{
			\sqrt{\rho_{\gamma_t}(y_n)} \int_{\R^d}w(y_{2n-1}- y_{2n})\sqrt{\rho_{\gamma_t}(y_{2n})}g_{t s_1\ldots s_n}(y_{2n}-x)   dy_{2n} }_{L^2_{y_{2n-1}}}
		.
	\end{align*}
	By the Young inequality, for any $s\in(0,1]$ and any $x\in\R^d$
		\begin{align*}
		\norm{\sqrt{\rho_{\gamma_t}(y_2)}\int_{\R^d}  g_{ts}(x-y_1)\sqrt{\rho_{\gamma_t}(y_1)}w(y_1-y_2)dy_1}_{L^2_{y_2}}
		\lesssim \normLp{g_{ts}(x-\cdot)\sqrt{\rho_{\gamma_t}}}{(2p)'}{}\normLp{w}{p}{}\normLp{\sqrt{\rho_{\gamma_t}}}{2p'}{}.
	\end{align*}
	As well, by the Young inequality applied to $w^2\in L^{p/2}(\R^d)$ and $\rho_{\gamma_t}\in L^{p'}(\R^d)$, for any $p\geq 2$
	\begin{align*}
		\norm{\sqrt{\rho_{\gamma_t}(y_j)}w(y_j-y_{j+1})\sqrt{\rho_{\gamma_t}(y_{j+1})}}_{L^2_{y_j,y_{j+1}}}
		&=\normLp{\rho_{\gamma_t}(w^2\ast\rho_{\gamma_t})}{1}{}^{1/2}
		\\&\lesssim \normLp{w}{p}{}\normLp{\rho_{\gamma_t}}{p'}{}.
	\end{align*}
	We deduce with the same H\"older and Young inequalities in the case $n=1$ that
	\begin{align*}
		&\norm{R_n(t)(x,x)}_{L^r_x(\R^d)}\leq  e^{\normLp{W_-}{\infty}{}} \big(Ct^{d/2}\normLp{w}{p}{(\R^d)}\big)^n \normLp{\rho_{\gamma_t}}{q}{(\R^d)}\normLp{\rho_{\gamma_t}}{p'}{(\R^d)}^{n-1}
		\\&\qquad  \int_0^1\ldots \int_0^1  \normLp{g_{t(1-s_1)}}{b}{}\normLp{g_{ts_1\ldots(1-s_n)}}{b}{(\R^d)}\prod_{j=2}^{n-1}\normLp{g_{ts_1\ldots s_j(1-s_{j+1})}}{1}{(\R^d)}\quad  s_1^{n-1}ds_1 \ldots d s_n.
	\end{align*}
	The term in the last line is equal to $I_n t^{-d/b'}$, with
	\begin{align*}
		I_n &:= \left(\int_0^1 (1-s_1)^{-\tfrac d{2b'}}s_1^{n-1-\tfrac d{2b'}} ds_1\right)\left(\int_0^1s_2^{n-2-\tfrac d{2b'}} ds_2\right)\ldots \left(\int_0^1s_n^{-\tfrac d{2b'}} ds_n\right).
	\end{align*}
	The quantity $I_n$ is finite since we have the assumption \eqref{cond:kill-exch_well-def-int}. Let us bound it with respect to $n$. We write it with the Beta function defined on $(\R_+^*+i\R)^2$ by
	\begin{equation*}
		B(z,\tilde{z}):=\int_0^1 t^{z-1}(1-t)^{\tilde{z}-1} dt
	\end{equation*}
	Then, we use relation with the Euler Gamma function (see for instance \cite{olver1997asymptotics}[Chap.2])
	\begin{equation*}
		B(z,\tilde{z})=\frac{\Gamma(z)\Gamma(\tilde{z})}{\Gamma(z+\tilde{z})}.
	\end{equation*}
	Then, using the condition \eqref{cond:kill-exch_well-def-int}, that $\Gamma$ is non-decreasing and that it satisfies the property $\Gamma(1+z)=z\Gamma(z)$, one has
	\begin{align*}
	I_n&
		= B\left(n-\frac d{2b'},1-\frac d{2b'}\right)\frac{1}{n-1-\frac d{2b'}}\ldots\frac{1}{1-\frac d{2b'}}
		=\frac{\Gamma\left(n-\frac d{2b'}\right)\Gamma\left(1-\frac d{2b'}\right)}{\Gamma\left(n-1+2\left(1-\frac d{2b'}\right)\right)}\frac 1{\Gamma\left(n-\frac d{2b'}\right)}
	\\&= \frac{\Gamma\left(1-\frac d{2b'}\right)}{\Gamma\left(n-1+2\left(1-\frac d{2b'}\right)\right)} \leq \frac{1}{\Gamma(n-1)} = \frac 1{(n-1)!}\:
	.
	\end{align*}
	That proves Lemma \ref{lemma:exp-kill-exchange_Rn}.
	\paragraph{\bf Proof of the bound \eqref{eq-demo:exp-kill-exchange_Rn-Linfty} with Lemma \ref{lemma:exp-kill-exchange_Rn}.}
	The idea is to iterate Lemma \ref{lemma:exp-kill-exchange_Rn} a finite number of time $k\in\N$ so that $r_1<r_2<\ldots <r_k=\infty$, $q_{j+1}=r_j$ and $\normLp{\rho_{\gamma_t}}{q_j}{}=\OR(t^{d/2})$ in any step $1\leq j\leq k$. In order to have \eqref{cond:kill-exch_well-def-int} and \eqref{cond:kill-exch_small-err}, one should have for $r_k=+\infty$, the inequality $\frac 1{q_k}+\frac 1{p}<\min(1,d/4)$. We iterate the bounds as long as
	\begin{equation}\label{cond:exp-kill-exchange_nonstop}
	\frac 1{q_j}+\frac 1{p}\geq \min(1,d/4) .
	\end{equation}
	Before providing an explicit sequence of exponents, let us explain why one can take $q_{j+1}=r_j$ if this bound $\normLp{\rho_{\gamma_t}}{q_j}{}=\OR(t^{-d/2})$ is true for the step $j\in\N$. In this case, if \eqref{cond:kill-exch_well-def-int} and \eqref{cond:kill-exch_small-err} hold, Lemma \ref{lemma:exp-kill-exchange_Rn} implies $\norm{R(t)(x,x)}_{L^{r_j}_x}=\OR(t^{-d/2})$. In particular, since $\gamma_t = \indicatrice{-t\Delta+W-t^{d/2}X_w(\gamma_t)\leq 0}$, then for any $x\in\R^d$
	\begin{align*}
	\rho_{\gamma_t}(x)&\leq \abs{e^{-t\Delta+W-t^{d/2}X_w(\gamma_t)}(x,x)}\leq \abs{e^{-t\Delta+W}(x,x)}+ \abs{R(t)(x,x)}.
	\end{align*}
	Then, by the triangle inequality, the Kato-Trotter formula \eqref{eq:Kato-Trotter-gauss} and the bound \eqref{eq-demo:bounds-Lq-gaussian}, one has that $\normLp{\rho_{\gamma_t}}{r_j}{(\R^d)}=\OR(t^{-d/2})$. With this remark, for $r_k=+\infty$, we deduce that $\normLp{\rho_{\gamma_t}}{\infty}{(\R^d)}=\OR(t^{-d/2})$. Finally, after a last application of Lemma \ref{lemma:exp-kill-exchange_Rn} with $r=q=+\infty$, $p=1+d/2$ and $p=+\infty$, one deduces the estimates \eqref{eq-demo:exp-kill-exchange_Rn-Linfty} for any $n\geq 1$.\\
	For instance, one can start with $q_1=1$. By definition $\frac 1p+\frac 1{q_1}\geq\min(1,d/4)$, then one can fix $m>1/(\min(1,4/d)p-1)$ and define $r_1$ by the relation $\frac 1{q_1}-\frac 1{r_1}=\frac 1{m p}$ so that \eqref{cond:kill-exch_well-def-int} and \eqref{cond:kill-exch_small-err} hold. Note that $\frac 1{r_1}=1-\frac 1{mp}\geq 0$. Then, we iterate the procedure with $\frac 1{q_j}-\frac 1{r_j}=\frac 1{m p}$ for any integer $j\geq 1$ so that $\frac 1{r_j}\geq 0$ (in this case such that $j\leq\lfloor mp\rfloor+2$), with the condition \eqref{cond:exp-kill-exchange_nonstop}.
	By induction, 
	\[ \frac 1{q_j} = \frac 1{r_{j-1}} = \frac 1{q_{j-1}}-\frac 1{mp} =\ldots = \frac 1{q_1}-\frac{j-1}{mp}=1-\frac{j-1}{mp} .\]
	For instance, one can take $r_k=+\infty$ for $k=2+\lfloor mp (1-\min(1,4/d))\rfloor$.
	As desired, deduce that $\normLp{\rho_{\gamma_t}}{\infty}{}=\OR(t^{-d/2})$. 
\end{proof}

\section{Proof of the lower bound for the whole many-body system (Theorem \ref{thm:lower-bound-whole-mb})}\label{sec:lower-bound-whole-mb}  

The proof of Theorem \ref{thm:lower-bound-whole-mb} is actually an adaptation of the canonical lower bound \cite[Prop. 3.5]{fournais2018semi}. We begin by comparing the two proofs. Then, we detail the proof of Theorem \ref{thm:lower-bound-whole-mb}.

We introduce for our proof the reduced Hartree-Fock and the Thomas-Fermi energy functionals enriched in parameter $\lambda\geq 0$
	\begin{equation}\label{eq-def:rHF-energy_enriched}
	\ER_{\hbar,V,w,\lambda}^{\rm rHF}(\gamma):=\hbar^d\tr((-\hbar^2\Delta+V)\gamma)+\lambda\frac{\hbar^{2d}}{2} D_w(\rho_\gamma,\rho_\gamma),
	\end{equation}
and
\begin{equation}\label{eq-def:TF-energy_enriched}
	\ER_{V,w,\lambda}^{\rm TF}(\rho):=c_{\rm TF}\int_{\R^d}\rho(x)^{1+2/d}dx+\int_{\R^d}V(x)\rho(x)dx+\frac{\lambda}{2} D_w(\rho_\gamma,\rho_\gamma),
\end{equation}
respectively defined in $\XR$ and $\XR^{\rm TF}_V$. 

\subsection{Sketch of the proof of Fournais-Lewin-Solovej and comparison with ours}

Let us recall the main steps of the proof of \cite[Prop. 3.5]{fournais2018semi} in the canonical case and let us explain how they are adapted to our case.
\begin{enumerate}
	\item Let $w_1$ and $w_2$ the functions such that $w=w_1-w_2$ and $\hat{w}_1:= (\hat{w})_+$ and $\hat{w}_2:=(\hat{w})_-$. We begin by writting the inequality for any $\Psi_N\in L^2_a(\R^{dN})$
	\begin{equation}\label{eq-demo:lower-bound-whole-mb_step1}
	\prodscal{\Psi_N}{P_N\Psi_N}\geq \inf_{y_1,\ldots,y_L\in\R^d} \inf_{\underset{\normLp{\tilde{\Psi}}{2}{(\R^{dM})}=1}{\tilde{\Psi}\in L^2(\R^{dM})}}\prodscal{\tilde{\Psi}}{\tilde{P}_N\tilde{\Psi}},
	\end{equation}
	where, for any $M\in\llbracket 1,N\rrbracket$ (that we define it later), $L:=N-M$, $y_\ell:= x_{M+\ell}$ for all $\ell\in\llbracket 1,L\rrbracket$, and the operator on $L^2(\R^{dM})$ 
	\begin{equation*}
	\begin{split}
	\tilde{P}_N:=
	&
	\frac NM \sum_{m=1}^M\left(-\hbar^2\Delta_{x_m}+V(x_m)\right)
	+\hbar^d\frac{N(N-1)}{M(M-1)}\sum_{1\leq m\leq m'\leq M} w_1(x_m-x_{m'})
	\\&\quad
	+\hbar^d\frac{N(N-1)}{L(L-1)}\sum_{1\leq\ell\leq\ell'\leq L}w_2(y_\ell-y_{\ell'})
	-\hbar^d\frac{N(N-1)}{LM}\sum_{m=1}^M\sum_{\ell=1}^L w_2(x_m-y_\ell)
	,
	\end{split}
	\end{equation*}
	is associated to the fixed variables $y_1,\ldots,y_L$.
	The problem is reduced to a $M$-particle problem.
	\item Then, we bound from below the right-hand term of the previous inequality, uniformly in the variables $y_1,\ldots, y_L\in\R^d$, with a one-particle functional. For any $\tilde{\Psi}\in L^2_a(\R^{dM})$
	\begin{equation}\label{eq-demo:lower-bound-whole-mb_step2}
		\prodscal{\tilde{\Psi}}{\tilde{P}_N\tilde{\Psi}}\geq \frac NM \ER^{\rm rHF}_{\hbar,V-\tilde{E}_{\hbar,N},w,\frac{N-1}{M-1}}\big(\gamma_{\tilde{\Psi}}^{(1)}\big)
	,
	\end{equation}
	where $\tilde{E}_{\hbar,N}$ is the effective chemical potential
	\begin{equation}\label{eq-def:eff-pot_lower-bound-whole-mb}
	\tilde{E}_{\hbar,N}:=\frac{w_1(0)}2\frac{N-1}{M-1}\hbar^d+\frac{w_2(0)}2\frac{N-1}{L-1}\hbar^d.
	\end{equation}
	\item 
	Finally, we link the reduced Hartree-Fock ground state energy asymptotics to the Thomas-Fermi ground state energy.
\end{enumerate}

\begin{rmk}[What changes here]
	In our case, the limit $\hbar\to 0$ is not coupled to $N\to+\infty$. Indeed, we recall that instead of considering 
	\begin{equation*}
	\liminf_{\underset{\hbar N^{1/d}\to 1}{\hbar\to 0}}  e_{N,V,w} = 	\liminf_{\underset{\hbar N^{1/d}\to 1}{\hbar\to 0}} \inf_{\underset{ \normLp{\Psi_N}{2}{}=1}{\Psi_N\in L^2_a(\R^{dN})}} \frac{\prodscal{\Psi_N}{P_N\Psi_N}}{N},
	\end{equation*}
	we look at 
	\begin{equation*}
	\liminf_{\hbar\to 0} \hbar^d \inf_{N\geq 0}e_{\hbar,N,V,w}
	=\liminf_{\hbar\to 0} \hbar^d \inf_{N\geq 0}\inf_{\underset{ \normLp{\Psi_N}{2}{}=1}{\Psi_N\in L^2_a(\R^{dN})}}\prodscal{\Psi_N}{P_N\Psi_N}.
	\end{equation*}
	Thus, it is more convenient to take $M:=\lfloor(1-\varepsilon)N\rfloor$ with a well-chosen parameter $\varepsilon\in(0,1)$ that we let go to $0$, instead of taking $M:=N-\lfloor\sqrt{N}\rfloor$.
	Indeed, in the proof of \cite{fournais2018semi}, the reduced Hartree-Fock ground state energy $e^{\rm rHF}_{\hbar,N,V-\tilde{E}_{\hbar,N},V,w,\frac{N-1}{M-1}}$ has a finite mean-field
	limit ($N\to +\infty$, $\hbar\to 0$, with $\hbar=N^{-1/d}$), which is $e^{\rm TF}_{V,w,1}$.
	Moreover, that is always true that the term $\ER^{\rm rHF}_{\hbar,V-\tilde{E}_{\hbar,N},w,\frac{N-1}{M-1}}\big(\gamma_{\tilde{\Psi}}^{(1)}\big)$ of \eqref{eq-demo:lower-bound-whole-mb_step2} is uniformly bounded by the ground state energy $e^{\rm rHF}_{\hbar,N,V-\tilde{E}_{\hbar,N},w,\frac{N-1}{M-1}}$. This proves the lower bounds.
	
	Actually, the highest order term of this reduced Hartree-Fock energy in \eqref{eq-def:eff-pot_lower-bound-whole-mb} is carried by $-\frac{w_2(0)}{2}\frac{N-1}{L-1}\hbar^d\sim -\sqrt{N}\hbar^d$. 
	This convergence was possible for the choice $M=N-\lfloor\sqrt{N}\rfloor$ with the relation $\hbar^d=1/N$, so that $\frac{w_2(0)}{2}\frac{N-1}{L-1}\hbar^d\sim \sqrt{N}\hbar^d\to 0$ at the coupled mean-field and semiclassical limit. But it does not work with this choice of $M$ in our grand-canonical setting, because this term diverges to $-\infty$ as $n\to+\infty$ for any fixed $\hbar>0$, so that
	\begin{equation*}
	\inf_{N\in\N^*}e^{\rm rHF}_{\hbar,V,\tilde{E}_{\hbar,N},\frac{N-1}{\lfloor N\rfloor-1}}=-\infty.
	\end{equation*}
	Therefore the proof slightly changes with the new expression of $M$ in Step 3.
\end{rmk}

\subsection{Proof of Theorem \ref{thm:lower-bound-whole-mb}}

\subsubsection*{Steps 1 and 2.} The two first steps are the same as in Fournais-Lewin-Solovej paper (see the beginning of \cite[Prop. 3.5]{fournais2018semi} and \cite[Lem. 3.6]{fournais2018semi}). We do not yet use the value of $M$ and we write the lower bounds for any $N\in\N^*$. We will see that these bounds are uniform with respect to $N$.

\subsubsection*{Step 3.} Recall that we have
\begin{align*}
e_{\hbar,V,w}
&\geq \inf_{N\in\N^*}e_{\hbar,N,V,w}
\\&\geq \hbar^d \inf_{N\in\N^*} \inf_{y_1,\cdots,y_L\in\R^d} \: \inf_{\underset{\prodscal{\tilde{\Psi}}{\tilde{P}_N\tilde{\Psi}}\leq -C_0\hbar^{-d} }{\tilde{\Psi}\in L^2_a(\R^{dM}),\:\normLp{\tilde{\Psi}}{2}{}=1 }}\prodscal{\tilde{\Psi}}{\tilde{P}_N\tilde{\Psi}}.
\end{align*}
Instead of dealing directly with ground states energies and therefore infimum, it is more conveniant to write bounds on the reduced Hartree-Fock functional.
Hence, by \eqref{eq-demo:lower-bound-whole-mb_step2}, for $M=\lfloor(1-\varepsilon)N\rfloor$ with $\varepsilon\in(0,1)$ to be precised and any normalized $\tilde{\Psi}\in L^2_a(\R^{dM})$ (recall that its associated one-body density matrix satisfies always $0\leq\gamma_{\tilde{\Psi}}^{(1)}\leq 1$)
\begin{align*} 
\hbar^d \prodscal{\tilde{\Psi}}{\tilde{P}_N\tilde{\Psi}} &\geq \ER_{\hbar,V-\tilde{E}_{\hbar,N},w,\frac{N-1}{M-1}}^{\rm rHF}\left(\gamma_{\tilde{\Psi}}^{(1)}\right)
\geq \ER^{\rm rHF}_{\hbar,V,w,1}\left(\gamma_{\tilde{\Psi}}^{(1)}\right)+r_{\hbar,N,\varepsilon}\left(\gamma_{\tilde{\Psi}}^{(1)}\right)
\\&\geq e^{\rm rHF}_{\hbar,V,w,1}+ r_{\hbar,N,\varepsilon}\left(\gamma_{\tilde{\Psi}}^{(1)}\right)
,
\end{align*}
where $r_{\hbar,N,\varepsilon}$ denotes the error defined in the set $\KR$
\begin{align*}
	r_{\hbar,N,\varepsilon}(\gamma) := \ER^{\rm rHF}_{\hbar,V-\tilde{E}_{\hbar,N},w,\frac{N-1}{M-1}}(\gamma)- \ER^{\rm rHF}_{\hbar,V,w,1}(\gamma)
	.
\end{align*}
Recall that
\begin{equation*}
	\lim_{\hbar\to 0} e^{\rm rHF}_{\hbar,V,w,1}=e^{\rm TF}_{V,w,1}.
\end{equation*}
Thus, it remains to check that for $\varepsilon$ small enough, we have the lower bound of the liminf: there exists $C\geq0$ such that for any $\varepsilon>0$
\begin{equation}\label{eq-demo:thm:lower-bound-whole-mb_liminf-error}
	\liminf_{\hbar\to 0}\inf_{N\geq\frac 2\varepsilon }\inf_{\underset{\prodscal{\tilde{\Psi}}{\tilde{P}_N\tilde{\Psi}}\leq -C_0\hbar^{-d} }{\tilde{\Psi}\in L^2(\R^{dM}),\:\normLp{\tilde{\Psi}}{2}{}=1 }}r_{\hbar,N,\varepsilon}\left(\gamma_{\tilde{\Psi}}^{(1)}\right)\geq -C\varepsilon,
\end{equation}
for any almost-minizers $\tilde{\Psi}\in L^2_a(\R^{dM})$ of the ground state energy of $\tilde{P}_N$. The conclusion stems from it.

\bigskip
We begin by discussing conditions on $N$, $\varepsilon$ and $h$ so that $M=\lfloor(1-\varepsilon)N\rfloor, \: L=N-M\in\llbracket 2, N-2\rrbracket$ and the avoid cases like $L=0$. For instance, for $\varepsilon>0$, the condition
\begin{equation}\label{eq-demo:thm:lower-bound-whole-mb_cond-N}
N\geq \frac 2\varepsilon
\end{equation}
ensures that $M\leq N-2$.
Furthermore, we impose also that $\varepsilon>0$ is small enough. For instance $\varepsilon\in(0,1/4]$ ensures
that $M\geq 2$ under the condition \eqref{eq-demo:thm:lower-bound-whole-mb_cond-N}.

Let us now explain why for $\hbar>0$ small enough, we can just consider \eqref{eq-demo:thm:lower-bound-whole-mb_cond-N}, i.e.\ 
	for any $\varepsilon\in(0,1/4)$, there exists $h_\varepsilon\in(0,1)$ so that for any $h\in(0,\hbar_\varepsilon)$
\begin{equation*}
	e_{\hbar,V,w}=\inf_{N\geq\frac 2\varepsilon}e_{\hbar,N,V,w}.
\end{equation*}
In other words, it means that the ground state energy cannot be reached for $N<2/\varepsilon$.
First,
since the ground state energy is bounded by the Hartree-Fock energy
\begin{align*}
	e_{\hbar,N,V,w}\leq e^{\rm HF}_{\hbar,N,V,w}\leq \ER^{\rm HF}_{\hbar,N,V,w}(\indicatrice{-\hbar^2\Delta+V\leq 0})\leq \hbar^d\tr_{L^2(\R^d)}((h^2\Delta+V)_-)\leq 0,
\end{align*}
there exists $C_0>0$ such that for any $h\in(0,1)$
\begin{equation}\label{eq-demo:thm:lower-bound-whole-mb_upper-bound}
\inf_{N\geq 0}e_{\hbar,N,V,w}(E)\leq -C_0.
\end{equation}
Let $\Psi_N\in L^2_a(\R^{dN})$ be a normalized state. Then, since $-\hbar^2\Delta+V$ is non-negative
\begin{align*}
\prodscal{\Psi_N}{P_N\Psi_N}
&\geq  -E\tr_{L^2(\R^d)}\left(\gamma_{\Psi_N}^{(1)}\right) +\hbar^d\tr_{L^2(\R^{2d})}\left(w\gamma_{\Psi_N}^{(2)}\right)
\\&\geq -EN -\hbar^d \normLp{w}{\infty}{(\R^d)}\tr_{L^2(\R^{2d})}\left(\gamma_{\Psi_N}^{(2)}\right)
=
-EN -\hbar^d \frac{N(N-1)}2\normLp{w}{\infty}{(\R^d)}
.
\end{align*}
Under the assumption $N<2/\varepsilon$, there exists $C>0$ such that
\begin{equation*}
\prodscal{\Psi_N}{P_N\Psi_N} \geq -\frac 2\varepsilon E -\hbar^d\frac 2{\varepsilon^2}\normLp{w}{\infty}{(\R^d)}
\leq -\frac{C}{\varepsilon^2}
.
\end{equation*}
Let $\hbar_\varepsilon>0$ such that for any $\hbar\in(0,\hbar_\varepsilon)$, we have $C\varepsilon^{-2}\leq C_0 \hbar^{-d}$.
Finally, for any $\hbar\in(0,\hbar_\varepsilon)$ and any $N<2/\varepsilon$
\begin{equation*}
e_{\hbar,N,V,w}(E)\geq -C_0 =\inf_{N\geq 0} e_{\hbar,N,V,w}(E).
\end{equation*}
%
\bigskip

Now, let us bound by below the error term $r_{\hbar,N,\varepsilon}(\gamma)$ uniformy in $N\geq 2/\varepsilon$ with respect to $h\in(0,\hbar_\varepsilon)$ and $\varepsilon\in(0,1/4)$.  We prove that there exist $C,C'>0$ such that for any bounded operator $\gamma\geq 0$ on $L^2(\R^d)$
\begin{equation*}
r_{\hbar,N,\varepsilon}(\gamma)\geq -\frac C\varepsilon \hbar^d\tr(\gamma)
-
\begin{cases}
0 &\text{under Assumption \ref{cond:w-Dterm}}, \\
C'\varepsilon\tr((1-\hbar^2\Delta)\gamma) &\text{under Assumption \ref{cond:w-Dterm-d=12}}.
\end{cases}
\end{equation*}
By definition,
\begin{align*}
r_{\hbar,N,\varepsilon}(\gamma)
&=
(E-\tilde{E}_{\hbar,N})\tr(\gamma)+\frac{\hbar^d}2\left(\frac{N-1}{M-1}-1\right)D_w(\rho_\gamma,\rho_\gamma)
\\&= -\frac{\hbar^d}2 \left(w_1(0)\frac{N-1}{M-1}+w_2(0)\frac{N-1}{L-1}\right)\tr(\gamma)+\frac{\hbar^d}2\left(\frac{N-1}{M-1}-1\right)D_w(\rho_\gamma,\rho_\gamma)
.
\end{align*}
On one hand, we deduce from the definition of $w_1$ and $w_2$ that
\begin{align*}
0 \leq w_1(0) &= \frac{1}{(2\pi)^{d/2}}\normLp{(\hat{w})_+}{1}{(\R^d)} \leq   \frac{1}{(2\pi)^{d/2}}\normLp{\hat{w}}{1}{(\R^d)}
,\\
0 \leq w_2(0) &= \frac{1}{(2\pi)^{d/2}}\normLp{(\hat{w})_-}{1}{(\R^d)} \leq   \frac{1}{(2\pi)^{d/2}}\normLp{\hat{w}}{1}{(\R^d)}
.
\end{align*}
On the other hand, since $N\geq 2/\varepsilon$, we have
\begin{equation*}
	\frac{N-1}{M-1}\leq\frac 1{1-2\varepsilon},\qquad \frac{N-1}{L-1}\leq \frac 2\varepsilon
	.
\end{equation*}
Furthermore, using also that $\varepsilon\leq 1/4$
\begin{equation*}
0\leq \frac{N-1}{M-1}-1= \frac{L}{M-1}\leq 3\varepsilon
.
\end{equation*}
Thus, there exist $C,C'>0$ such that for any $N\geq 2/\varepsilon$
\begin{equation*}
r_{\hbar,N,\varepsilon}(\gamma)\geq -\frac C\varepsilon \hbar^d\tr(\gamma)-C'\varepsilon \hbar^d D_w(\rho_\gamma,\rho_\gamma).
\end{equation*}
\begin{itemize}
	\item 
	If Assumption \ref{cond:w-Dterm} holds, we can just write the lower bound without the direct term, for any $N\geq 2/\varepsilon$ and any $\gamma\geq 0$
	\begin{equation*}
	r_{\hbar,N,\varepsilon}(\gamma)\geq -\frac C\varepsilon \hbar^d\tr(\gamma).
	\end{equation*}
	\item
	If we have Assumption \ref{cond:w-Dterm-d=12}, one has
	\begin{equation*}
	h^d D_w(\rho_\gamma,\rho_\gamma)\lesssim \tr((1-\hbar^2\Delta)\gamma).
	\end{equation*}
	As a consequence,  for any $N\geq 2/\varepsilon$ and any $\gamma\geq 0$
	\begin{equation*}
	r_{\hbar,N,\varepsilon}(\gamma)\geq -\frac C\varepsilon \hbar^d\tr(\gamma)-C'\varepsilon\tr((1-\hbar^2\Delta)\gamma).
	\end{equation*}
\end{itemize}
Moreover, by the bound \eqref{eq-demo:thm:lower-bound-whole-mb_upper-bound}, there exist $N\geq 2/\varepsilon$ and normalized $\Psi_N\in L^2_a(\R^d)$ such that for any $h\in(0,\hbar_\varepsilon)$, the upper bound
$\prodscal{\Psi_N}{P_N\Psi_N}\leq -C_0\hbar^{-d}$ holds. As well, we restrict ourselves to normalized states $\tilde{\Psi}\in L^2_a(\R^{dM})$ (and $x_1,\ldots,x_M\in\R^d$) such that the bound $\prodscal{\tilde   {\Psi}}{\tilde{P}_N\tilde{\Psi}}\leq -C_0\hbar^{-d}$ holds.
Hence, for all these states $\tilde{\Psi}\in L^2_a(\R^{dM})$, one has
\begin{align*}
C_0\hbar^{-d}
&\geq \ER^{\rm rHF}_{\hbar,V,w,1}\big(\gamma_{\tilde{\Psi}}^{(1)}\big)
-\frac C\varepsilon \hbar^d\tr\big(\gamma_{\tilde{\Psi}}^{(1)}\big)
-
\begin{cases}
0 &\text{under Assumption \ref{cond:w-Dterm}}, \\
C'\varepsilon\tr\left((1-\hbar^2\Delta)\gamma_{\tilde{\Psi}}^{(1)}\right) &\text{under Assumption \ref{cond:w-Dterm-d=12}}.
\end{cases}
\end{align*}	
Let us explain why the one-body density matrices $\gamma_{\tilde{\Psi}}^{(1)}$ of ground states of $\tilde{P}_N$ satisfy Assumption \ref{cond:trace_bd_h^d}.
By Lemma \ref{lemma:E_HF_bound_below}, there exists $C''>0$ such that
\begin{equation*}
\ER^{\rm rHF}_{\hbar,V,w,1}\big(\gamma_{\tilde{\Psi}}^{(1)}\big)\geq \frac 14\tr\big((-\hbar^2\Delta+V+1)\gamma_{\tilde{\Psi}}^{(1)}\big)-C''\hbar^{-d}.
\end{equation*}
We then deduce that there exist $C_0'>0$ and $h_\varepsilon'>0$ (for instance $h_\varepsilon'=\min(h_\varepsilon,(16 C\varepsilon)^{-1/d})$) such that for any $\varepsilon\in\left(0,\min(\tfrac 14,\tfrac 1{8C'})\right)$ and for any $h\in(0,\hbar_\varepsilon']$
\begin{equation*}
\tr\left((-\hbar^2\Delta+V+1)\gamma_{\tilde{\Psi}}^{(1)}\right) \leq C_0'\hbar^{-d}.
\end{equation*}
As a consequence, there exists $C>0$ such that for $N\geq 2/\varepsilon$ such that for any $h\in(0,\hbar_\varepsilon)$ and any normalized $\tilde{\Psi}\in L^2_a(\R^{dM})$ such that the upper bound on the energy $\prodscal{\tilde{\Psi}}{\tilde{P}_N\tilde{\Psi}}$ is satisfied
\begin{equation*}
r_{\hbar,N,\varepsilon}\left(\gamma_{\tilde{\Psi}}^{{(1)}}\right)\geq -C\hbar^{-d}\varepsilon.
\end{equation*}
Thus, we get \eqref{eq-demo:thm:lower-bound-whole-mb_liminf-error}, which allows us to conclude the proof of Theorem \ref{thm:lower-bound-whole-mb}.

\section*{\bf Acknowledgments}
The author is grateful to her PhD advisor Julien Sabin for valuable discussions and to the Laboratoire de Mathématiques d'Orsay where most of the research has been done. This project has also been partially supported by the European Research Council (ERC) through the Starting Grant {\sc FermiMath}, grant agreement nr. 101040991.

\bibliographystyle{siam}
\bibliography{biblioNhi2.bib}

\end{document}